\PassOptionsToPackage{dvipsnames}{xcolor} 
\documentclass[11pt,letter]{article}
\usepackage{amssymb,amsmath,amsthm,multirow,color,xcolor,cancel,bbm,tikz,mathrsfs,hyperref,authblk,stmaryrd}
\usepackage[textheight=24cm,textwidth=16cm]{geometry}
\tikzstyle{main node}=[draw,circle,inner sep=1,outer sep=2,thick,minimum size=12pt]
\usetikzlibrary{decorations.pathreplacing}

\hypersetup{
    colorlinks=true,
    linkcolor=blue,
    citecolor=cyan,
    filecolor=magenta,      
    urlcolor=green,
    pdftitle={Overleaf Example},
    pdfpagemode=FullScreen,
    }

\newtheorem{lemma}{Lemma}
\newtheorem{theorem}{Theorem}

\newtheorem{remark}{Remark}

\newcommand{\EM}[1]{{\it\textcolor{Maroon}{#1}}}

\newcommand{\B}{\{0,1\}}
\newcommand{\ONE}{\mathbf{1}}
\newcommand{\ZERO}{\mathbf{0}}
\newcommand{\cst}{\mathrm{cst}}
\newcommand{\id}{\mathrm{id}}

\def\S{\mathcal{S}}
\def\A{\mathcal{A}}
\def\P{\mathcal{P}}
\def\S{\mathcal{S}}
\def\K{\mathcal{K}}
\def\L{\mathcal{L}}
\def\calB{\mathcal{B}}
\def\FP{\mathrm{FP}}
\def\fp{\mathrm{fp}}
\def\Pr{\mathrm{Pr}}

\title{Asynchronous dynamics of isomorphic Boolean networks}

\author{
Florian Bridoux\footnote{ Universit\'e Côte d’Azur, CNRS, I3S, Sophia Antipolis, France. ({\tt  bridoux@i3s.unice.fr})},
Aymeric Picard Marchetto\footnote{Universit\'e Côte d’Azur, CNRS, I3S, Sophia Antipolis, France. ({\tt  picard@i3s.unice.fr})},
Adrien Richard\footnote{Universit\'e Côte d’Azur, CNRS, I3S, Sophia Antipolis, France. ({\tt  adrien.richard@cnrs.fr})},
}

\begin{document}
%%%%%%%%%%%%%%%%%%%%%%%%%%%%%%%%%%%%%%%%%%%%%%%%%%%%%%%%%%%%%%%%%%%%%%%%%%%%%%%%%%%%%%%%%%%%%%%%%%%%%%%
%%%%%%%%%%%%%%%%%%%%%%%%%%%%%%%%%%%%%%%%%%%%%%%%%%%%%%%%%%%%%%%%%%%%%%%%%%%%%%%%%%%%%%%%%%%%%%%%%%%%%%%

\maketitle

\begin{abstract}
A Boolean network is a function $f:\B^n\to\B^n$ from which several dynamics can be derived, depending on the context. The most classical ones are the synchronous and asynchronous dynamics. Both are digraphs on $\B^n$, but the synchronous dynamics $\S(f)$ has an arc from $x$ to $f(x)$ while the asynchronous dynamics $\A(f)$ has an arc from $x$ to $x+e_i$ whenever $x_i\neq f_i(x)$. Clearly, $\S(f)$ and $\A(f)$ share the same information, but what can be said about these objects up to isomorphism? We prove that if $\A(f)$ is only known up to isomorphism then, with high probability, $\S(f)$ can be fully reconstructed up to isomorphism. We then show that the converse direction is far from being true. In particular, if $\S(f)$ is only known up to isomorphism, very little can be said on the attractors of $\A(f)$. For instance, if $f$ has $p$ fixed points, then $\A(f)$ has at least $\max(1,p)$ attractors, and we prove that this trivial lower bound is tight: there always exists $h$ with $\S(h)\sim \S(f)$ such that $\A(h)$ has exactly $\max(1,p)$ attractors. But $\A(f)$ may often have many more attractors since we prove that, with high probability, there exists $h$ with $\S(h)\sim\S(f)$ such that $\A(h)$ has $\Omega(2^n)$ attractors.
\end{abstract}

%%%%%%%%%%%%%%%%%%%%%%%%%%%%%%%%%%%%%%%%%%%%%%%%%%%%%%%%%%%%%%%%%%%
\section{Introduction}
%%%%%%%%%%%%%%%%%%%%%%%%%%%%%%%%%%%%%%%%%%%%%%%%%%%%%%%%%%%%%%%%%%%

A \EM{Boolean network} with $n$ components is a function 
\[
f:\B^n\to \B^n,\quad x=(x_1,\dots,x_n)\mapsto f(x)=(f_1(x),\dots,f_n(x)).
\]
Many dynamics can be derived from $f$, the most two common ones being the \EM{synchronous} and the \EM{asynchronous} dynamics, denoted \EM{$\S(f)$} and \EM{$\A(f)$}, respectively. Both are digraphs with vertex set $\B^n$ but they differ in their arcs: $\S(f)$ has an arc from $x$ to $f(x)$ for every $x\in \B^n$ (and is usually identified with $f$),  while $\A(f)$ has an arc from $x$ to $y$ whenever $x$ and $y$ only differ in component $i$ and $f_i(x) \neq x_i$. See Figure \ref{fig:S_and_A} for an illustration. 

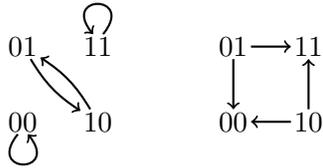
\begin{figure}[t]
\[
\begin{array}{ccc}
%Synchronous
\begin{array}{c}
\begin{tikzpicture}
\node[inner sep=1] (00) at (0,0){$00$};
\node[inner sep=1] (01) at (0,1){$01$};
\node[inner sep=1] (10) at (1,0){$10$};
\node[inner sep=1] (11) at (1,1){$11$};
\path[thick,->]
(01) edge[bend right=13] (10)
(10) edge[bend right=13] (01)
;
%\draw[->,thick] (00.-112) .. controls ({0-0.5},{0-0.7}) and ({{0+0.5}},{0-0.7}) .. (00.-68);
%\draw[->,thick] (11.+68) .. controls ({1+0.5},{1+0.7}) and ({{1-0.5}},{1+0.7}) .. (11.+112);
\end{tikzpicture}
\end{array}
&&
%Asynchronous
\begin{array}{c}
\begin{tikzpicture}
\node[inner sep=1] (00) at (0,0){$00$};
\node[inner sep=1] (01) at (0,1){$01$};
\node[inner sep=1] (10) at (1,0){$10$};
\node[inner sep=1] (11) at (1,1){$11$};
\path[thick,->]
(01) edge (00)
(01) edge (11)
(10) edge (00)
(10) edge (11)
;
\end{tikzpicture}
\end{array}
\end{array}
\]
{\caption{\label{fig:S_and_A} Synchronous (left) and asynchronous (right) dynamics of the shift $f(x_1,x_2)=(x_2,x_1)$. 
}}
\end{figure}

\medskip
The choice of dynamics depends on the context, and these are multiple : Boolean networks have many real-life applications, including gene networks \cite{J02,K69,T73,TK01}, neural networks \cite{G85,H82,MP43}, reaction systems \cite{ehrenfeucht2007reaction}, social interactions \cite{GT83,PS83}, argumentation frameworks \cite{azpeitia2024bridging} and more \cite{GM90,TA90}. Notably, Boolean networks were introduced as models for the dynamics of gene networks in the seminal works of Kaufman \cite{K69} and Thomas \cite{T73}; but Kaufman proposed to use the synchronous dynamics while Thomas argued that the asynchronous dynamics is more realistic. Since then, both dynamics are used to model the same object, which is disturbing since the two dynamics often describe radically different behaviours. It then seems natural to compare these dynamics, see e.g.  \cite{GDX08,GS08,GS10,GN12,NS17,R95}. 

\medskip
Many important dynamical parameters are invariant by isomorphism such as the number of and size of attractors\footnote{The \EM{attractors} of a digraph are its terminal strong components; see Section \ref{sec:def} for a formal definition.}, the transient length and so on. In other words, one often consider that the important part of the behaviors described by $\S(f)$ or $\A(f)$ is contained in the structure of the state transitions rather than in state labels. This leads us to study the relationships between $\S(f)$ and $\A(f)$ up to isomorphism; the isomorphic relation is denoted \EM{$\sim$}, and \EM{$h\sim f$} means $\S(f)\sim\S(h)$. In other words, given two Boolean networks $f,h$, we are interested in the following two questions:
\begin{itemize}
\item 
What can be said about $\S(h)$ when $\A(h)\sim\A(f)$?
\item 
What can be said about $\A(h)$ when $\S(h)\sim \S(f)$?
\end{itemize}

\medskip
Our first result shows that $\A(h)\sim \A(f)$ almost always implies $\S(h)\sim \S(f)$: the unlabelled synchronous dynamics can almost always be fully reconstructed from the unlabelled asynchronous dynamics. 

\begin{theorem}\label{thm:A_to_S}
Taking $f$ uniformly at random, the probability that $\A(h)\sim \A(f)$ implies $h\sim f$ for every $h$ tends to $1$ as $n\to\infty$.
\end{theorem}

Our second result shows that, in the other direction, the unlabelled asynchronous dynamics cannot be fully reconstructed from the unlabelled synchronous dynamics, unless the latter corresponds to a constant or the identity. 

\begin{theorem}\label{thm:S_to_A}
If $f\neq\cst,\id$ and $n\geq 3$,  there exists $h\sim f$ such that $\A(h)\not\sim \A(f)$.
\end{theorem}

However, we may ask if at least some important parameters of the unlabelled asynchronous dynamics can be reconstructed from the unlabelled synchronous dynamics. Arguably, one of the most important parameters is the number of attractors in $\A(f)$ \cite{HMM13,R09,RT23,ZYL13}, and we thus focus on this parameter. A basic observation is that $f(x)=x$ if and only if $x$ is of out-degree $0$ in $\A(f)$. In other words, fixed points of $f$ correspond to attractors of $\A(f)$ of size one. Since $\A(f)$ has always at least one attractor, we deduce that 
\begin{equation}\label{eq:trivial_bound}
\textrm{$\A(f)$ has at least $\max(1,\fp(f))$ attractors,}
\end{equation}
where $\fp(f)$ is the number of fixed points of $f$. This is thus a lower bound on the number of asynchronous attractors that only depends on the unlabelled synchronous dynamics. But can we say something stronger? We provide a negative answer, showing that this trivial lower bound is always tight. 

\begin{theorem}\label{thm:small_att}
If $f$ has a fixed point, there exists $h\sim f$ such that $\A(h)$ has exactly $\fp(f)$ attractors. Otherwise, there exists $h\sim f$ such that $\A(h)$ has a unique attractor, which is of size $\leq 4$. 
\end{theorem}

Even if we cannot extract a non-trivial lower bound on the number of asynchronous attractors from the unlabelled synchronous dynamics, we may ask if a non-trivial upper bound could be obtained. However, we will prove that, most often, any such upper-bound is in $\Omega(2^n)$. Indeed, we will prove that if $f^2$ has $d$ images, then there exists $h\sim f$ such that $\A(h)$ has at least $\lfloor d/10\rfloor$ attractors, each of size $\leq 4$ (Lemma \ref{lem:many_att}). Since the expected value of $d$ is asymptotically $(1-e^{-1+e^{-1}})2^n\simeq 0.468\cdot 2^n$ \cite{FO89}, and since the limit distribution of $d$ is a Gaussian distribution with variance $c\cdot 2^n$ for some constant $c>0$ \cite{DS97}, we obtain the following.  

\begin{theorem}\label{thm:many_att}
Taking $f$ uniformly at random, the probability that there exists $h\sim f$ such that $\A(h)$ has at least $0.046\cdot 2^n$ attractors, each of size $\leq 4$, tends to $1$ as $n\to\infty$.
\end{theorem}

Concerning attractor sizes, the upper bound $4$ in Theorems \ref{thm:small_att} and \ref{thm:many_att} is tight, in that if $f$ is a permutation without limit cycle of length $\leq 2$, then the attractors of $\A(f)$ are all of size $\geq 4$ (Remark~\ref{rem:4}). In the opposite direction, if $f$ is a permutation without fixed point (derangement), then $\A(f)$ may contain a unique attractor spanning the $2^n$ configurations. 

\begin{theorem}\label{thm:strong}
If $f$ is a derangement, there exists $h\sim f$ such that $\A(h)$ is strongly connected.
\end{theorem}

All these results indicate that, from the unlabelled synchronous dynamics, almost nothing can be said on the number and the size of the asynchronous attractors. 

\medskip
The paper is organized as follows. The main definitions are introduced in Section \ref{sec:def}. Then, Theorems \ref{thm:A_to_S}, \ref{thm:S_to_A}, \ref{thm:small_att}, \ref{thm:many_att} and \ref{thm:strong} are proved in Sections \ref{sec:A_to_S}, \ref{sec:S_to_A}, \ref{sec:small_att}, \ref{sec:many_att} and \ref{sec:strong} respectively. The most technical proofs are those of Theorems  \ref{thm:A_to_S} and \ref{thm:small_att}. A conclusion and some perspectives are given in Section~\ref{sec:conclusion}. 

%%%%%%%%%%%%%%%%%%%%%%%%%%%%%%%%%%%%%%%%%%%%%%%%%%%%%%%%%%%%%%%%%%%
\section{Definitions}\label{sec:def}
%%%%%%%%%%%%%%%%%%%%%%%%%%%%%%%%%%%%%%%%%%%%%%%%%%%%%%%%%%%%%%%%%%%

In the following, $n$ is always a positive integer, and  $\EM{[n]}=\{1,\dots,n\}$. Elements of $\B^n$ are called \EM{configurations}, and those of $[n]$ are called \EM{components}. Let $x,y\in \B^n$. The addition $x+y$ is computed component-wise modulo~$2$. We denote by \EM{$\ZERO$} (resp. \EM{$\ONE$}) the configuration $x$ such that $x_i = 0$ (resp. $1$) for all $i \in [n]$. We set $\overline{x_i}=x_i+1$ and $\EM{\overline{x}}=x+\ONE$. For $i \in [n]$ we denote by \EM{$e_i$} the configuration $z$ such that $z_i=1$ and $z_j=0$ for all $j\neq i$. More generally, given distinct $i_1,\dots,i_k\in [n]$, we set $\EM{e_{i_1,\dots,i_k}}=e_{i_1}+\cdots+e_{i_k}$. The \EM{weight} \EM{$w(x)$} of $x$ is the number of $i\in [n]$ with $x_i = 1$. The \EM{distance} \EM{$d(x,y)$} between $x$ and $y$ is the number of $i\in [n]$ such that $x_i\neq y_i$. Hence, $d(x,y)=w(x+y)$. The partial order \EM{$\leq$} on $\B^n$ is defined by  $x \leq y$ if and only if $x_i \leq y_i$ for all $i \in [n]$. Given $X\subseteq \B^n$, we set $\EM{\overline{X}}=\B^n\setminus X$. Given $a,b\in \B^n$, we denote by \EM{$(a\leftrightarrow b)$} be the permutation of $\B^n$ that transposes $a$ and $b$, that is, $(a\leftrightarrow b)(a)=b$, $(a\leftrightarrow b)(b)=a$ and $(a\leftrightarrow b)(x)=x$ for all $x\neq a,b$.

\medskip
Given a graph (resp. digraph) $G$, we denote by $\EM{V(G)}$ its vertex set and $\EM{E(G)}$ its edges (resp. arcs). If $G$ is undirected and has an edge between $u$ and $v$ we write $uv\in E(G)$ or $vu\in E(G)$ indifferently. If $G$ is directed and has an edge from $u$ to $v$ we write $uv\in E(G)$ or $u\to v\in E(G)$. Given two digraphs $G,V$, we write $G\,\EM{\sim}\, H$ to mean that that $G$ and $H$ are isomorphic, that is, there exits a bijection $\pi:V(G)\to V(H)$ such that, for all $u,v\in V(G)$, we have $uv\in E(G)$ if and only if $\pi(u)\pi(v)\in E(H)$; then $\pi$ is an isomorphism from $G$ to $H$. We denote by $\EM{C_\ell}$ the (unlabelled) directed cycle of length $\ell$, and $\EM{P_{\ell}}$ the (unlabelled) directed path of length $\ell$ (with $\ell+1$ vertices). Given an (unlabelled) digraph $G$ and $k\geq 1$, we denote by $\EM{kG}$ the (unlabelled) digraph which consists of $k$ vertex-disjoint copies of $G$. Given (unlabelled) digraphs $G$ and $H$, the disjoint union of $G$ and $H$ is denoted \EM{$G+H$}. Given $U\subseteq V(G)$, we denote by $\EM{G[U]}$ the subgraph of $G$ induced by $U$. An \EM{attractor} of $G$ is the vertex set of a terminal strongly connected component; equivalently, this is an inclusion-minimal non-empty set $X\subseteq V(G)$ such that there is no arc from $X$ to $V(G)\setminus X$

\medskip
We denote by \EM{$F(n)$} the set of functions $f:\B^n\to\B^n$. Let $f\in F(n)$. The \EM{synchronous graph} of $f$ is the directed graph $\S(f)$ on $\B^n$ with an arc from $x$ to $f(x)$ for all $x\in\B^n$. A \EM{limit cycle} of $f$ is a cycle of $\S(f)$. A \EM{periodic configuration} of $f$ of period $p$ is a configuration that belongs to a cycle of $\S(f)$ of length $p$. A \EM{fixed point} of $f$ is a configuration $x$ such that $f(x)=x$. Equivalently, it is a periodic configuration of period $1$. We denote by $\EM{\FP(f)}$ the set of fixed points of $f$ and $\EM{\fp(f)}=|\FP(f)|$. The \EM{asynchronous graph} of $f\in F(n)$ is the digraph $\A(f)$ on $\B^n$ with an arc from $x$ to $x+e_i$ for all $x\in\B^n$ and  $i\in [n]$ with $f_i(x)\neq x_i$. Given two functions $f,h\in F(n)$, we write \EM{$f\sim h$} to means that $f$ and $h$ are isomorphic, that is, $\S(f)\sim \S(h)$.

%%%%%%%%%%%%%%%%%%%%%%%%%%%%%%%%%%%%%%%%%%%%%%%%%%%%%%%%%%%%%%%%%%%
\section{Synchronous reconstruction}\label{sec:A_to_S}
%%%%%%%%%%%%%%%%%%%%%%%%%%%%%%%%%%%%%%%%%%%%%%%%%%%%%%%%%%%%%%%%%%%

In this section, we prove Theorem~\ref{thm:A_to_S}, that for almost all Boolean networks $f$, if $\A(h)\sim \A(f)$ then $h\sim f$: the unlabelled synchronous dynamics can be fully reconstructed from the unlabelled asynchronous dynamics.

\medskip
First, remark that this is not true for \emph{all} Boolean networks as shown in Figure \ref{fig:counter_A_to_S}. In this example, where $n=3$, the asynchronous dynamics only contains $3$ arcs among the $n2^n=24$ possible arcs, and is thus very sparse. Usual asynchronous dynamics are not sparse: they contain $n2^{n-1}$ arcs in average. We will prove that isomorphisms between such usual asynchronous dynamics preserve the distance between configurations, which allows the reconstruction of the synchronous dynamics without possible ambiguity (up to isomorphism).  

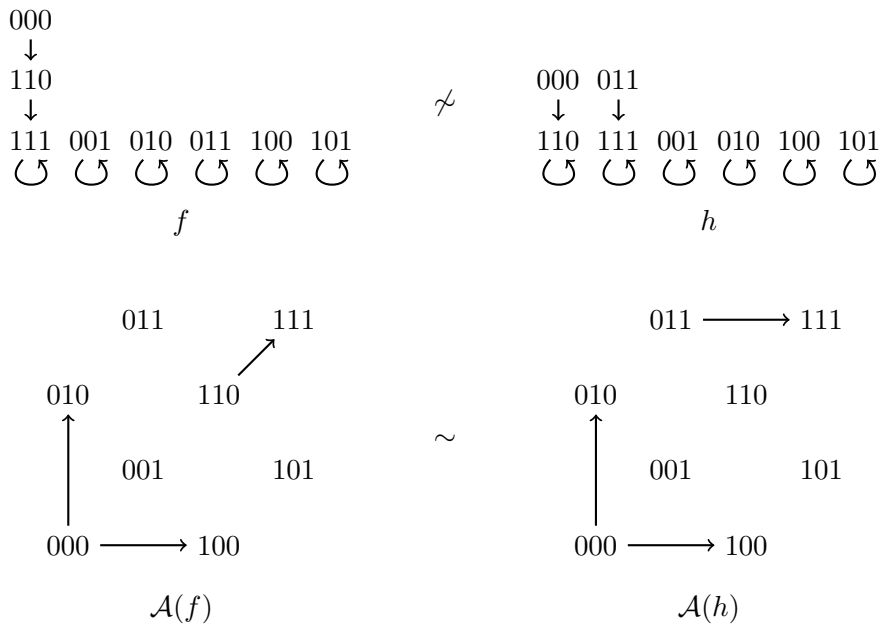
\begin{figure}[h]
\[
\arraycolsep=8pt
\begin{array}{ccc}
\begin{array}{c}
\def\sep{0.8}
  \begin{tikzpicture}
  \node (001) at ({1*\sep},0){$001$};
  \node (010) at ({2*\sep},0){$010$};
  \node (011) at ({3*\sep},0){$011$};
  \node (100) at ({4*\sep},0){$100$};
  \node (101) at ({5*\sep},0){$101$};  
  \node (111) at ({0*\sep},0){$111$};
  \node (000) at ({0*\sep},{2*\sep}){$000$};
  \node (110) at ({0*\sep},{1*\sep}){$110$};  
        ;
  \draw[->,thick]
  (000) edge (110)
  (110) edge (111)
  ;
  \draw[->,thick] (111.-112) .. controls ({0*\sep-0.5},{0-0.7}) and ({{0*\sep+0.5}},{0-0.7}) .. (111.-68);
  \draw[->,thick] (001.-112) .. controls ({1*\sep-0.5},{0-0.7}) and ({{1*\sep+0.5}},{0-0.7}) .. (001.-68);
  \draw[->,thick] (010.-112) .. controls ({2*\sep-0.5},{0-0.7}) and ({{2*\sep+0.5}},{0-0.7}) .. (010.-68);
  \draw[->,thick] (011.-112) .. controls ({3*\sep-0.5},{0-0.7}) and ({{3*\sep+0.5}},{0-0.7}) .. (011.-68);
  \draw[->,thick] (100.-112) .. controls ({4*\sep-0.5},{0-0.7}) and ({{4*\sep+0.5}},{0-0.7}) .. (100.-68);
  \draw[->,thick] (101.-112) .. controls ({5*\sep-0.5},{0-0.7}) and ({{5*\sep+0.5}},{0-0.7}) .. (101.-68);
  \end{tikzpicture}
\end{array}
&\not\sim &
\begin{array}{c}
\def\sep{0.8}
  \begin{tikzpicture}
  \node (110) at ({0*\sep},0){$110$};
  \node (111) at ({1*\sep},0){$111$};
  \node (001) at ({2*\sep},0){$001$};
  \node (010) at ({3*\sep},0){$010$};
  \node (100) at ({4*\sep},0){$100$};  
  \node (101) at ({5*\sep},0){$101$};
  \node (000) at ({0*\sep},{1*\sep}){$000$};
  \node (011) at ({1*\sep},{1*\sep}){$011$}; 
  \node (fack) at ({0*\sep},{2*\sep}){\textcolor{white}{$000$}};  
        ;
  \draw[->,thick]
  (000) edge (110)
  (011) edge (111)
  ;
  \draw[->,thick] (110.-112) .. controls ({0*\sep-0.5},{0-0.7}) and ({{0*\sep+0.5}},{0-0.7}) .. (110.-68);
  \draw[->,thick] (111.-112) .. controls ({1*\sep-0.5},{0-0.7}) and ({{1*\sep+0.5}},{0-0.7}) .. (111.-68);
  \draw[->,thick] (001.-112) .. controls ({2*\sep-0.5},{0-0.7}) and ({{2*\sep+0.5}},{0-0.7}) .. (001.-68);
  \draw[->,thick] (010.-112) .. controls ({3*\sep-0.5},{0-0.7}) and ({{3*\sep+0.5}},{0-0.7}) .. (010.-68);
  \draw[->,thick] (100.-112) .. controls ({4*\sep-0.5},{0-0.7}) and ({{4*\sep+0.5}},{0-0.7}) .. (100.-68);
  \draw[->,thick] (101.-112) .. controls ({5*\sep-0.5},{0-0.7}) and ({{5*\sep+0.5}},{0-0.7}) .. (101.-68);
  \end{tikzpicture}
\end{array}
\\[5mm]
f && h
\\[8mm]
  \begin{array}{c}
  \begin{tikzpicture}
  \pgfmathparse{1}
  \node (000) at (0,0){$000$};
  \node (001) at (1,1){$001$};
  \node (010) at (0,2){$010$};
  \node (011) at (1,3){$011$};
  \node (100) at (2,0){$100$};
  \node (101) at (3,1){$101$};
  \node (110) at (2,2){$110$};
  \node (111) at (3,3){$111$};
  \path[thick,->,draw]
  (000) edge (100)
  (000) edge (010)
  (110) edge (111)
  ;
  \end{tikzpicture}
  \end{array}
&
\sim
&
  \begin{array}{c}
  \begin{tikzpicture}
  \pgfmathparse{1}
  \node (000) at (0,0){$000$};
  \node (001) at (1,1){$001$};
  \node (010) at (0,2){$010$};
  \node (011) at (1,3){$011$};
  \node (100) at (2,0){$100$};
  \node (101) at (3,1){$101$};
  \node (110) at (2,2){$110$};
  \node (111) at (3,3){$111$};
  \path[thick,->,draw]
  (000) edge (100)
  (000) edge (010)
  (011) edge (111)
  ;
  \end{tikzpicture}
  \end{array}
\\[18mm]
\A(f) && \A(h)
\end{array}
\]
{\caption{\label{fig:counter_A_to_S} Example of non-isomorphic Boolean networks with isomorphic asynchronous graphs. 
}}
\end{figure}

\medskip
To proceed to the details, we need some definitions. Let \EM{$Q_n$} denote the \EM{$n$-cube}, that is, the graph on $\B^n$ where two configurations $x,y$ are adjacent iff $d(x,y)=1$. An \EM{isometry} is an automorphism of $Q_n$, that is, a permutation $\pi$ of $\B^n$ such that $\pi(x)\pi(y)\in E(Q_n)$ for all $xy\in E(Q_n)$. We say that two graphs or digraphs $G,H$ on $\B^n$ are \EM{isometric} if there exists an isomorphism between $G$ and $H$ which is an isometry.    

\begin{lemma}\label{lem:isometry}
Let $f,h\in F(n)$ and suppose that $\A(f)$ and $\A(h)$ are isometric. Then $f\sim h$.
\end{lemma}

\begin{proof}
Given $x\in\B^n$, let $\A_f(x)$ be the out-neighbors of $x$ in $\A(f)$. For all $i\in [n]$, we have $f_i(x)\neq x_i$ iff $\A(f)$ has an arc from $x$ to $x+e_i$. Hence, 
\begin{eqnarray}
f(x) &=&x+\sum_{x+e_i \in\A_f(x)}e_i \\ &=&x+\sum_{y\in\A_f(x)}(x+y).\label{eq:A_to_f}
\end{eqnarray}

Suppose that some isomorphism $\pi$ from $\A(f)$ to $\A(h)$ is an isometry. Then, there exists a permutation $\sigma$ of $[n]$ and a configuration $a\in\B^n$ such that $\pi(x)=\sigma(x)+a$ for all $x\in\B^n$, where $\sigma(x)=(x_{\sigma(1)},\dots,x_{\sigma(n)})$; see \cite{FMR21,H00,Ru2017} for instance. Given $x,y\in\B^n$, we have $\sigma(x+y)=\sigma(x)+\sigma(y)$. Thus, $\pi(x+y)= \sigma(x+y) + a = \sigma(x) +\sigma(y) +a = \pi(x) + a +\pi(y) +a +a = \pi(x) +\pi(y) +a$. We deduce that, for every odd integer $k$ and $x^1,\dots,x^k\in\B^n$, we have $\pi(x^1+\dots+x^k)=\pi(x^1)+\dots+\pi(x^k)$. Since \eqref{eq:A_to_f} is a sum of $2 |\A_f(x)|+1$ terms (in which the term $x$ appears $|\A_f(x)|+1$ times), and since $y\in\A_f(x)$ iff $\pi(y)\in\A_h(\pi(x))$ (because $\pi$ is an isomorphism from $\A(f)$ to $\A(h)$), we have 
\[
\pi(f(x))=\pi(x)+\sum_{y\in \A_f(x)}(\pi(x)+\pi(y))=\pi(x)+\sum_{\pi(y)\in \A_h(\pi(x))}(\pi(x)+\pi(y))=h(\pi(x)). 
\]
Thus, $\pi$ is an isomorphism from $f$ to $h$.
\end{proof}

The above lemma allows us to translate Theorem~\ref{thm:A_to_S} into a problem concerning spanning subgraphs of $Q_n$ (that is, subgraphs of $Q_n$ obtained by deleting edges only). So let \EM{$\Omega_n$} be the set of spanning subgraphs of $Q_n$ and $G\in\Omega_n$. An \EM{embedding} of $G$ is a permutation $\pi$ of $\B^n$ such that $\pi(x)\pi(y)\in E(Q_n)$ for all $xy\in E(G)$. An edge $xy\in E(Q_n)$ is \EM{solid} in $G$ if $\pi(x)\pi(y)\in E(Q_n)$ for all embeddings $\pi$ of $G$. We denote by \EM{$\tilde G$} the graph obtained from $G$ by adding all the edges of $Q_n$ that are solid in $G$. Clearly, every edge of $G$ is solid, that is, $G\subseteq \tilde G$. We say that $G$ is \EM{solid} if $\tilde G=Q_n$; equivalently, all the embeddings of $G$ are isometries. 

\medskip
Let $U\A(f)$ be the undirected version of $\A(f)$ that is, the graph in $\Omega_n$ with an edge between $x$ and $y$ if and only if $\A(f)$ has an arc from $x$ to $y$ or from $y$ to $x$. Our interest for solid spanning subgraphs of $Q_n$ comes from the following property, which is an easy consequence of Lemma~\ref{lem:isometry}. It shows that if $U\A(f)$ is solid, then the unlabelled synchronous dynamics of $f$ can be reconstructed from its unlabelled asynchronous dynamics. 

\begin{lemma}
For every $f,h\in F(n)$, if $U\A(f)$ is solid then 
\[
\A(f)\sim \A(h)~\Longrightarrow~ f\sim h.
\] 
\end{lemma}

\begin{proof}
Let $\pi$ be an isomorphism from $\A(f)$ to $\A(h)$. Then, $\pi$ is an isomorphism from $U\A(f)$ to $U\A(h)$, and thus $\pi$ is an embedding of $U\A(f)$. Since $U\A(f)$ is solid, $\pi$ is an isometry. So $\A(f)$ and $\A(h)$ are isometric and thus $f\sim h$ by Lemma~\ref{lem:isometry}.  
\end{proof}

Consequently, to prove Theorem~\ref{thm:A_to_S} it is sufficient to prove that, taking $f\in F(n)$ uniformly at random, the probability that $U\A(f)$ is solid tends to $1$ as $n\to\infty$. Given $xy\in E(Q_n)$ and $f\in F(n)$ taken uniformly at random, the probability that $\A(f)$ has no arc from $x$ to $y$ {\em and} no arc from $y$ to $x$ is $1/4$. Hence, for all $xy\in E(Q_n)$, the probability that $xy$ is an edge of $U\A(f)$ is $3/4$ and these events are pairwise independent. Hence, to prove Theorem~\ref{thm:A_to_S} it is sufficient to prove that if $G$ is a random spanning subgraph of $Q_n$ obtained by selecting each edge of $Q_n$ with probability $p=3/4$ then the probability that $G$ is solid ($\tilde G=Q_n$) tends to $1$ as $n\to\infty$. We will actually prove something slightly stronger, Theorem~\ref{thm:solid} below, that $p\geq 0.72$ is enough for this property.

\medskip
Let us first fix some notations. If $A\subseteq E(Q_n)$ and $G\in\Omega_n$, we abusively write $A\subseteq G$ instead of $A\subseteq E(G)$, and we denote by $|G|$ the number of edges in $G$ (thus $|Q_n|=n2^{n-1}$).  Given an underlying fixed probability $0\leq p\leq 1$ of edge selection in $Q_n$, we denote by $\Pr$ the probability function over the sample space $\Omega_n$ where the probability of an outcome $H\in\Omega_n$~is 
\[
\Pr[G=H]=p^{|H|}(1-p)^{|Q_n|-|H|}.
\]

\begin{theorem}\label{thm:solid}
For $0.72\leq p\leq 1$ we have 
\[
\lim_{n\to\infty} \Pr[\tilde G=Q_n]= 1.
\]
\end{theorem}

We begin with a simple observation. 

\begin{lemma}\label{lem:subgraph_tilde}
For every $G,H\in\Omega_n$, if $G\subseteq H$ then $\tilde G\subseteq \tilde H$.
\end{lemma}

\begin{proof}
Let $xy\in E(\tilde G)$ and let $\pi$ be an embedding of $H$. Since $G\subseteq H$, $\pi$ is also an embedding of $G$, and since $xy$ is solid in $G$, we have $\pi(x)\pi(y)\in E(Q_n)$, and thus $xy\in E(\tilde H)$. 
\end{proof}

For $A,B\subseteq E(Q_n)$, it is rather intuitive that $\Pr[B\subseteq G\mid A\subseteq \tilde G]\geq \Pr[B\subseteq G]$: the fact that a set $A$ of edges is solid in $G$ cannot decrease the probability that a set $B$ of edges is contained in $G$.
It is equivalently intuitive that $\Pr[B\not\subseteq G\mid A\subseteq \tilde G]\leq \Pr[B\not\subseteq G]$ as the event $B\not\subseteq G$ is simply the negation of the event $B\subseteq G$. We prove this below in a stronger form. 

\begin{lemma}\label{lem:conditional_proba}
Let $A\subseteq E(Q_n)$ and disjoint sets $B_1,\dots,B_k\subseteq E(Q_n)$. Let $P$ be the set of $G\in\Omega_n$ with  $B_i\not\subseteq G$ for all $1\leq i\leq k$, and suppose that $A\subseteq\tilde G$ for some $G\in P$. Then,  
\[
\Pr[G\in P\mid A\subseteq \tilde G]\leq \Pr[G\in P].
\] 
\end{lemma}

\begin{proof}
Let $P_0=\Omega_n$ and, for all $1\leq i\leq k$, let $P_i$ be the set of $G\in\Omega_n$ with  $B_j\not\subseteq G$ for $1\leq j\leq i$. We have,
\begin{eqnarray}\label{eq:P}
\Pr[G\in P\mid A\subseteq \tilde G]
&=&\prod_{i=1}^k \Pr[B_i\not\subseteq G\mid A\subseteq \tilde G\land G\in P_{i-1}]\nonumber\\
&=&\prod_{i=1}^k 1-\Pr[B_i\subseteq G\mid A\subseteq \tilde G\land G\in P_{i-1}].%\\
\end{eqnarray}

\medskip
Fix $1\leq i\leq k$ and let us prove, by induction on $|B_i|$, that
\begin{eqnarray}\label{eq:B_i}
\Pr[B_i\subseteq G\mid A\subseteq \tilde G\land G\in P_{i-1}]\geq \Pr[B_i\subseteq G].
\end{eqnarray}
For $|B_i|=0$ this is obvious because $\emptyset \subsetneq G$ is tautologically true. So suppose that $|B_i|>0$. Let $e\in B_i$ and $X:=B_i\setminus e$. Let also
\begin{eqnarray*}
p_X&:=&\Pr[X\subseteq G\mid A\subseteq \tilde G\land G\in P_{i-1}],\\
p_e&:=&\Pr[e\in G\mid A\subseteq \tilde G\land G\in P_{i-1}\land  X\subseteq G].
\end{eqnarray*}
Hence, \eqref{eq:B_i} says 
\[
p_X\cdot p_e\geq \Pr[B_i\subseteq G]=\Pr[X\subseteq G]\cdot\Pr[e\in G].
\]
By induction, $p_X\geq \Pr[X\subseteq G]$ thus it is sufficient to prove that $p_e\geq \Pr[e\in G]=p$. Let $R$ be the set of $H\in P_{i-1}$ such that $A\subseteq \tilde H$ and $X \subseteq H$. Let $R_0$ be the set of $H\in R$ with $e\not\in H$ and $R_1:=R\setminus R_0$. Hence, 
\[
p_e=\Pr[G\in R_1\mid G\in R].
\]
If $R_0$ is empty, then $R=R_1$ thus $p_e=1$ and we are done. So suppose that $R_0$ is not empty. For every $H\in R_0$, let $H'$ be obtained from $H$ by adding $e$. We have $X \subseteq H\subseteq H'$ and thus $A\subseteq \tilde H\subseteq \tilde H'$ by Lemma~\ref{lem:subgraph_tilde}. Furthermore, if $i=1$ then $H' \in P_{i-1} = P_0 = \Omega_n$ and otherwise, since $e\not\in B_j$ for $1\leq j\leq i-1$, we have $H'\in P_{i-1}$. Thus, $H'\in R_1$. Hence, $H\mapsto H'$ is an injection from $R_0$ to~$R_1$.  For every $H\in R_0$, 
\[
\Pr[G=H]=p^{|H|}(1-p)^{|Q_n|-|H|}=\left(\frac{1-p}{p}\right)p^{|H'|}(1-p)^{|Q_n|-|H'|}= \left(\frac{1-p}{p}\right) \Pr[G=H'].  
\]
Hence, since $H\mapsto H'$ is an injection from $R_0$ to~$R_1$, 
\begin{eqnarray*}
\Pr[G\in R_0]&=&\sum_{H\in R_0} \Pr[G=H]\\
&=&\left(\frac{1-p}{p}\right)\sum_{H\in R_0}  \Pr[G=H']\\
&\leq& \left(\frac{1-p}{p}\right)\sum_{F\in R_1} \Pr[G=F]\\
&=& \left(\frac{1-p}{p}\right)\Pr[G\in R_1].
\end{eqnarray*}
Consequently, setting $r_0:=\Pr[G\in R_0]$ and $r_1:=\Pr[G\in R_1]$, we have 
\[
p_e=\Pr[G\in R_1\mid G\in R]=\frac{r_1}{r_0+r_1}
\geq  \frac{r_1}{\left(\frac{1-p}{p}\right)r_1+r_1}=p.
\]
This proves \eqref{eq:B_i}. Plugging \eqref{eq:B_i} into \eqref{eq:P} we obtain 
\begin{eqnarray*}
\Pr[G\in P\mid A\subseteq \tilde G]
&\leq &\prod_{i=1}^k 1-\Pr[B_i\subseteq G]\\
&= &\prod_{i=1}^k \Pr[B_i\not\subseteq G]\\
&=&\Pr[G\in P].
\end{eqnarray*}
\end{proof}

The fact that an edge is solid in $G$ is somewhat enigmatic. This leads us to find a concrete sufficient condition for an edge to be solid. We need some definitions. A \EM{staple} $S$ on an edge $xy$ of $Q_n$ is a path of $Q_n$ of length $3$ between $x$ and $y$, that is, three edges in $Q_n$ of the form $\{xx',x'y',y'y\}$. Hence, there are $n-1$ staples on $xy$, and they are pairwise disjoint; see Figure \ref{fig:staples}. The crucial observation is that if $G$ contains at least $4$ staples on $xy$, then $xy$ is solid. Actually, the weaker condition that $\tilde G$ contains at least $4$ staples on $xy$ is still sufficient for $xy$ to be solid:

\begin{figure}[h]
\[
  \begin{tikzpicture}
  \pgfmathparse{1}
  \node (0000) at (0,0){$0000$};
  \node (1000) at (3,0){$1000$};
  \node (0100) at (-0.9,1){$0100$};
  \node (0010) at (0,3){$0010$};
  \node (0001) at (0.9,2){$0001$};
  \node (1100) at (2.1,1){$1100$};
  \node (1010) at (3,3){$1010$};
  \node (1001) at (3.9,2){$1001$};
  \path[thick,draw]
  (0000) edge (0100) 
  (0100) edge (1100)
  (1100) edge (1000)
  (0000) edge (0010) 
  (0010) edge (1010)
  (1010) edge (1000)
  (0000) edge (0001) 
  (0001) edge (1001)
  (1001) edge (1000)
  ;
  \end{tikzpicture}
\]
{\caption{\label{fig:staples} The three staples on $\ZERO e_1$ for $n=4$.}

}
\end{figure}
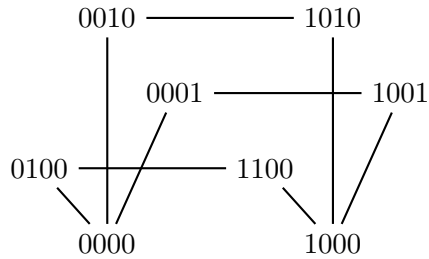

\begin{lemma}\label{lem:staples}
Let $G\in\Omega_n$ and $xy\in E(Q_n)$. If $\tilde G$ contains at least $4$ staples on $xy$ then $xy\in E(\tilde G)$. 
\end{lemma}

\begin{proof}
Let $\pi$ be an embedding of $G$. Suppose that $\tilde G$ contains $4$ distinct staples on $xy$, say $S_1,\dots,S_4$, and let $S_i=\{xx^i,x^iy^i,y^iy\}$. Then, $P_i:=\pi(x)$-$\pi(x^i)$-$\pi(y^i)$-$\pi(y)$ is a path of length $3$ in $Q_n$ between $\pi(x)$ and $\pi(y)$, and thus $d(\pi(x),\pi(y))\leq 3$ and $d(\pi(x),\pi(y))$ is odd. If $d(\pi(x),\pi(y))=3$ then $Q_n$ has exactly $3$ edges-disjoint paths of length $3$ between $\pi(x)$ and $\pi(y)$, which is a contradiction since $P_1,P_2,P_3,P_4$ are edges-disjoint paths of length $3$ between $\pi(x)$ and $\pi(y)$. Thus, $d(\pi(x),\pi(y))<3$ and this forces $d(\pi(x),\pi(y))=1$. Consequently, $xy\in\tilde G$.
\end{proof}

We are now in position to prove Theorem \ref{thm:solid}.

\begin{proof}[{\bf Proof of Theorem \ref{thm:solid}}]
If $p=1$ the result is obvious so suppose that $0.72\leq p<1$. Let $x^1y^1,\dots,x^my^m$ be an enumeration of the $m=n2^{n-1}$ edges of $Q_n$ in such a way that for all $1\leq i\leq j\leq m$ we have $w(x^i)<w(y^i)$ and $w(x^i)\leq w(x^j)$. Let $A_0=\emptyset$ and, for $1\leq i\leq m$, let $A_i$ be the set of edges $x^jy^j$ with $1\leq j\leq i$. We have
\begin{eqnarray}\label{eq:pro_solid}
\Pr[\tilde G=Q_n]&=&\prod_{i=1}^m \Pr[x^iy^i\in \tilde G\mid A_{i-1}\subseteq \tilde G]\nonumber\\
&=&\prod_{i=1}^m 1-\Pr[x^iy^i\not\in \tilde G\mid A_{i-1}\subseteq \tilde G].
\end{eqnarray}

\medskip
Let 
\[
q=1-p^3+0.01,\qquad \ell=\frac{1-p}{1-p^3},
\]
and note that $q<1$ since $p\geq 0.72$. In order to lower bound the probability that $\tilde G=Q_n$, we will give an upper bound on the probability that $x^iy^i\not\in\tilde G$ under the condition $A_{i-1}\subseteq \tilde G$, which only depends on $n$ and $w(x^i)$: for $n$ large enough, 
\begin{eqnarray}\label{eq:bound_no_solid}
\Pr[x^iy^i\not\in\tilde G\mid A_{i-1} \subseteq \tilde G ]\leq q^n\ell^{w(x^i)}. 
\end{eqnarray}

\medskip
Let $\S^+$ be the set of staples $\{x^ix,xy,yy^i\}$ on $x^iy^i$ with $w(x^i)<w(x)$, and $\S^-$ be the set of staples $\{x^ix,xy,yy^i\}$ on $x^iy^i$ with $w(x)<w(x^i)$. We have $|\mathcal S^+| = n - w(y^i) = n - w(x^i) - 1$ and $|\S^-|=w(x^i)$. Given a staple $S=\{x^ix,xy,yy^i\}$, we set $S'=\{yy^i\}$ if $S\in\S^-$ and $S'=S$ otherwise. The crucial observation is that if $S\in\S^-$ then $x^ix,xy\in A_{i-1}$ so
\begin{eqnarray}\label{eq:obs}
A_{i-1}\subseteq \tilde G\qquad\Longrightarrow \qquad\left(S\subseteq \tilde G\iff S'\subseteq \tilde G\quad\forall S\in\S^+\cup\S^-\right). 
\end{eqnarray}

\medskip
Given $\calB\subseteq \S^+\cup\S^-$, let $P(\calB)$ be the set of $G\in\Omega_n$ such that $S'\not\subseteq G$ for all $S\in \calB$.  For $S\in\calB$, the probability that $S'\subseteq G$ is $p^3$ if $S\in\S^+$ and $p$ otherwise, so 
\begin{eqnarray*}
\Pr[G\in P(\cal B)]&=&(1-p^3)^{|\calB\cap\S^+|} (1-p)^{|\calB\cap\S^-|}.
\end{eqnarray*}

\medskip
By Lemma~\ref{lem:staples}, if $x^iy^i\not\in\tilde G$ then $\tilde G$ contains at most $3$ staples on $x^iy^i$. Equivalently, there exists a set $\calB\subseteq\S^+\cup\S^-$ of size at least $n-4$ such that $S\not\subseteq \tilde G$ for all $S\in\calB$; by \eqref{eq:obs}, under the condition $A_{i-1}\subseteq \tilde G$, this is equivalent to $\tilde G\in P(\calB)$, which implies $G\in P(\calB)$. Consequently, denoting by $\Lambda$ the set of $\calB\subseteq\S^+\cup\S^-$ of size at least $n-4$, we deduce, using the union bound for the first inequality and Lemma \ref{lem:conditional_proba} for the second, that 
\begin{eqnarray*}
\Pr[x^iy^i\not\in\tilde G\mid A_{i-1}\subseteq\tilde G]
&\leq& 
\sum_{\calB\in\Lambda}\Pr[G\in P(\calB)\mid A_{i-1}\subseteq \tilde G]\\
&\leq &
\sum_{\calB\in\Lambda}\Pr[G\in P(\calB)]\\
&=&\sum_{\calB\in\Lambda}(1-p^3)^{|\calB\cap\S^+|} (1-p)^{|\calB\cap\S^-|}.
\end{eqnarray*}
Since $|\calB|\geq n-4$ we have $|\calB\cap \S^+|\geq n-w(x^i)-4$ and $|\calB\cap\S^-|\geq w(x^i)-3$. Furthermore, $|\Lambda|=\sum_{k=0}^3{n-1\choose n-1-k}\leq n^3$. We deduce that  
\begin{eqnarray*}
\Pr[x^iy^i\not\in\tilde G\mid A_{i-1}\subseteq\tilde G]
&\leq &
n^3(1-p^3)^{n-w(x^i)-4}(1-p)^{w(x^i)-3}.
\end{eqnarray*}
Since $1-p^3<q$, for $n$ large enough we have $n^3(1-p^3)^{n-4}(1-p)^{-3}\leq q^n$ and thus 
\begin{eqnarray*}
\Pr[x^iy^i\not\in\tilde G\mid A_{i-1}\subseteq\tilde G]
&\leq &
q^n(1-p^3)^{-w(x^i)}(1-p)^{w(x^i)}~=~q^n\ell^{w(x^i)}.
\end{eqnarray*}
This proves \eqref{eq:bound_no_solid}.

\medskip
Pulgging \eqref{eq:bound_no_solid} in \eqref{eq:pro_solid} and noting that, for any integer $0\leq w\leq n-1$ there are $(n-w){n\choose w}$ edges $x^iy^i$ with $w(x^i)=w$, we obtain, using the rough bound  $(n-w){n\choose w}\leq n(en/w)^w$, 
\begin{eqnarray}\label{eq:GQ_n}
\Pr[\tilde G=Q_n]
&\geq& \prod_{i=1}^m 1-q^n\ell^{w(x^i)}\nonumber\\
&= &\prod_{w=0}^{n-1}\left(1-q^n\ell^{w}\right)^{(n-w){n\choose w}}\nonumber\\ 
&\geq &(1-q^n)^n\prod_{w=1}^{n-1}\left(1-q^n \ell^{w}\right)^{n(en/w)^w}.
\end{eqnarray}

%\medskip
%For any $0<\delta<1$, we have 
%\[
%-\ln(1-\delta)
%~=~\sum_{k=1}^\infty \frac{\delta^k}{k} 
%~\leq~ \sum_{k=1}^\infty \delta^k
%~=~ -1+\sum_{k=0}^\infty \delta^k
%~=~ -1+\frac{1}{1-\delta}
%~=~ \frac{\delta}{1-\delta}.
%\]
%Hence, if $\delta\leq 1/2$ then $-\ln(1-\delta)\leq 2\delta$ and thus $(1-\delta)\geq e^{-2\delta}$. For $n$ large enough, we have $q^n\ell^w\leq q^n\leq 1/2$, and we deduce that  $1-q^n\ell^w\geq e^{-2q^n\ell^w}$ and $1-q^n\geq e^{-2q^n}$. Combining this with \eqref{eq:GQ_n} we obtain 
%\begin{eqnarray*}
%\Pr[\tilde G=Q_n]&\geq &e^{-2nq^n}
%\prod_{w=1}^{n-1} e^{-2q^n\ell^wn(en/w)^w}
%\end{eqnarray*}

\medskip
Since $1-x\geq e^{-2x}$ for any $0\leq x\leq 1/2$, and since $0\leq q^n\ell^w\leq q^n\leq 1/2$ for $n$ large enough, we have $1-q^n\ell^w\geq e^{-2q^n\ell^w}$ and $1-q^n\geq e^{-2q^n}$. Combining this with \eqref{eq:GQ_n} we obtain 
\begin{eqnarray*}
\Pr[\tilde G=Q_n]&\geq &e^{-2nq^n}
\prod_{w=1}^{n-1} e^{-2q^n\ell^wn(en/w)^w}.
\end{eqnarray*}

\medskip
The derivative of the real function $x\mapsto \ell^x (en/x)^x$ is $x\mapsto (e\ell n/x)^x\ln(\ell n/x)$. Hence, this derivative has a unique root, namely $\ell n$, and is positive for $x<\ell n$ and negative for $x>\ell n$. Consequently, the maximum of $\ell^x (en/x)^x$ is obtained for $x=\ell n$. We deduce that 
\begin{eqnarray*}
\Pr[\tilde G=Q_n]&\geq &e^{-2nq^n}\prod_{w=1}^{n-1} e^{-2q^n\ell^{\ell n}n(en/{\ell n})^{\ell n}}\\
&= &e^{-2nq^n}\prod_{w=1}^{n-1} e^{-2n(qe^\ell)^n}\\
&\geq &e^{-2[n^2(qe^\ell)^n+nq^n]}.\\
\end{eqnarray*}

\medskip
Since $p\geq 0.72$ we have $qe^\ell<1$. Hence, $n^2(qe^\ell)^n+nq^n$ tends to $0$ as $n\to\infty$ and thus
\[
\lim_{n\to\infty} \Pr[\tilde G=Q_n]=1.
\]
\end{proof}

%%%%%%%%%%%%%%%%%%%%%%%%%%%%%%%%%%%%%%%%%%%%%%%%%%%%%%%%%%%%%%%%%%%
\section{Asynchronous reconstruction}\label{sec:S_to_A}
%%%%%%%%%%%%%%%%%%%%%%%%%%%%%%%%%%%%%%%%%%%%%%%%%%%%%%%%%%%%%%%%%%%

In this section, we prove Theorem~\ref{thm:S_to_A}, that if $f$ is neither constant nor the identity and $n\geq 3$, then there exists $h\sim f$ with $\A(h)\not\sim \A(f)$: the unlabelled asynchronous dynamics cannot be fully reconstructed from the unlabelled synchronous dynamics. 

\begin{remark}
It is necessary to exclude the case where $f$ is constant or the identity. Firstly, if $f=\id$ and $h\sim f$, then $h=\id$ thus $\A(f)=\A(h)$. Secondly, if $f$ always returns some configuration $a$ and $h\sim f$, then $h$ always returns some configuration $b$, and we easily check that $x\mapsto x+a+b$ is an isomorphism from $\A(f)$ to $\A(h)$. Furthermore, we will show that for $n\leq 2$, there are, up to isomorphism, exactly two Boolean networks which do not satisfy the conclusion of Theorem~\ref{thm:S_to_A} (cf. Remark \ref{rem:counter_example_A_to_S}). 
\end{remark}

For every $f\in F(n)$, we denote by \EM{$\Delta^+(f)$} the set of configurations $x$ such that $f(x)=\overline{x}$; equivalently, $\Delta^+(f)$ is the set of configurations with out-degree $n$ in $\A(f)$. We say that $f$ contains $2P_1$ if its synchronous graph contains two vertex-disjoint paths of length $1$; equivalently, there exist configurations $x,y$ such that $x,y,f(x),f(y)$ are all distinct. The next lemma shows that the presence of $2P_1$ is sufficient for the existence of $h\sim f$ such that $|\Delta^+(h)|\neq|\Delta^+(f)|$, which implies  $\A(h)\not\sim \A(f)$.

\begin{lemma}\label{lem:2P_1}
If $f\in F(n)$ contains $2P_1$, then there exists $h\sim f$ with $|\Delta^+(h)|\neq |\Delta^+(f)|$.
\end{lemma}

\begin{proof}
Suppose that $f$ contains $2P_1$. Hence, there are configurations $a,b$ such that $a,b,f(a),f(b)$ are all distinct. Let $x,y \in \B^n$ with $0<d(x,y)<n$ (so $x,y,\overline{x}$ and $\overline{y}$ are all distinct). Let 
\[
\pi = (a \leftrightarrow x)\circ (b \leftrightarrow y) \circ(f(a) \leftrightarrow \overline{x})\circ(f(b) \leftrightarrow \overline{y})
\]
and 
\[
h^1 = \pi \circ f \circ \pi^{-1},\quad h^2 = (y \leftrightarrow \overline{x}) \circ h^1 \circ (y \leftrightarrow \overline{x}), \quad h^3 = (\overline{x} \leftrightarrow \overline{y}) \circ h^1 \circ (\overline{x} \leftrightarrow \overline{y}).
\]
See Figure \ref{fig:2P_1} for an illustration. 

\medskip
For all $i\in [3]$, let $\Delta^i=\Delta^+(h^i)$. We will prove that either $|\Delta^2|<|\Delta^1|$ or $|\Delta^3|<|\Delta^1|$, which proves the lemma (because $f\sim h^1 \sim h^2 \sim h^3$). Let $X=\{x,y,\overline{x},\overline{y}\}$ and $Y=\B^n\setminus X$; note that since $z\in X\iff \overline{z}\in X$, we have $z\in Y \iff \overline{z}\in Y$. For all $i,j\in [3]$, if $z \in \Delta^i\cap Y$ then $h^i(z) = \overline{z}\in Y$  thus $h^i(z)=h^j(z)$ since $(y \leftrightarrow \overline{x})$ and $(\overline{x} \leftrightarrow \overline{y})$ act as the identity on $Y$, so that $z\in \Delta^j\cap Y$. Consequently, $\Delta^i\cap Y\subseteq \Delta^j\cap Y$ and thus
\[
\Delta^1\cap Y= \Delta^2 \cap Y=\Delta^3\cap Y. 
\]
Hence, to prove that $|\Delta^2|<|\Delta^1|$ or $|\Delta^3|<|\Delta^1|$, we have to prove that $|\Delta^2\cap X|<|\Delta^1\cap X|$ or $|\Delta^3\cap X|<|\Delta^1\cap X|$. Since $|\Delta^1\cap X|\geq 2$ (because $h^1(x) = \overline{x}$ and $h^1(y) = \overline{y}$), it is actually sufficient to prove that $|\Delta^2\cap X|\leq 1$ or $|\Delta^3\cap X|=0$. Suppose first that $h^1(\overline{x}) \neq \overline{y}$ or $h^1(\overline{y}) \neq \overline{x}$. Then $h^2(y) \neq \overline{y}$ or $h^2(\overline{y}) \neq y$, and since $h^2(x) = y \neq \overline{x}, h^2(\overline{x}) = \overline{y} \neq x$ we have $|\Delta^2\cap X|\leq 1$. Suppose now that $h^1(\overline{x}) = \overline{y}$ and $h^1(\overline{y}) = \overline{x}$. Then, $h^3(\overline{x}) = \overline{y} \neq x$ and $h^3(\overline{y}) = \overline{x} \neq y$, and since $h^3(x) = \overline{y} \neq \overline{x}, h^3(y) = \overline{x} \neq \overline{y}$ we have $|\Delta^3\cap X|=0$.
\end{proof}

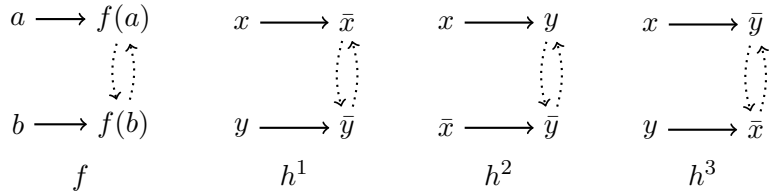
\begin{figure}[h]
	\def\sep{1.4}
\[
\arraycolsep=12pt
\begin{array}{cccc}
    \begin{tikzpicture}
        \node (a) at (0,0){$a$};
        \node (fa) at (\sep,0){$f(a)$};
        \node (b) at (0,-\sep){$b$};
        \node (fb) at (\sep,-\sep){$f(b)$};
        \path[thick,->,draw,black]
        (a) edge (fa)
        (b) edge (fb)
        %(fa) edge[dotted,bend right=15] (fb)
 		%(fb) edge[dotted,bend right=15] (fa)
        ;
    \end{tikzpicture}
    &
    \begin{tikzpicture}
        \node (a) at (0,0){$x$};
        \node (fa) at (\sep,0){$\overline{x}$};
        \node (b) at (0,-\sep){$y$};
        \node (fb) at (\sep,-\sep){$\overline{y}$};
        \path[thick,->,draw,black]
        (a) edge (fa)
        (b) edge (fb)
		%(fa) edge[dotted,bend right=15] (fb)
 		%(fb) edge[dotted,bend right=15] (fa)
        ;
    \end{tikzpicture}
    &
    \begin{tikzpicture}
        \node (a) at (0,0){$x$};
        \node (fa) at (\sep,0){$y$};
        \node (b) at (0,-\sep){$\overline{x}$};
        \node (fb) at (\sep,-\sep){$\overline{y}$};
        \path[thick,->,draw,black]
        (a) edge (fa)
        (b) edge (fb)
		%(fa) edge[dotted,bend right=15] (fb)
 		%(fb) edge[dotted,bend right=15] (fa)
        ;
    \end{tikzpicture}
    &
    \begin{tikzpicture}
        \node (a) at (0,0){$x$};
        \node (fa) at (\sep,0){$\overline{y}$};
        \node (b) at (0,-\sep){$y$};
        \node (fb) at (\sep,-\sep){$\overline{x}$};
        \path[thick,->,draw,black]
        (a) edge (fa)
        (b) edge (fb)
		%(fa) edge[dotted,bend right=15] (fb)
 		%(fb) edge[dotted,bend right=15] (fa)
        ;
    \end{tikzpicture}
    \\
    \textrm{$f$}
    &\textrm{$h^1$}
    &\textrm{$h^2$}
    &\textrm{$h^3$}
\end{array}
\]
  \caption{\label{fig:2P_1}
Illustration for Lemma \ref{lem:2P_1}. 
}
\end{figure}

We now give another sufficient condition for the existence of $h\sim f$ with $|\Delta^+(h)|\neq|\Delta^+(f)|$.

\begin{lemma}\label{lem:no_2P_2}
Let $f\in F(n)$ with $f\neq\id$ and $n\geq 2$. Suppose that $f$ does not contain $2P_1$ and has at least $2$ periodic configurations. Suppose also that $n\geq 3$ if $f$ has a limit cycle of length $\geq 3$. Then there exists $h \sim f$ with $|\Delta^+(h)|\neq |\Delta^+(f)|$. 
\end{lemma}

\begin{proof}
Suppose that $f$ satisfies the conditions of the statement. Since $f\neq\id$, there exists a configuration $a$ such that $a\neq f(a)$ and such that $f(a)$ is a periodic configuration; if $f$ has a limit cycle of length at least $2$, we chose $a$ in this limit cycle so that both $a$ and $f(a)$ are periodic configurations. Let $g=(\overline{a} \leftrightarrow f(a))\circ f\circ (\overline{a} \leftrightarrow f(a))$. Then, $g(a)=\overline{a}$ and $|\Delta^+(g)| \geq 1$. We prove below that there exists $h\sim g$ with $|\Delta^+(h)|<|\Delta^+(g)|$, which implies the lemma.  

\medskip
Suppose first that all the periodic configurations of $g$ are fixed points (thus, $a$ is a non-periodic configuration). Let $X$ be the non-periodic configurations of $g$ distinct from $a$. Then, $\overline{a}$ is a fixed point of $g$ and since $g$ does not contain $2P_1$, either $g(X)=\{a\}$ or $g(X)=\{\overline{a}\}$.  Furthermore, since $g$ has at least two periodic configurations, it has a fixed point $b\neq \overline{a}$. We then consider two cases, illustrated in Figure \ref{fig:no_2P_1_2}. 
\begin{enumerate}
\item
First, suppose that $g(X)=\{a\}$ (see Figure \ref{fig:no_2P_1_2}(a)). Let $h= (\overline{a} \leftrightarrow b) \circ g \circ (\overline{a} \leftrightarrow b)$. If $x \neq h(x)$ and $x\neq a$ then $h(x)=a$ and $x\neq \overline{a}$ (since $\overline{a}$ is a fixed point of $h$). Since $h(a)=b\neq\overline{a}$ this proves that $|\Delta^+(h)|=0$. 
\item 
Second, suppose that $g(X)=\{\overline{a}\}$ (see Figure \ref{fig:no_2P_1_2}(b)). Let $h= (a \leftrightarrow b) \circ g\circ (a \leftrightarrow b)$. If $x\neq h(x)$ then $h(x)=\overline{a}$ and $x\neq a$ (since $a$ is a fixed point of $h$). Thus, $|\Delta^+(h)|=0$. 
\end{enumerate}

\medskip
Suppose now that $g$ has a limit cycle $L$ of length $2$.  Since we chose $a$ so that both $a$ and $f(a)$ are periodic configurations of $f$, we have that $a$ and $g(a)=\overline{a}$ are periodic configurations of $g$, so the configurations of $L$ are $a$ and $\overline{a}$, which implies $|\Delta^+(g)| \geq 2$ (see Figure \ref{fig:no_2P_1_2}(c)). Furthermore, since $g$ does not contain $2P_1$, all the periodic configurations $x\neq a,\overline{a}$ are fixed points. Let $X$ be the set of non-periodic configurations of $g$. Since $g$ does not contain $2P_1$, either $g(X)=\{a\}$ or $g(X)=\{\overline{a}\}$. The two cases being symmetric, we assume, without loss, that $g(X)=\{a\}$. Let $b \neq a,\overline{a}$ and $h= (\overline{a} \leftrightarrow b) \circ g \circ (\overline{a} \leftrightarrow b)$. Since $h(a)=b$ and $h(b)=a$, we have $a,b\not\in \Delta^+(h)$. If $x\neq h(x)$ and $x\neq a,b$ then $h(x)=a$, and thus $x\in \Delta^+(h)$ if and only if $x=\overline{a}$. We deduce that $\Delta^+(h)$ is either empty or only contains $\overline{a}$. Hence, $|\Delta^+(h)|\leq 1 < |\Delta^+(g)|$.

\medskip
Suppose finally that $g$ has a limit cycle $L$ of length $\ell\geq 3$. Since $g$ does not contain $2P_1$, we have $\ell=3$ and all the configurations outside of $L$ are fixed points, so $g$ is a permutation and $L$ contains $a,\overline{a}$ and a third configuration $b$ (see Figure \ref{fig:no_2P_1_2}(d)). Let $c$ be a fixed point of $g$ distinct from $\overline{b}$, which exists since $n\geq 3$, and let $h= (\overline{a} \leftrightarrow c) \circ g \circ (\overline{a} \leftrightarrow c)$. We have $h(a) = c \neq \overline{a}$, $h(c) = b \neq \overline{c}$   and $h(b) = a \neq \overline{b}$. Furthermore, $h(x)=x\neq\overline{x}$ for all $x\neq a,b,c$, and thus, $|\Delta^+(h)|=0$.
\end{proof}

\begin{figure}[h]
\[
\tag{a}
\arraycolsep=20pt
\begin{array}{cc}
    \begin{tikzpicture}
		\node (X) at (0,2.5){$X$};
	    \node (x1) at (-0.5,2){\small $\bullet$};
        \node (x2) at (+0.5,2){\small $\bullet$};
		\node (na) at (0,0){$\overline{a}$};
        \node (a) at (0,1){$a$};
        \node (b) at (1,0){$b$};
        \node (c) at (2,0){\small$\bullet$};
        \draw [draw=black,dotted] (-0.7,1.8) rectangle (+0.7,2.2);
        \path[thick,->,draw,black]
        (a) edge (na)
		(x1) edge (a)
		(x2) edge (a)
		(x1) edge[-,dotted] (x2)
		(b) edge[-,dotted] (c)
        ;
        \draw[->,thick] (na.-112) .. controls ({0-0.5},{0-0.7}) and ({0+0.5},{0-0.7}) .. (na.-68);
        \draw[->,thick] (b.-112) .. controls ({1-0.5},{0-0.7}) and ({1+0.5},{0-0.7}) .. (b.-68);
		\draw[->,thick] (c.-112) .. controls ({2-0.5},{0-0.7}) and ({2+0.5},{0-0.7}) .. (c.-68);
    \end{tikzpicture}
&
    \begin{tikzpicture}
		\node (X) at (0,2.5){$X$};
	    \node (x1) at (-0.5,2){\small $\bullet$};
        \node (x2) at (+0.5,2){\small $\bullet$};
		\node (na) at (0,0){$b$};
        \node (a) at (0,1){$a$};
        \node (b) at (1,0){$\overline{a}$};
        \node (c) at (2,0){\small$\bullet$};
        \draw [draw=black,dotted] (-0.7,1.8) rectangle (+0.7,2.2);
        \path[thick,->,draw,black]
        (a) edge (na)
		(x1) edge (a)
		(x2) edge (a)
		(x1) edge[-,dotted] (x2)
		(b) edge[-,dotted] (c)
        ;
        \draw[->,thick] (na.-112) .. controls ({0-0.5},{0-0.7}) and ({0+0.5},{0-0.7}) .. (na.-68);
        \draw[->,thick] (b.-112) .. controls ({1-0.5},{0-0.7}) and ({1+0.5},{0-0.7}) .. (b.-68);
		\draw[->,thick] (c.-112) .. controls ({2-0.5},{0-0.7}) and ({2+0.5},{0-0.7}) .. (c.-68);
    \end{tikzpicture}
    \\
g & h
\end{array}
\]
~\\
\[
\tag{b}
\arraycolsep=20pt
\begin{array}{ccc}
    \begin{tikzpicture}
		\node (X) at (-1,1.5){$X$};
	    \node (x1) at (-1.5,1){\small $\bullet$};
        \node (x2) at (-0.5,1){\small $\bullet$};
		\node (na) at (0,0){$\overline{a}$};
        \node (a) at (0,1){$a$};
        \node (b) at (1,0){$b$};
        \node (c) at (2,0){\small$\bullet$};
        \draw [draw=black,dotted] (-1.7,0.8) rectangle (-0.3,1.2);
        \path[thick,->,draw,black]
        (a) edge (na)
		(x1) edge (na)
		(x2) edge (na)
		(x1) edge[-,dotted] (x2)
		(b) edge[-,dotted] (c)
        ;
        \draw[->,thick] (na.-112) .. controls ({0-0.5},{0-0.7}) and ({0+0.5},{0-0.7}) .. (na.-68);
        \draw[->,thick] (b.-112) .. controls ({1-0.5},{0-0.7}) and ({1+0.5},{0-0.7}) .. (b.-68);
		\draw[->,thick] (c.-112) .. controls ({2-0.5},{0-0.7}) and ({2+0.5},{0-0.7}) .. (c.-68);
    \end{tikzpicture}
&
    \begin{tikzpicture}
		\node (X) at (-1,1.5){$X$};
	    \node (x1) at (-1.5,1){\small $\bullet$};
        \node (x2) at (-0.5,1){\small $\bullet$};
		\node (na) at (0,0){$\overline{a}$};
        \node (a) at (0,1){$b$};
        \node (b) at (1,0){$a$};
        \node (c) at (2,0){\small$\bullet$};
        \draw [draw=black,dotted] (-1.7,0.8) rectangle (-0.3,1.2);
        \path[thick,->,draw,black]
        (a) edge (na)
		(x1) edge (na)
		(x2) edge (na)
		(x1) edge[-,dotted] (x2)
		(b) edge[-,dotted] (c)
        ;
        \draw[->,thick] (na.-112) .. controls ({0-0.5},{0-0.7}) and ({0+0.5},{0-0.7}) .. (na.-68);
        \draw[->,thick] (b.-112) .. controls ({1-0.5},{0-0.7}) and ({1+0.5},{0-0.7}) .. (b.-68);
		\draw[->,thick] (c.-112) .. controls ({2-0.5},{0-0.7}) and ({2+0.5},{0-0.7}) .. (c.-68);
    \end{tikzpicture}
\\
g & h 
\end{array}
\]
~\\
\[
\tag{c}
\arraycolsep=20pt
\begin{array}{cc}
    \begin{tikzpicture}
		\node (X) at (0,2.5){$X$};
	    \node (x1) at (-0.5,2){\small $\bullet$};
        \node (x2) at (+0.5,2){\small $\bullet$};
		\node (na) at (0,0){$\overline{a}$};
        \node (a) at (0,1){$a$};
        \node (b) at (1,0){$b$};
        \node (c) at (2,0){\small$\bullet$};
        \draw [draw=black,dotted] (-0.7,1.8) rectangle (+0.7,2.2);
        \path[thick,->,draw,black]
        (a) edge[bend left] (na)
        (na) edge[bend left] (a)
		(x1) edge (a)
		(x2) edge (a)
		(x1) edge[-,dotted] (x2)
		(b) edge[-,dotted] (c)
        ;
        \draw[->,thick] (b.-112) .. controls ({1-0.5},{0-0.7}) and ({1+0.5},{0-0.7}) .. (b.-68);
		\draw[->,thick] (c.-112) .. controls ({2-0.5},{0-0.7}) and ({2+0.5},{0-0.7}) .. (c.-68);
    \end{tikzpicture}
&
    \begin{tikzpicture}
		\node (X) at (0,2.5){$X$};
	    \node (x1) at (-0.5,2){\small $\bullet$};
        \node (x2) at (+0.5,2){\small $\bullet$};
		\node (na) at (0,0){$b$};
        \node (a) at (0,1){$a$};
        \node (b) at (1,0){$\overline{a}$};
        \node (c) at (2,0){\small$\bullet$};
        \draw [draw=black,dotted] (-0.7,1.8) rectangle (+0.7,2.2);
        \path[thick,->,draw,black]
        (a) edge[bend left] (na)
        (na) edge[bend left] (a)
        (x1) edge (a)
		(x2) edge (a)
		(x1) edge[-,dotted] (x2)
		(b) edge[-,dotted] (c)
        ;
        \draw[->,thick] (b.-112) .. controls ({1-0.5},{0-0.7}) and ({1+0.5},{0-0.7}) .. (b.-68);
		\draw[->,thick] (c.-112) .. controls ({2-0.5},{0-0.7}) and ({2+0.5},{0-0.7}) .. (c.-68);
    \end{tikzpicture}
    \\
g & h 
\end{array}
\]
~\\
\[
\tag{d}
\arraycolsep=20pt
\begin{array}{cc}
    \begin{tikzpicture}
        \node (a) at (-2,0){$a$};
		\node (na) at (-1,0){$\overline{a}$};
        \node (b) at (0,0){$b$};
        \node (c) at (1,0){$c$};

        \node (d) at (2,0){\small$\bullet$};
        \path[thick,->,draw,black]
        (a) edge[] (na)
        (na) edge[] (b)
        (b) edge[bend left] (a)
		(c) edge[-,dotted] (d)
        ;
        \draw[->,thick] (c.-112) .. controls ({1-0.5},{0-0.7}) and ({1+0.5},{0-0.7}) .. (c.-68);
		\draw[->,thick] (d.-112) .. controls ({2-0.5},{0-0.7}) and ({2+0.5},{0-0.7}) .. (d.-68);
    \end{tikzpicture}
&
    \begin{tikzpicture}
        \node (a) at (-2,0){$a$};
		\node (na) at (-1,0){$c$};
        \node (b) at (0,0){$b$};
        \node (c) at (1,0){$\overline{a}$};

        \node (d) at (2,0){\small$\bullet$};
        \path[thick,->,draw,black]
        (a) edge[] (na)
        (na) edge[] (b)
        (b) edge[bend left] (a)
		(c) edge[-,dotted] (d)
        ;
        \draw[->,thick] (c.-112) .. controls ({1-0.5},{0-0.7}) and ({1+0.5},{0-0.7}) .. (c.-68);
		\draw[->,thick] (d.-112) .. controls ({2-0.5},{0-0.7}) and ({2+0.5},{0-0.7}) .. (d.-68);
    \end{tikzpicture}
    \\
g & h 
\end{array}
\]
{\caption{\label{fig:no_2P_1_2} Illustration for Lemma \ref{lem:no_2P_2}.}}
\end{figure}
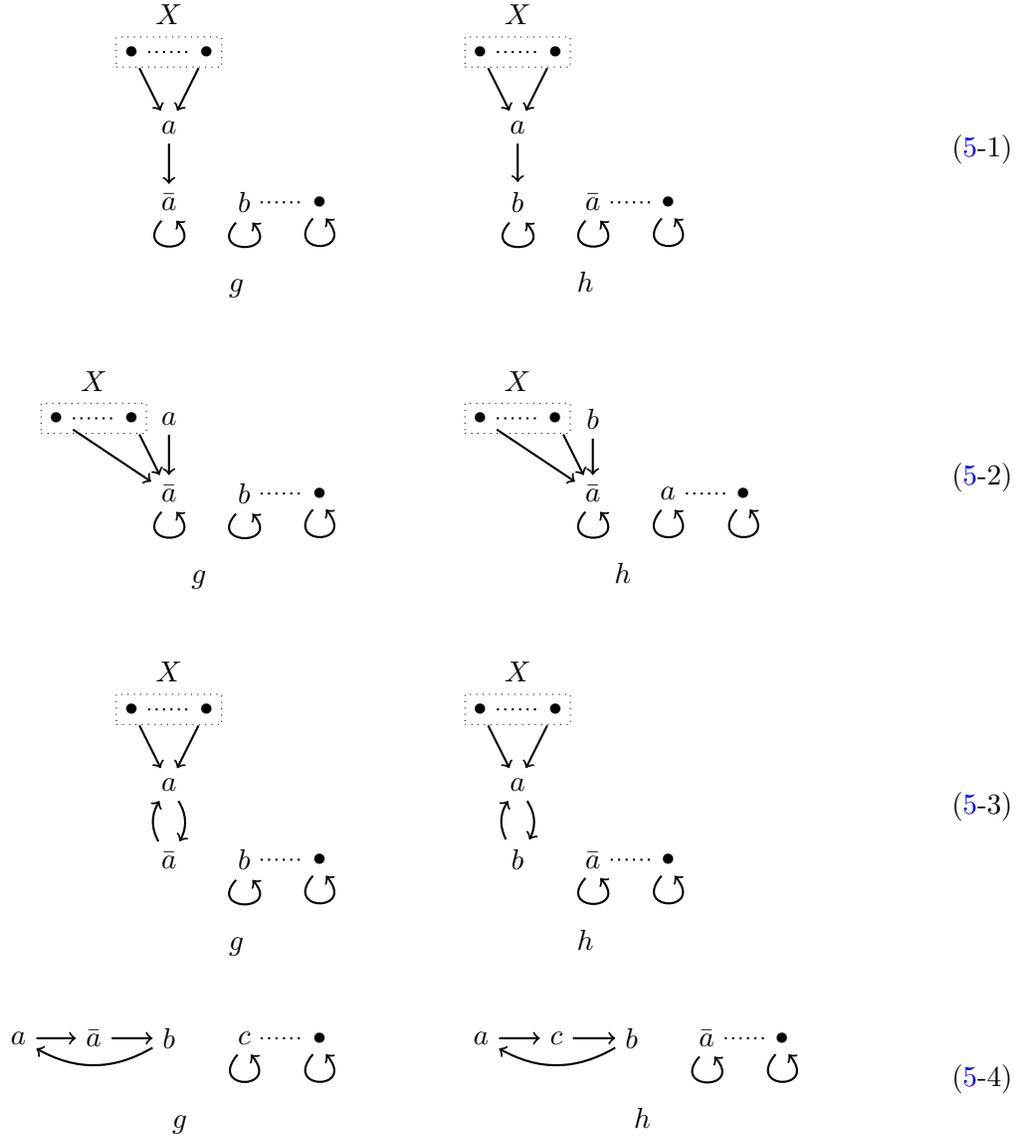

By the two previous lemmas, it remains to treat the case where $f$ does not contain $2P_1$ and has a unique periodic configuration. One can prove that these properties force $|\Delta^+(f)|=1$. However, if $f$ is not constant, there exists $h\sim f$ such that $\A(h)$ and $\A(f)$ do not have the same number of connected components, which implies  $\A(h)\not\sim \A(f)$. This finishes the proof of Theorem~\ref{thm:S_to_A}. 

\begin{lemma}\label{lem:no_2P_1_3}
Let $f\in F(n)$ with $f\neq\cst$ and $n\geq 2$. If $f$ does not contain $2P_1$ and has a unique periodic configuration, then there exists $h \sim f$ such that $\A(h)$ and $\A(f)$ do not have the same number of connected components.
\end{lemma}

\begin{proof}
Suppose that $f$ does not contain $2P_1$. Let $b$ be the unique periodic configuration of $f$, which is thus a fixed point, and let $a$ be a pre-image of $b$. Since $f\neq\cst$ and $f$ does not contain $2P_1$, for all $x\neq a,b$ we have $f(x)=a$. Let $\pi_1,\pi_2$ be permutations of $\B^n$ such that $\pi_1$ sends $a,b$ on $\ZERO,\ONE$, and $\pi_2$ sends $a,b$ on $\ZERO,e_1$; note that $e_1\neq\ONE$ since $n\geq 2$. Let $h^1 = \pi_1 \circ f \circ \pi_1^{-1}$ and $h^2 = \pi_2 \circ f \circ \pi_2^{-1}$. See Figure~\ref{fig:no_2P_1_3} for an illustration. 

\medskip
Let us prove that $\A(h^1)$ has at least two connected components. Indeed, for all $i\in [n]$, we have $f(\ONE+e_i)=\ZERO$, thus $\A(h^1)$ has no arc from $\ONE+e_1$ to $\ONE$. Since $h^1(\ONE)=\ONE$, we deduce that $\A(h^1)$ has a connected component consisting of a single element, $\ONE$. Thus $\A(h^1)$ has indeed at least two connected components. 

\medskip
We now prove that $\A(h^2)$ has a unique connected component. We first prove that $\A(h^2)$ has a path from any $x\neq \ZERO,e_1$ to $\ZERO$. To this end, it is sufficient to prove that $\A(h^2)$ has an arc from any such $x$ to some configuration $y$ with $w(y)<w(x)$. Since $x\neq \ZERO,e_1$, we have $h^2(x)=\ZERO$ and there exists $j\in [n]$ with $x_j=1$. So $w(x+e_j)<w(x)$ and $\A(h^2)$ has an arc from $x$ to $x+e_j$, so we are done. Furthermore, since $h^2(\ZERO)=e_1$, $\A(h^2)$ has an arc from $\ZERO$ to $e_1$. We conclude that $\A(h^2)$ has a path from any configuration to $e_1$, and thus it has a unique connected component.  
\end{proof}

\begin{figure}[h]
\[
\arraycolsep=20pt
\begin{array}{ccc}
    \begin{tikzpicture}
	    \node (x1) at (-0.5,2){\small $\bullet$};
        \node (x2) at (+0.5,2){\small $\bullet$};
		\node (b) at (0,0){$b$};
        \node (a) at (0,1){$a$};
        \path[thick,->,draw,black]
        (a) edge (b)
		(x1) edge (a)
		(x2) edge (a)
		(x1) edge[-,dotted] (x2)
        ;
        \draw[->,thick] (b.-112) .. controls ({0-0.5},{0-0.7}) and ({0+0.5},{0-0.7}) .. (b.-68);
    \end{tikzpicture}
&
    \begin{tikzpicture}
	    \node (x1) at (-0.5,2){\small $\bullet$};
        \node (x2) at (+0.5,2){\small $\bullet$};
		\node (b) at (0,0){$\ONE$};
        \node (a) at (0,1){$\ZERO$};
        \path[thick,->,draw,black]
        (a) edge (b)
		(x1) edge (a)
		(x2) edge (a)
		(x1) edge[-,dotted] (x2)
        ;
        \draw[->,thick] (b.-112) .. controls ({0-0.5},{0-0.7}) and ({0+0.5},{0-0.7}) .. (b.-68);
    \end{tikzpicture}
&
    \begin{tikzpicture}
	    \node (x1) at (-0.5,2){\small $\bullet$};
        \node (x2) at (+0.5,2){\small $\bullet$};
		\node (b) at (0,0){$e_{1}$};
        \node (a) at (0,1){$\ZERO$};
        \path[thick,->,draw,black]
        (a) edge (b)
		(x1) edge (a)
		(x2) edge (a)
		(x1) edge[-,dotted] (x2)
        ;
        \draw[->,thick] (b.-112) .. controls ({0-0.5},{0-0.7}) and ({0+0.5},{0-0.7}) .. (b.-68);
    \end{tikzpicture}
\\[3mm]
f & h^1 & h^2
\end{array}
\]
{\caption{\label{fig:no_2P_1_3} Illustration for Lemma \ref{lem:no_2P_1_3}.}}
\end{figure}
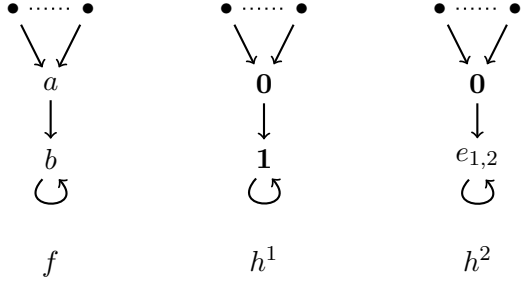

\begin{remark}\label{rem:counter_example_A_to_S}
For $n=1$, there is a unique Boolean network $f\neq\cst,\id$ : the negation, which thus does not satisfy the conclusion of Theorem~\ref{thm:S_to_A}. For $n=2$, there is, up to isomorphism, exactly one Boolean network $f\neq\cst,\id$ which does not satisfy the conditions of Lemmas~\ref{lem:2P_1}, \ref{lem:no_2P_2} and \ref{lem:no_2P_1_3}, namely, the one with a fixed point and a limit cycle of length 3. This Boolean network does not satisfy the conclusion of Theorem~\ref{thm:S_to_A}; see Figure~\ref{fig:counter_example_A_to_S}. Thus there are exactly two Boolean networks $f\neq\cst,\id$ which do not satisfy the conclusion of Theorem~\ref{thm:S_to_A}. 
\end{remark}

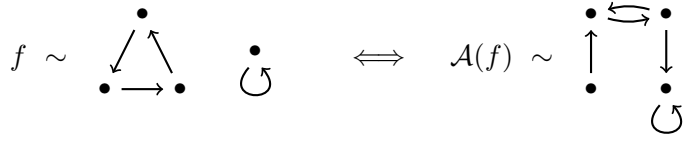
\begin{figure}[h]
\[
f~\sim 
\begin{array}{c}
    \begin{tikzpicture}
    	\draw[white] (-0.2,-0.2) rectangle (2.3,1.2);
	    \node (a) at (0,0){\small $\bullet$};
        \node (b) at (1,0){\small $\bullet$};
		\node (c) at (0.5,1){\small $\bullet$};
        \node (d) at (2,0.5){\small $\bullet$};
        \path[thick,->,draw,black]
        (a) edge (b)
		(b) edge (c)
		(c) edge (a)
        ;
        \draw[->,thick] (d.-112) .. controls ({2-0.5},{0.5-0.7}) and ({2+0.5},{0.5-0.7}) .. (d.-68);
    \end{tikzpicture}
\end{array}
\quad\iff\quad
\A(f)~\sim 
\begin{array}{c}
    \begin{tikzpicture}
        \draw[white] (-0.2,-0.7) rectangle (1.3,1.7);
	    \node (a) at (0,0){\small $\bullet$};
        \node (b) at (0,1){\small $\bullet$};
		\node (c) at (1,1){\small $\bullet$};
        \node (d) at (1,0){\small $\bullet$};
        \path[thick,->,draw,black]
        (a) edge (b)
		(b) edge[bend right=15] (c)
		(c) edge[bend right=15] (b)
		(c) edge (d)
        ;
        %\draw[->,thick] (d.-112) .. controls ({1-0.5},{0-0.7}) and ({1+0.5},{0-0.7}) .. (d.-68);
    \end{tikzpicture}
\end{array}
\]
{\caption{\label{fig:counter_example_A_to_S} Illustration for Remark \ref{rem:counter_example_A_to_S}.}}
\end{figure}

%%%%%%%%%%%%%%%%%%%%%%%%%%%%%%%%%%%%%%%%%%%%%%%%%%%%%%%%%%%%%%%%%%%
\section{Few and small attractors}\label{sec:small_att}
%%%%%%%%%%%%%%%%%%%%%%%%%%%%%%%%%%%%%%%%%%%%%%%%%%%%%%%%%%%%%%%%%%%

%%%%%%%%%%%%%%%%%%%%%%%%%%%%%%%%%%%%%%%%%%%%%%%%%%%%%%%%%%%%%%%%%%%
\subsection{Asynchronous convergence toward fixed points}
%%%%%%%%%%%%%%%%%%%%%%%%%%%%%%%%%%%%%%%%%%%%%%%%%%%%%%%%%%%%%%%%%%%

In this subsection we prove the first assertion in Theorem \ref{thm:small_att}: if $f$ has at least one fixed point, there exists $h\sim f$ such that $\A(h)$ has $\fp(f)$ attractors. Since $\fp(f)=\fp(h)$, this is equivalent to say that $\A(h)$ has a path from any configuration to a fixed point of $h$. We actually prove something stronger; the precise statement, Theorem~\ref{thm:FP} below, needs some definitions. An arc from $x$ to $y$ in $\A(f)$ is \EM{increasing} if $x<y$ and \EM{decreasing} otherwise. A path is \EM{decreasing} if all its arcs are decreasing. Hence, a decreasing path from $x$ to $y$ is of length $d(x,y)\leq n$. A path is  \EM{almost decreasing} if all its arc, except possibly one, are decreasing. Hence, an almost-decreasing path from $x$ to $y$ is of length at most $d(x,y)+1\leq n+1$.  

\begin{theorem}\label{thm:FP}
Let $f\in F(n)$ with at least one fixed point. There exists $h\sim f$ such that $\A(h)$ has an almost decreasing path from any configuration  to $\FP(h)$. 
\end{theorem}

To illustrate the proof technique, suppose that $f$ has a unique fixed point, say $z$, and no other periodic configurations, that is $f$ is a nilpotent function. Then $\S(f)$ is an in-tree plus a loop on the root $z$. Let $\preceq$ be a topological order of this is in-tree, that is, a total order on $\B^n$ such that $f(x)\prec x$ for all $x\neq z$. There is then a permutation $\pi$ of $\B^n$ such that $x\preceq y$ implies $w(\pi(x))\leq w(\pi(y))$, we say that such a permutation is monotone. Since $z$ is the $\preceq$-minimal element, we have $\pi(z)=\ZERO$ by monotonicity, and thus $\ZERO$ is the unique fixed point of $h=\pi\circ f\circ\pi^{-1}$. Furthermore, for every $x\neq \ZERO$, we have $\pi^{-1}(x)\neq z$, thus $f(\pi^{-1}(x))\prec \pi^{-1}(x)$ so that $w(h(x))\leq w(x)$ by monotonicity. Since $x$ is not a fixed point of $h$, we deduce that $h_i(x)<x_i$ for some $i\in [n]$, that is, $\A(h)$ has a decreasing arc starting from~$x$. It follows that $\A(h)$ has a decreasing path from any configuration to $\ZERO$. See Figure~\ref{fig:nilpotent} for an illustration. 

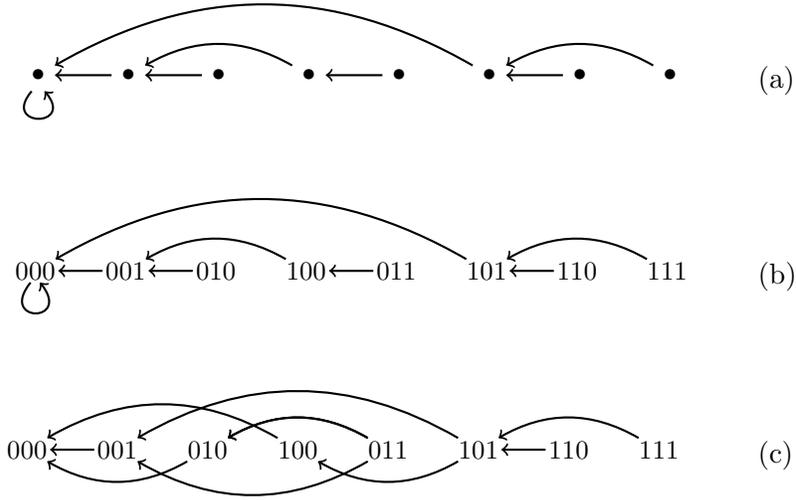
\begin{figure}
\[
\begin{array}{cl}
	\begin{array}{c}
    \begin{tikzpicture}
	    \node (000) at ({0*1.2},0){\small $\bullet$};
        \node (001) at ({1*1.2},0){\small $\bullet$};
        \node (010) at ({2*1.2},0){\small $\bullet$};
        \node (100) at ({3*1.2},0){\small $\bullet$};
        \node (011) at ({4*1.2},0){\small $\bullet$};
        \node (101) at ({5*1.2},0){\small $\bullet$};
        \node (110) at ({6*1.2},0){\small $\bullet$};
        \node (111) at ({7*1.2},0){\small $\bullet$};
        \path[thick,->,draw,black]
        (001) edge (000)
        (010) edge (001)
        (100) edge[bend right=30] (001)
        (011) edge (100)
        (101) edge[bend right=30] (000)
        (101) edge[white,bend left=30] (000)%for layout only
        (110) edge (101)
        (111) edge[bend right=30] (101)
        ;
        \draw[->,thick] (000.-112) .. controls ({0-0.5},{0-0.7}) and ({0+0.5},{0-0.7}) .. (000.-68);

    \end{tikzpicture}
    \end{array}
&\quad\textrm{(a)}\\
	\begin{array}{c}
    \begin{tikzpicture}
	    \node[inner sep=1] (000) at ({0*1.2},0){\small $000$};
        \node[inner sep=1] (001) at ({1*1.2},0){\small $001$};
        \node[inner sep=1] (010) at ({2*1.2},0){\small $010$};
        \node[inner sep=1] (100) at ({3*1.2},0){\small $100$};
        \node[inner sep=1] (011) at ({4*1.2},0){\small $011$};
        \node[inner sep=1] (101) at ({5*1.2},0){\small $101$};
        \node[inner sep=1] (110) at ({6*1.2},0){\small $110$};
        \node[inner sep=1] (111) at ({7*1.2},0){\small $111$};
        \path[thick,->,draw,black]
        (001) edge (000)
        (010) edge (001)
        (100) edge[bend right=30] (001)
        (011) edge (100)
        (101) edge[bend right=30] (000)
		(101) edge[white,bend left=30] (000)%for layout only
        (110) edge (101)
        (111) edge[bend right=30] (101)
        ;
        \draw[->,thick] (000.-112) .. controls ({0-0.5},{0-0.7}) and ({0+0.5},{0-0.7}) .. (000.-68);
    \end{tikzpicture}
    \end{array}
&\quad\textrm{(b)}\\
	\begin{array}{c}
    \begin{tikzpicture}
	    \node[inner sep=1] (000) at ({0*1.2},0){\small $000$};
        \node[inner sep=1] (001) at ({1*1.2},0){\small $001$};
        \node[inner sep=1] (010) at ({2*1.2},0){\small $010$};
        \node[inner sep=1] (100) at ({3*1.2},0){\small $100$};
        \node[inner sep=1] (011) at ({4*1.2},0){\small $011$};
        \node[inner sep=1] (101) at ({5*1.2},0){\small $101$};
        \node[inner sep=1] (110) at ({6*1.2},0){\small $110$};
        \node[inner sep=1] (111) at ({7*1.2},0){\small $111$};
        \path[thick,->,draw,black]
		(101) edge[white,bend left=30] (001)%for layout only
        (001) edge (000)
        (010) edge[bend left=30] (000)
        (100) edge[bend right=30] (000)
        (011) edge[bend right=30] (010)
        
        (011) edge[bend right=30] (010)
        (011) edge[bend left=30] (001)
        
        (101) edge[bend left=30] (100)
        (101) edge[bend right=30] (001)

        (110) edge (101)
        (111) edge[bend right=30] (101)
        ;
    \end{tikzpicture}
    \end{array}
&\quad\textrm{(c)}
\end{array}
\]
{\caption{\label{fig:nilpotent} (a) Unlabelled synchronous dynamics of a nilpotent function $f\in F(3)$ with a topological order, increasing from left to right. (b) Function $h\sim f$ obtained with a monotone permutation. (c) Decreasing transitions of $\A(h)$, showing that there is a decreasing path from any configuration to~$000$.}}
\end{figure}

\smallskip
With a similar but more technical argument, we can obtain a more general result of convergence by decreasing paths, Lemma~\ref{decreasing_lemma} below. The statement is kept as general as possible, in order to be used later to treat the case where $f$ has no fixed point, and it involves several definitions. 

\medskip
Let $f,h\in F(n)$ and $X\subseteq \B^n$. We say that $h$ is \EM{$X$-converging (by decreasing paths)} if $\A(h)$ has a decreasing path from any configuration to $X\cup\FP(h)$. We say that $h$ is \EM{robustly $X$-converging} if $\pi\circ h\circ \pi^{-1}$ is $X$-converging for every permutation $\pi$ of $\B^n$ acting as the identity on $\overline{X}$. Taking $\pi=\id$, we see that if $h$ is robustly $X$-converging, then $h$ is $X$-converging. We write \EM{$h\sim_X f$} to means that $h\sim f$ and that, for all $x\in X$, 
\begin{itemize}
\item
$h(x)= f(x)$ if $f(x)\in X$;
\item
$h(x)\not\in X$ otherwise. 
\end{itemize}
In other words, $h\sim_X f$ if and only if $h\sim f$ and the subgraph of $\S(h)$ induced by $X$ is equal to the subgraph of $\S(f)$ induced by $X$. Finally, we say that $X$ is an \EM{down set} if $x\in X$ and $y\leq x$ implies $y\in X$; and $X$ is an \EM{upper set} if $\overline{X}$ is a down set. 

\begin{lemma}\label{decreasing_lemma}
Let $f\in F(n)$ and let $X\subseteq \B^n$ be a non-empty down set with $f(\overline{X})\neq \overline{X}$. Suppose that, for all $1<\ell<n$, $\overline{X}$ contains at least ${n-1\choose \ell-1}+1$ configurations of weight $\ell$. Then there exists a robustly $X$-converging function $h\sim_X f$. 
\end{lemma}

This lemma will be proved later (in Section \ref{sec:decreasing_lemma}); for the moment we show that it easily implies Theorem~\ref{thm:FP}. 

\begin{proof}[{\bf Proof of Theorem~\ref{thm:FP} (assuming Lemma~\ref{decreasing_lemma})}]
We consider two cases.
\begin{enumerate}
\item
Suppose that $f$ is not a permutation. Since $f$ has a fixed point, there exists $g\sim f$ such that $g(\ZERO)=\ZERO$. Let $X=\{\ZERO\}$, which is a down set. Since $g$ is not a permutation, we have $g(\overline{X})\neq \overline{X}$. Furthermore, for all $1<\ell<n$, the number of $x\in \overline{X}$ with $w(x)=\ell$ is ${n\choose \ell}>{n-1\choose \ell-1}$. Thus, by Lemma~\ref{decreasing_lemma}, there exists a robustly $X$-converging function $h\sim_X g$. Hence, $\A(h)$ has a decreasing path from any configuration to $\{\ZERO\}\cup\FP(h)$. Since $h\sim_X g$ we have $h(\ZERO)=g(\ZERO)=\ZERO$ and thus $\A(h)$ has a decreasing path from any configuration to~$\FP(h)$. 
\item 
Suppose that $f$ is a permutation. If $f=\id$ the result is obvious so suppose that $f\neq\id$. Since $f$ has a fixed point and $f\neq\id$, there exists $g\sim f$ such that $g(\ZERO)=\ZERO$ and $g_1(e_1)=0$. Let $X=\{\ZERO,e_1\}$, which is a down set. Since $g(e_1)\neq e_1$ and $g(e_1)\neq\ZERO=g(\ZERO)$ (because $g$ is a permutation), we have $g(e_1)\in\overline{X}$ and thus $g(\overline{X})\neq \overline{X}$. Furthermore, for all $1<\ell<n$, the number of $x\in \overline{X}$ with $w(x)=\ell$ is ${n\choose \ell}>{n-1\choose \ell-1}$.
Thus, by Lemma~\ref{decreasing_lemma}, there exists a robustly $X$-converging function $h\sim_X g$. Hence, $\A(h)$ has a decreasing path from any configuration to $\{\ZERO,e_1\}\cup\FP(h)$, and since $h\sim_X g$ we have $h(\ZERO)=g(\ZERO)=\ZERO$. We then consider two subcases:
\begin{enumerate}
\item[2.1.] 
Suppose that $h_1(e_1)=0$. Then $\A(h)$ has a decreasing arc from $e_1$ to $\ZERO$. Thus, $\A(h)$ has a decreasing path from any configuration to $\FP(h)$.
\item[2.2.] 
Suppose that $h_1(e_1)=1$. Then $h(e_1)\not\in X$: since $h\sim _X g$, if $g(e_1)\in X$ then $h_1(e_1)=g_1(e_1)=0$, a contradiction; so $g(e_1)\not\in X$ and thus $h(e_1)\not\in X$. Let $\pi=(\ZERO\leftrightarrow e_1)$ and $h'=\pi\circ h\circ \pi$. Since $h$ is robustly $X$-converging and $\pi$ acts as the identity on $\overline{X}$, $h'$ is $X$-converging: $\A(h')$ has a decreasing path from any configuration to $\{\ZERO,e_1\}\cup \FP(h')$. Furthermore, $h'(e_1)=\pi(h(\pi(e_1)))=\pi(h(\ZERO))=\pi(\ZERO)=e_1$ and $h'(\ZERO)=\pi(h(\pi(\ZERO)))=\pi(h(e_1))=h(e_1)$, where the last equality holds since $h(e_1)\not\in X$. So $h'_1(\ZERO)=h_1(e_1)=1$. Thus, $\A(h')$ has an increasing arc from $\ZERO$ to $e_1$ and we deduce that $\A(h')$ has an almost decreasing path from any configuration to $\FP(h')$.
\end{enumerate} 
\end{enumerate} 
\end{proof}

%%%%%%%%%%%%%%%%%%%%%%%%%%%%%%%%%%%%%%%%%%%%%%%%%%%%%%%%%%%%%%%%%%%
\subsection{Asynchronous convergence toward a unique small attractor}\label{sec:one_att}
%%%%%%%%%%%%%%%%%%%%%%%%%%%%%%%%%%%%%%%%%%%%%%%%%%%%%%%%%%%%%%%%%%%
 
In this subsection we prove the second assertion of Theorem \ref{thm:small_att} for $n\geq 5$, with additional information on the paths reaching the unique small asynchronous attractors; the case $n\leq 4$ is treated in Appendix \ref{sec:n4}. 
 
\begin{theorem}\label{thm:A4}
Let $f\in F(n)$ without fixed point, $n\geq 5$. There exists $h\sim f$ such that $\A(h)$ has an attractor $A$ of size $\leq 4$ and $\A(h)$ has an almost decreasing path from any configuration to~$A$.
\end{theorem} 

The proof technique follows that of Theorem~\ref{thm:FP}. To explain the analogy, let us first give a rough summary of the proof of Theorem~\ref{thm:FP}. First, we obtain $g\sim f$ by ``plugging'' in $X$ a ``pattern'', which is either a fixed point $\ZERO$  (case 1) or a fixed point $\ZERO$ and $e_1$ with $g_1(e_1)=0$ (case 2);  $g$ is chosen such that $g(\overline{X})\neq \overline{X}$ and $\A(g)$ has a decreasing path from any configuration in $X$ to a fixed point. We then use Lemma~\ref{decreasing_lemma} to obtain a robustly $X$-converging function $h\sim_X g$. If $\A(h)$ still has a decreasing path from any configuration in $X$ to a fixed point then we are done (case 1 and 2.1). Otherwise, we apply a ``correction'': we swap configurations in $X$ to obtain $h'\sim h$ so that $\A(h')$ has an almost decreasing path from any configuration in $X$ to a fixed point (case 2.2). Since $h$ is robustly $X$-converging, $h'$ is still $X$-converging and we are done. 

\medskip
To prove Theorem~\ref{thm:A4} we proceed similarly, except that we obtain $g\sim f$ by ``plugging'' in $X$ more sophisticated ``patterns''. Any such pattern forces $\A(g)$ to contain an attractor $A\subseteq X$ of size at most $4$ and an almost decreasing path from any configuration in $X$ to $A$. Furthermore, the pattern is chosen so that $g(\overline{X})\neq \overline{X}$. We then use Lemma~\ref{decreasing_lemma} to obtain a robustly $X$-converging function $h\sim_X g$. Then $A$ is still an attractor of $\A(h)$, and if $\A(h)$ still has an almost decreasing path from any configuration in $X$ to $A$ then we are done. Otherwise, we apply a ``correction'': we permute configurations in $X$ to obtain $h'\sim h$ so that $\A(h')$ has an almost decreasing path from any configuration in $X$ to an attractor $A'$ with the same size as $A$. Since $h$ is robustly $X$-converging, $h'$ is still $X$-converging and we are done. Actually, this works well unless $f^2=\id$, and for that case we have a direct proof (Lemma~\ref{lem:f^2=id}).

\medskip
There are 16 possibles patterns: patterns $\EM{\P_1},\EM{\P_2},\EM{\P_3},\EM{\P_4},\EM{\P_5},\EM{\P_6}$ are described in Figure \ref{fig:closed_patterns} and called \EM{closed patterns}; for $a=0,1$, patterns $\EM{\P^a_1},\EM{\P^a_2},\EM{\P^a_3},\EM{\P^a_4},\EM{\P^a_5}$ are described in Figures \ref{fig:0-open_patterns} and \ref{fig:1-open_patterns} and called \EM{$a$-open patterns}. A closed pattern $\P_k$ is a digraph with vertices in $\B^n$, and we say that $f$ \EM{contains} $\P_k$ if $\P_k$ is a subgraph of $\S(f)$. An $a$-open pattern $\P^a_k$ is a digraph with vertices in $\B^n$ plus a \EM{special arc}, which is $x\to a*$ for some configuration $x$. We then say that $x$ is the \EM{special configuration} of the pattern and that $\S(f)$ \EM{contains} $x\to a*$ if $f_1(x)=a$. Then $f$ \EM{contains} a pattern $\P^a_k$ if $\S(f)$  contains every arc of $\P^a_k$, including the special arc.

\medskip
Each pattern $\P$ is associated with an \EM{asynchronous pattern} $\EM{\A(\P)}$, described in Figures \ref{fig:closed_patterns}, \ref{fig:0-open_patterns} and \ref{fig:1-open_patterns}; we easily check that if $f$ contains $\P$ then $\A(\P)$ is a subgraph of $\A(f)$. This asynchronous pattern $\A(\P)$ has a unique attractor, denoted $\EM{A(\P)}$, which is of size at most $4$ (represented in blue in Figures \ref{fig:closed_patterns}, \ref{fig:0-open_patterns} and \ref{fig:1-open_patterns}). The vertex set of $\A(\P)$ (which does not include the configurations $a*$ of the $a$-open patterns) is denoted \EM{$X(\P)$}. Note that $\A(\P)$ has an almost decreasing path from any configuration in $X(\P)$ to $A(\P)$. Note also that $X(\P)$ is a down set. 

%%%% Closed patterns
\begin{figure}[h]
\[
\def\x{1.3}
\def\z{0.7}
\def\zz{0.6}
\begin{array}{rcl}
	%C2
	\P_1
	\begin{array}{c}
    \begin{tikzpicture}
     	\draw[white] (-0.5,0) rectangle (3.5,1);	
        \node[inner sep=1] (0) at (0,0.5){\small $\ZERO$};
		\node[inner sep=1] (1) at (1,0.5){\small $e_1$};
		\path[thick,->]
        (0) edge[bend left=15] (1)
        (1) edge[bend left=15] (0)
        ;
    \end{tikzpicture}
    \end{array}
    &&
    \A(\P_1)~
    \begin{array}{c}
    \begin{tikzpicture}
    	\draw[white] (-0.3,-0.3) rectangle ({\x+\z+0.3},{\x+0.3});
        \node[inner sep=1,blue] (1) at (\x,{\x/2}){\small $e_1$};
        \node[inner sep=1,blue] (0) at (0,{\x/2}){\small $\ZERO$};
        \path[thick,->]
        (1) edge[<->,blue] (0)
        ;
    \end{tikzpicture}
    \end{array}
    \\
    %C4
   	\P_2
    \begin{array}{c}
    \begin{tikzpicture}
    	\draw[white] (-0.5,-0.5) rectangle (3.5,1.5);
        \node[inner sep=1] (2) at (0,1){\small $e_2$};
        \node[inner sep=1] (12) at (1,1){\small $e_{1,2}$};
        \node[inner sep=1] (1) at (1,0){\small $e_1$};
        \node[inner sep=1] (0) at (0,0){\small $0$};
        \path[thick,->]
        (0) edge (1)
        (1) edge (12)
        (12) edge (2)
        (2) edge (0)
        ;
    \end{tikzpicture}
    \end{array}
    &&
    \A(\P_2)~
    \begin{array}{c}
    \begin{tikzpicture}
    	\draw[white] (-0.3,-0.3) rectangle ({\x+\z+0.3},{\x+0.3});
        \node[inner sep=1,blue] (1) at (\x,0){\small $e_1$};
        \node[inner sep=1,blue] (2) at (0,\x){\small $e_2$};
        \node[inner sep=1,blue] (0) at (0,0){\small $\ZERO$};
        \node[inner sep=1,blue] (12) at (\x,\x){\small $e_{1,2}$};
        \path[thick,->,blue]
        (0) edge (1)
        (1) edge(12)
        (12) edge (2)
		(2) edge (0)
        ;
    \end{tikzpicture}
    \end{array}
	\\
	%->C3
	\P_3
    \begin{array}{c}
    \begin{tikzpicture}
 	   \draw[white] (-0.5,-0.5) rectangle (3.5,1.5);	
        \node[inner sep=1] (1) at (0,0){\small $e_1$};
        \node[inner sep=1] (2) at (0,1){\small $e_2$};
        \node[inner sep=1] (0) at (1,0.5){\small $\ZERO$};
        \node[inner sep=1] (12) at (2,0.5){\small $e_{1,2}$};
        \path[thick,->]
        (1) edge (0)
        (2) edge (0)
        (0) edge (12)
        (12) edge[bend left=30] (1)
        ;
    \end{tikzpicture}
    \end{array}
    %A->C3
    &&
    \A(\P_3)~
    \begin{array}{c}
    \begin{tikzpicture}
    	\draw[white] (-0.3,-0.3) rectangle ({\x+\z+0.3},{\x+0.3});
        \node[inner sep=1,blue] (1) at (\x,0){\small $e_1$};
        \node[inner sep=1,blue] (2) at (0,\x){\small $e_2$};
        \node[inner sep=1,blue] (0) at (0,0){\small $\ZERO$};
        \node[inner sep=1] (12) at (\x,\x){\small $e_{1,2}$};
        \path[thick,->]
        (1) edge[<->,blue] (0)
        (2) edge[<->,blue] (0)
        (12) edge (1)
        ;
    \end{tikzpicture}
    \end{array}
    \\
    %2C3
   	\P_4
    \begin{array}{c}
    \begin{tikzpicture}
    	\draw[white] (-0.5,-0.5) rectangle (3.5,1.5);
        \node[inner sep=1] (2) at (0,1){\small $e_2$};
        \node[inner sep=1] (0) at (1,1){\small $\ZERO$};
        \node[inner sep=1] (12) at (2,1){\small $e_{1,2}$};
        \node[inner sep=1] (13) at (2,0){\small $e_{1,3}$};
        \node[inner sep=1] (1) at (1,0){\small $e_1$};
        \node[inner sep=1] (3) at (0,0){\small $e_3$};
        \path[thick,->]
        (2) edge (0)
        (0) edge (12)
        (12) edge[bend left=35] (2)
        (13) edge (1)
        (1) edge (3)
        (3) edge[bend left=35] (13)
        ;
    \end{tikzpicture}
    \end{array}
    &&
    %A2C3
    \A(\P_4)~
    \begin{array}{c}
    \begin{tikzpicture}
    	\draw[white] (-0.3,-0.3) rectangle ({\x+\z+0.3},{\x+0.3});
        \node[inner sep=1,blue] (2) at (0,\x){\small $e_2$};
        \node[inner sep=1,blue] (0) at (0,0){\small $\ZERO$};
        \node[inner sep=1] (12) at (\x,\x){\small $e_{1,2}$};
        \node[inner sep=1,blue] (13) at ({\x+\z},{0+\zz}){\small $e_{1,3}$};
        \node[inner sep=1,blue] (1) at (\x,0){\small $e_1$};
        \node[inner sep=1] (3) at (\z,\zz){\small $e_3$};
        \path[thick,->]
        (2) edge[<->,blue] (0)
        (0) edge[<->,blue] (1)
        (12) edge (2)
        (13) edge[<->,blue] (1)
        (3) edge (13)
        ;
    \end{tikzpicture}
    \end{array}
    \\
	%->C5
	\P_5
    \begin{array}{c}
    \begin{tikzpicture}
    	\draw[white] (-0.5,-0.5) rectangle (3.5,1.5);
        \node[inner sep=1] (2) at (0,1){\small $e_2$};
        \node[inner sep=1] (0) at (1,1){\small $\ZERO$};
        \node[inner sep=1] (12) at (2,1){\small $e_{1,2}$};
        \node[inner sep=1] (13) at (2,0){\small $e_{1,3}$};
        \node[inner sep=1] (1) at (1,0){\small $e_1$};
        \node[inner sep=1] (3) at (0,0){\small $e_3$};
        \path[thick,->]
   	  	(2) edge (0)
        (0) edge (12)
        (12) edge (13)
        (13) edge (1)
        (1) edge (3)
        (3) edge (0)
        ;
       \end{tikzpicture}
    \end{array}
    &&
    %A->C5
    \A(\P_5)~
    \begin{array}{c}
    \begin{tikzpicture}
    	\draw[white] (-0.3,-0.3) rectangle ({\x+\z+0.3},{\x+0.3});
        \node[inner sep=1,blue] (2) at (0,\x){\small $e_2$};
        \node[inner sep=1,blue] (0) at (0,0){\small $\ZERO$};
        \node[inner sep=1] (12) at (\x,\x){\small $e_{1,2}$};
        \node[inner sep=1,blue] (13) at ({\x+\z},{0+\zz}){\small $e_{1,3}$};
        \node[inner sep=1,blue] (1) at (\x,0){\small $e_1$};
        \node[inner sep=1] (3) at (\z,\zz){\small $e_3$};
        \path[thick,->]
   	  	(2) edge[<->,blue] (0)
        (0) edge[<->,blue] (1)
        (12) edge (1)
        (13) edge[<->,blue] (1)
        (1) edge[<->,blue] (0)
        (3) edge (0)
        ;
       \end{tikzpicture}
    \end{array}
	\\
	%C6
	\P_6
    \begin{array}{c}
    \begin{tikzpicture}
	    \draw[white] (-0.5,-0.5) rectangle (3.5,1.5);
        \node[inner sep=1] (2) at (0,1){\small $e_2$};
        \node[inner sep=1] (0) at (1,1){\small $\ZERO$};
        \node[inner sep=1] (12) at (2,1){\small $e_{1,2}$};
        \node[inner sep=1] (13) at (2,0){\small $e_{1,3}$};
        \node[inner sep=1] (1) at (1,0){\small $e_1$};
        \node[inner sep=1] (3) at (0,0){\small $e_3$};
        \path[thick,->]
        (2) edge (0)
        (0) edge (12)
        (12) edge (13)
        (13) edge (1)
        (1) edge (3)
        (3) edge (2)
        ;
    \end{tikzpicture}
    \end{array}
    &&
    %AC6
    \A(\P_6)~
    \begin{array}{c}
    \begin{tikzpicture}
    	\draw[white] (-0.3,-0.3) rectangle ({\x+\z+0.3},{\x+0.3});
        \node[inner sep=1,blue] (2) at (0,\x){\small $e_2$};
        \node[inner sep=1,blue] (0) at (0,0){\small $\ZERO$};
        \node[inner sep=1] (12) at (\x,\x){\small $e_{1,2}$};
        \node[inner sep=1,blue] (13) at ({\x+\z},{0+\zz}){\small $e_{1,3}$};
        \node[inner sep=1,blue] (1) at (\x,0){\small $e_1$};
        \node[inner sep=1] (3) at (\z,\zz){\small $e_3$};
        \path[thick,->]
        (2) edge[<->,blue] (0)
        (0) edge[<->,blue] (1)
        (12) edge (1)
        (13) edge[<->,blue] (1)
        (3) edge (0)
        ;
    \end{tikzpicture}
    \end{array}
\end{array}
\]
{\caption{\label{fig:closed_patterns} Closed patterns.}}
\end{figure}
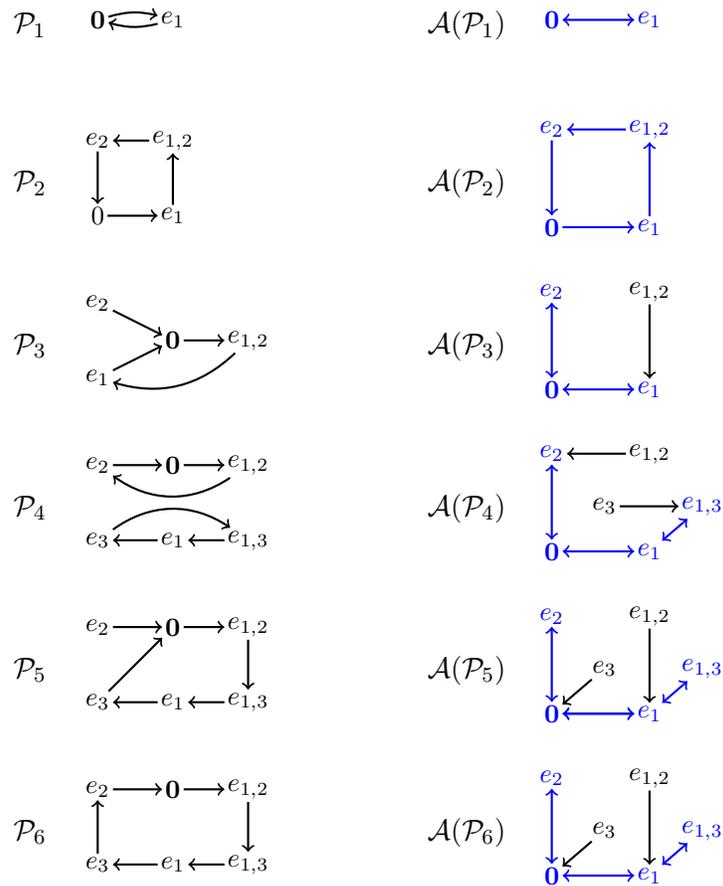

%%%% 0-open patterns
\begin{figure}[h]
\[
\def\x{1.3}
\def\z{0.7}
\def\zz{0.6}
\begin{array}{lclcl}
	%C2+P20
	\P^0_{1}
	\begin{array}{c}
    \begin{tikzpicture}
 	   \draw[white] (-0.5,-0.5) rectangle (3.5,1.5);	
        \node[inner sep=1] (0) at (0,1){\small $\ZERO$};
        \node[inner sep=1] (2) at (1,1){\small $e_2$};
        \node[inner sep=1] (12) at (0,0){\small $e_{1,2}$};
        \node[inner sep=1] (1) at (1,0){\small $e_1$};
        \node[inner sep=1] (0x) at (2,0){\small $0*$};
        \path[thick,->]
        (0) edge (2)
        (2) edge[bend left=35] (0)
        (12) edge (1)
        (1) edge (0x)
        ;
    \end{tikzpicture}
    \end{array}
    &&
    %AC2+P20
    \A(\P^0_{1})~
    \begin{array}{c}
    \begin{tikzpicture}
    	\draw[white] (-0.3,-0.3) rectangle ({\x+\z+0.3},{\x+0.3});
        \node[inner sep=1] (1) at (\x,0){\small $e_1$};
        \node[inner sep=1,blue] (2) at (0,\x){\small $e_2$};
        \node[inner sep=1,blue] (0) at (0,0){\small $\ZERO$};
        \node[inner sep=1] (12) at (\x,\x){\small $e_{1,2}$};
        \path[thick,->]
        (0) edge[<->,blue] (2)
        (12) edge (1)
		(1) edge (0)
        ;
    \end{tikzpicture}
    \end{array}
    \\
    %C3+P30
   	\P^0_{2}
    \begin{array}{c}
    \begin{tikzpicture}
    	\draw[white] (-0.5,-0.5) rectangle (3.5,1.5);
        \node[inner sep=1] (13) at (0,0){\small $e_{1,3}$};
        \node[inner sep=1] (3) at (1,0){\small $e_3$};
        \node[inner sep=1] (1) at (2,0){\small $e_1$};
        \node[inner sep=1] (0x) at (3,0){\small $0*$};
        \node[inner sep=1] (2) at (0,1){\small $e_{2}$};
        \node[inner sep=1] (0) at (1,1){\small $\ZERO$};
        \node[inner sep=1] (23) at (2,1){\small $e_{2,3}$};
        \path[thick,->]
        (13) edge (3)
        (3) edge (1)
        (1) edge (0x)
        (2) edge (0)
        (0) edge (23)
        (23) edge[bend left=35] (2)
        ;
    \end{tikzpicture}
    \end{array}
    &&
    %AC3+P30
    \A(\P^0_{2})~
    \begin{array}{c}
    \begin{tikzpicture}
    	\draw[white] (-0.3,-0.3) rectangle ({\x+\z+0.3},{\x+0.3});
        \node[inner sep=1,blue] (2) at (0,\x){\small $e_2$};
        \node[inner sep=1,blue] (0) at (0,0){\small $\ZERO$};
		\node[inner sep=1] (23) at (\z,{\x+\zz}){\small $e_{2,3}$};
        \node[inner sep=1,blue] (13) at ({\x+\z},{0+\zz}){\small $e_{1,3}$};
        \node[inner sep=1] (1) at (\x,0){\small $e_1$};
        \node[inner sep=1,blue] (3) at (\z,\zz){\small $e_3$};
        \path[thick,->]
        (2) edge[<->,blue] (0)
        (0) edge[<->,blue] (3)
		(3) edge[<->,blue] (13)
        (23) edge (2)
        (1) edge (0)
        ;
    \end{tikzpicture}
    \end{array}
    \\
    %C4+P30
   	\P^0_{3}
    \begin{array}{c}
    \begin{tikzpicture}
    	\draw[white] (-0.5,-0.5) rectangle (3.5,1.5);
        \node[inner sep=1] (13) at (0,0){\small $e_{1,3}$};
        \node[inner sep=1] (3) at (1,0){\small $e_3$};
        \node[inner sep=1] (1) at (2,0){\small $e_1$};
        \node[inner sep=1] (0x) at (3,0){\small $0*$};
        \node[inner sep=1] (2) at (0,1){\small $e_{2}$};
        \node[inner sep=1] (0) at (1,1){\small $\ZERO$};
        \node[inner sep=1] (23) at (2,1){\small $e_{2,3}$};
        \node[inner sep=1] (12) at (3,1){\small $e_{1,2}$};
        \path[thick,->]
        (13) edge (3)
        (3) edge (1)
        (1) edge (0x)
        (2) edge (0)
        (0) edge (23)
        (23) edge (12)
        (12) edge[bend left=35] (2)
        ;
    \end{tikzpicture}
    \end{array}
    &&
    %AC4+P30
    \A(\P^0_{3})~
    \begin{array}{c}
    \begin{tikzpicture}
    	\draw[white] (-0.3,-0.3) rectangle ({\x+\z+0.3},{\x+0.3});
        \node[inner sep=1,blue] (2) at (0,\x){\small $e_2$};
        \node[inner sep=1,blue] (0) at (0,0){\small $\ZERO$};
		\node[inner sep=1] (23) at (\z,{\x+\zz}){\small $e_{2,3}$};
        \node[inner sep=1,blue] (13) at ({\x+\z},{0+\zz}){\small $e_{1,3}$};
        \node[inner sep=1] (1) at (\x,0){\small $e_1$};
        \node[inner sep=1,blue] (3) at (\z,\zz){\small $e_3$};
		\node[inner sep=1] (12) at (\x,\x){\small $e_{1,2}$};
        \path[thick,->]
        (2) edge[<->,blue] (0)
        (0) edge[<->,blue] (3)
		(3) edge[<->,blue] (13)
        (23) edge (2)
        (12) edge (2)
        (1) edge (0)
        ;
    \end{tikzpicture}
    \end{array}
    \\
    %C5+P30
   	\P^0_{4}
    \begin{array}{c}
    \begin{tikzpicture}
    	\draw[white] (-0.5,-0.5) rectangle (3.5,1.5);
        \node[inner sep=1] (13) at (0,0){\small $e_{1,3}$};
        \node[inner sep=1] (3) at (1,0){\small $e_3$};
        \node[inner sep=1] (1) at (2,0){\small $e_1$};
        \node[inner sep=1] (0x) at (3,0){\small $0*$};
        \node[inner sep=1] (2) at (0,1){\small $e_{2}$};
        \node[inner sep=1] (0) at (1,1){\small $\ZERO$};
        \node[inner sep=1] (23) at (2,1){\small $e_{2,3}$};
        \node[inner sep=1] (123) at (3,1){\small $e_{1,2,3}$};
		\node[inner sep=1] (12) at (4,1){\small $e_{1,2}$};
        \path[thick,->]
        (13) edge (3)
        (3) edge (1)
        (1) edge (0x)
        (2) edge (0)
        (0) edge (23)
        (23) edge (123)
        (123) edge (12)
        (12) edge[bend left=30] (2)
        ;
    \end{tikzpicture}
    \end{array}
    &&
    %AC5+P30
    \A(\P^0_{4})~
    \begin{array}{c}
    \begin{tikzpicture}
    	\draw[white] (-0.3,-0.3) rectangle ({\x+\z+0.3},{\x+0.3});
        \node[inner sep=1,blue] (2) at (0,\x){\small $e_2$};
        \node[inner sep=1,blue] (0) at (0,0){\small $\ZERO$};
		\node[inner sep=1] (23) at (\z,{\x+\zz}){\small $e_{2,3}$};
        \node[inner sep=1,blue] (13) at ({\x+\z},{0+\zz}){\small $e_{1,3}$};
        \node[inner sep=1] (1) at (\x,0){\small $e_1$};
        \node[inner sep=1,blue] (3) at (\z,\zz){\small $e_3$};
		\node[inner sep=1] (12) at (\x,\x){\small $e_{1,2}$};
		\node[inner sep=1] (123) at ({\x+\z},{\x+\zz}){\small $e_{1,2,3}$};
        \path[thick,->]
        (2) edge[<->,blue] (0)
        (0) edge[<->,blue] (3)
		(3) edge[<->,blue] (13)
        (23) edge (123)
        (123) edge (12)
        (12) edge (2)
        (1) edge (0)
        ;
    \end{tikzpicture}
    \end{array}
    \\
	%P60
	\P^0_{5}
    \begin{array}{c}
    \begin{tikzpicture}
    	\draw[white] (-0.5,-0.5) rectangle (3.5,1.5);
		\node[inner sep=1] (2) at (0,0.5){\small $e_2$};
		\node[inner sep=1] (0) at (1,0.5){\small $\ZERO$};		
		\node[inner sep=1] (23) at (2,0.5){\small $e_{2,3}$};		
		\node[inner sep=1] (13) at (3,0.5){\small $e_{1,3}$};		
		\node[inner sep=1] (3) at (4,0.5){\small $e_3$};		
		\node[inner sep=1] (1) at (5,0.5){\small $e_1$};		
        \node[inner sep=1] (0x) at (6,0.5){\small $0*$};		
        \path[thick,->]
        (2) edge (0)
        (0) edge (23)
        (23) edge (13)
        (13) edge (3)
        (3) edge (1)
        (1) edge (0x)
        ;
    \end{tikzpicture}
	\end{array}
	&&
    %AP60
    \A(\P^0_5)~
    \begin{array}{c}
    \begin{tikzpicture}
    	\draw[white] (-0.3,-0.3) rectangle ({\x+\z+0.3},{\x+0.3});
        \node[inner sep=1,blue] (2) at (0,\x){\small $e_2$};
        \node[inner sep=1,blue] (0) at (0,0){\small $\ZERO$};
		\node[inner sep=1] (23) at (\z,{\x+\zz}){\small $e_{2,3}$};
        \node[inner sep=1,blue] (13) at ({\x+\z},{0+\zz}){\small $e_{1,3}$};
        \node[inner sep=1] (1) at (\x,0){\small $e_1$};
        \node[inner sep=1,blue] (3) at (\z,\zz){\small $e_3$};
        \path[thick,->]
        (2) edge[<->,blue] (0)
        (0) edge[<->,blue] (3)
		(3) edge[<->,blue] (13)
        (23) edge (3)
        (1) edge (0)
        ;
    \end{tikzpicture}
    \end{array}
\end{array}
\]
{\caption{\label{fig:0-open_patterns} $0$-open patterns.}}
\end{figure}
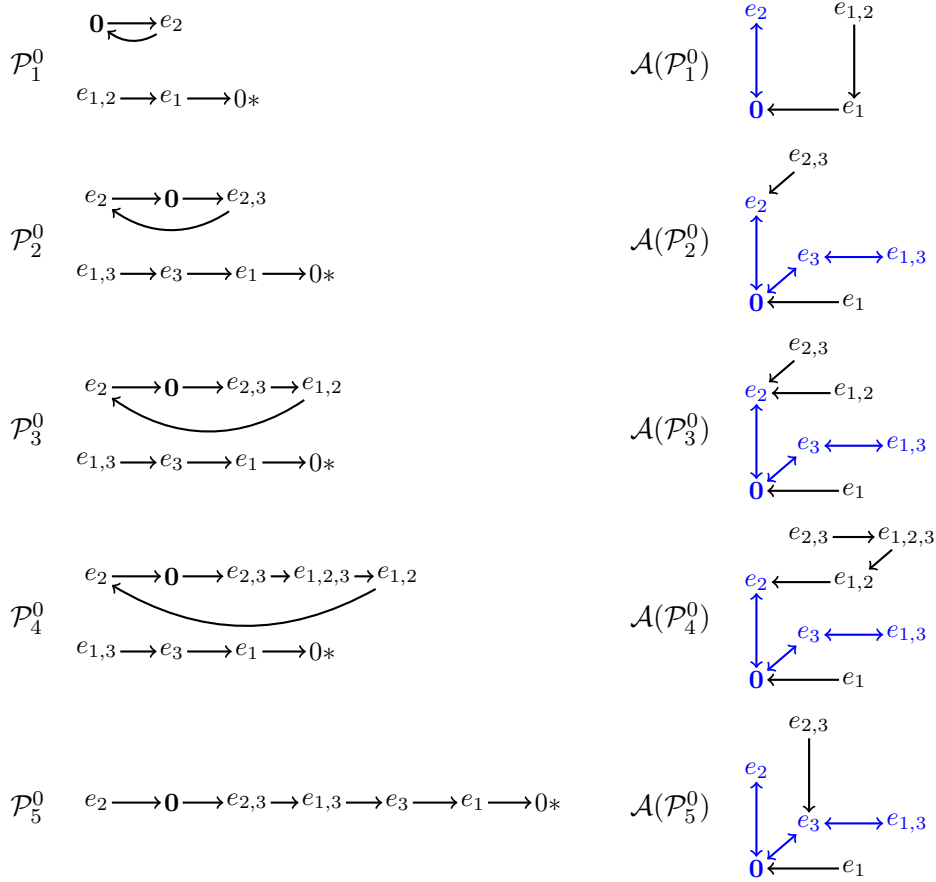

%%%% 1-open patterns
\begin{figure}[h]
\[
\def\x{1.3}
\def\z{0.7}
\def\zz{0.6}
\begin{array}{lclcl}
	%C2+P21
	\P^1_{1}
	\begin{array}{c}
    \begin{tikzpicture}
 	   \draw[white] (-0.5,-0.5) rectangle (3.5,1.5);	
        \node[inner sep=1] (1) at (0,1){\small $e_1$};
        \node[inner sep=1] (12) at (1,1){\small $e_{1,2}$};
        \node[inner sep=1] (2) at (0,0){\small $e_2$};
        \node[inner sep=1] (0) at (1,0){\small $\ZERO$};
        \node[inner sep=1] (1x) at (2,0){\small $1*$};
        \path[thick,->]
        (1) edge (12)
        (12) edge[bend left=35] (1)
        (2) edge (0)
        (0) edge (1x)
        ;
    \end{tikzpicture}
    \end{array}
    &&
    %AC2+P21
    \A(\P^1_{1})~
    \begin{array}{c}
    \begin{tikzpicture}
    	\draw[white] (-0.3,-0.3) rectangle ({\x+\z+0.3},{\x+0.3});
        \node[inner sep=1,blue] (1) at (\x,0){\small $e_1$};
        \node[inner sep=1] (2) at (0,\x){\small $e_2$};
        \node[inner sep=1] (0) at (0,0){\small $\ZERO$};
        \node[inner sep=1,blue] (12) at (\x,\x){\small $e_{1,2}$};
        \path[thick,->]
        (1) edge[<->,blue] (12)
        (2) edge (0)
		(0) edge (1)
        ;
    \end{tikzpicture}
    \end{array}
    \\
    %C3+P31
   	\P^1_{2}
    \begin{array}{c}
    \begin{tikzpicture}
    	\draw[white] (-0.5,-0.5) rectangle (3.5,1.5);
        \node[inner sep=1] (2) at (0,1){\small $e_2$};
        \node[inner sep=1] (0) at (1,1){\small $\ZERO$};
        \node[inner sep=1] (12) at (2,1){\small $e_{1,2}$};
        \node[inner sep=1] (1x) at (3,0){\small $1*$};
        \node[inner sep=1] (13) at (0,0){\small $e_{1,3}$};
        \node[inner sep=1] (1) at (1,0){\small $e_1$};
        \node[inner sep=1] (3) at (2,0){\small $e_3$};
        \path[thick,->]
        (2) edge (0)
        (0) edge (12)
        (12) edge[bend left=35] (2)
        (13) edge (1)
        (1) edge (3)
        (3) edge (1x)
        ;
    \end{tikzpicture}
    \end{array}
    &&
    %AC3+P31
    \A(\P^1_{2})~
    \begin{array}{c}
    \begin{tikzpicture}
    	\draw[white] (-0.3,-0.3) rectangle ({\x+\z+0.3},{\x+0.3});
        \node[inner sep=1,blue] (2) at (0,\x){\small $e_2$};
        \node[inner sep=1,blue] (0) at (0,0){\small $\ZERO$};
        \node[inner sep=1] (12) at (\x,\x){\small $e_{1,2}$};
        \node[inner sep=1,blue] (13) at ({\x+\z},{0+\zz}){\small $e_{1,3}$};
        \node[inner sep=1,blue] (1) at (\x,0){\small $e_1$};
        \node[inner sep=1] (3) at (\z,\zz){\small $e_3$};
		\node[inner sep=1,white] (23) at (\z,{\x+\zz}){\small $e_{2,3}$};%layout
        \path[thick,->]
        (2) edge[<->,blue] (0)
        (0) edge[<->,blue] (1)
        (12) edge (2)
        (13) edge[<->,blue] (1)
        (3) edge (13)
        ;
    \end{tikzpicture}
    \end{array}
    \\
    %C4+P31
   	\P^1_{3}
    \begin{array}{c}
    \begin{tikzpicture}
    	\draw[white] (-0.5,-0.5) rectangle (3.5,1.5);
        \node[inner sep=1] (2) at (0,1){\small $e_2$};
        \node[inner sep=1] (0) at (1,1){\small $\ZERO$};
        \node[inner sep=1] (12) at (2,1){\small $e_{1,2}$};
        \node[inner sep=1] (23) at (3,1){\small $e_{2,3}$};
        \node[inner sep=1] (1x) at (3,0){\small $1*$};
        \node[inner sep=1] (13) at (0,0){\small $e_{1,3}$};
        \node[inner sep=1] (1) at (1,0){\small $e_1$};
        \node[inner sep=1] (3) at (2,0){\small $e_3$};
        \path[thick,->]
        (2) edge (0)
        (0) edge (12)
        (12) edge (23)
		(23) edge[bend left=35] (2)
        (13) edge (1)
        (1) edge (3)
        (3) edge (1x)
        ;
    \end{tikzpicture}
    \end{array}
    &&
    %AC4+P31
    \A(\P^1_{3})~
    \begin{array}{c}
    \begin{tikzpicture}
    	\draw[white] (-0.3,-0.3) rectangle ({\x+\z+0.3},{\x+0.3});
        \node[inner sep=1,blue] (2) at (0,\x){\small $e_2$};
        \node[inner sep=1,blue] (0) at (0,0){\small $\ZERO$};
        \node[inner sep=1] (12) at (\x,\x){\small $e_{1,2}$};
        \node[inner sep=1,blue] (13) at ({\x+\z},{0+\zz}){\small $e_{1,3}$};
		\node[inner sep=1] (23) at (\z,{\x+\zz}){\small $e_{2,3}$};
        \node[inner sep=1,blue] (1) at (\x,0){\small $e_1$};
        \node[inner sep=1] (3) at (\z,\zz){\small $e_3$};
        \path[thick,->]
        (2) edge[<->,blue] (0)
        (0) edge[<->,blue] (1)
        (12) edge (2)
        (13) edge[<->,blue] (1)
        (3) edge (13)
        (23) edge (2)
        ;
    \end{tikzpicture}
    \end{array}
    \\
    %C5+P31
   	\P^1_{4}
    \begin{array}{c}
    \begin{tikzpicture}
    	\draw[white] (-0.5,-0.5) rectangle (3.5,1.5);
        \node[inner sep=1] (2) at (0,1){\small $e_2$};
        \node[inner sep=1] (0) at (1,1){\small $\ZERO$};
        \node[inner sep=1] (12) at (2,1){\small $e_{1,2}$};
        \node[inner sep=1] (123) at (3,1){\small $e_{1,2,3}$};
        \node[inner sep=1] (23) at (4,1){\small $e_{2,3}$};        
        \node[inner sep=1] (13) at (0,0){\small $e_{1,3}$};
        \node[inner sep=1] (1) at (1,0){\small $e_1$};
        \node[inner sep=1] (3) at (2,0){\small $e_3$};
        \node[inner sep=1] (1x) at (3,0){\small $1*$};
        \path[thick,->]
        (2) edge (0)
        (0) edge (12)
        (12) edge (123)
        (123) edge (23)
        (23) edge[bend left=30] (2)
        (13) edge (1)
        (1) edge (3)
        (3) edge (1x)
        ;
    \end{tikzpicture}
    \end{array}
    &&
    %AC5+P31
    \A(\P^1_{4})~
    \begin{array}{c}
    \begin{tikzpicture}
    	\draw[white] (-0.3,-0.3) rectangle ({\x+\z+0.3},{\x+0.3});
        \node[inner sep=1,blue] (2) at (0,\x){$e_2$};
        \node[inner sep=1,blue] (0) at (0,0){\small $\ZERO$};
        \node[inner sep=1] (12) at (\x,\x){\small $e_{1,2}$};
        \node[inner sep=1,blue] (13) at ({\x+\z},{0+\zz}){\small $e_{1,3}$};
        \node[inner sep=1] (123) at ({\x+\z},{\x+\zz}){\small $e_{1,2,3}$};
		\node[inner sep=1] (23) at (\z,{\x+\zz}){\small $e_{2,3}$};
        \node[inner sep=1,blue] (1) at (\x,0){\small $e_1$};
        \node[inner sep=1] (3) at (\z,\zz){\small $e_3$};
        \path[thick,->]
        (2) edge[<->,blue] (0)
        (0) edge[<->,blue] (1)
        (12) edge (123)
        (13) edge[<->,blue] (1)
        (3) edge (13)
        (23) edge (2)
        (123) edge (23)
        ;
    \end{tikzpicture}
    \end{array}
    \\
	%P61
	\P^1_5
    \begin{array}{c}
    \begin{tikzpicture}
        \draw[white] (-0.5,-0.5) rectangle (3.5,1.5);
        \node[inner sep=1] (2) at (0,0.5){\small $e_2$};
        \node[inner sep=1] (0) at (1,0.5){\small $\ZERO$};
        \node[inner sep=1] (12) at (2,0.5){\small $e_{1,2}$};
        \node[inner sep=1] (13) at (3,0.5){\small $e_{1,3}$};
        \node[inner sep=1] (1) at (4,0.5){\small $e_1$};
        \node[inner sep=1] (3) at (5,0.5){\small $e_3$};
        \node[inner sep=1] (1x) at (6,0.5){\small $1*$};
        \path[thick,->]
        (2) edge (0)
        (0) edge (12)
        (12) edge (13)
        (13) edge (1)
        (1) edge (3)
        (3) edge (1x)
        ;
    \end{tikzpicture}
	\end{array}
	&&
    %AP61
    \A(\P^1_5)~
    \begin{array}{c}
    \begin{tikzpicture}
    	\draw[white] (-0.3,-0.3) rectangle ({\x+\z+0.3},{\x+0.3});
        \node[inner sep=1,blue] (2) at (0,\x){\small $e_2$};
        \node[inner sep=1,blue] (0) at (0,0){\small $\ZERO$};
        \node[inner sep=1] (12) at (\x,\x){\small $e_{1,2}$};
        \node[inner sep=1,blue] (13) at ({\x+\z},{0+\zz}){\small $e_{1,3}$};
        \node[inner sep=1,blue] (1) at (\x,0){\small $e_1$};
        \node[inner sep=1] (3) at (\z,\zz){\small $e_3$};
		\node[inner sep=1,white] (23) at (\z,{\x+\zz}){\small $e_{2,3}$};%layout
        \path[thick,->]
        (2) edge[<->,blue] (0)
        (0) edge[<->,blue] (1)
        (12) edge (1)
        (13) edge[<->,blue] (1)
        (3) edge (13)
        ;
    \end{tikzpicture}
    \end{array}
\end{array}
\]
{\caption{\label{fig:1-open_patterns} $1$-open patterns.}}
\end{figure}
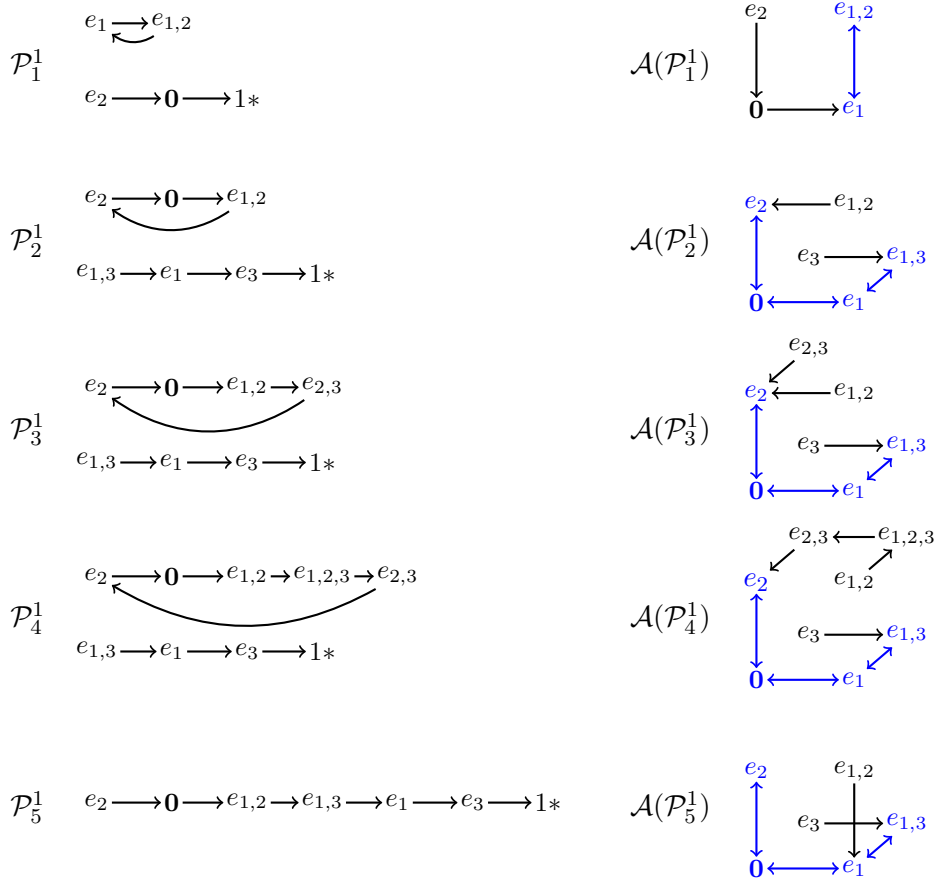

\medskip
The following is straightforward to check, and show that patterns produce small attractors. 

\begin{lemma}\label{lem:pattern}
If $f\in F(n)$ contains a pattern $\P$ then $A(\P)$ is an attractor of $\A(f)$ of size $\leq 4$ and $\A(f)$ has an almost decreasing path from any configuration in $X(\P)$ to $A(\P)$. 
\end{lemma}

We say that $f$ \EM{properly contains} a pattern $\P$ if $f$ contains $\P$ and $X=X(\P)$ satisfies the condition of Lemma~\ref{decreasing_lemma}; since $X$ is always a down set, this condition is: $f(\overline{X})\neq\overline{X}$ and, for all $1<\ell<n$, $\overline{X}$ contains at least ${n-1\choose \ell-1}+1$ configurations of weight $\ell$. We now show that proper containment of patters is unavoidable (this corresponds to the first step in the sketch of the proof of Theorem~\ref{thm:A4} given above):

\begin{lemma}\label{lem:unavoidable}
Let $f\in F(n)$ without fixed point, and suppose that $f^2\neq\id$ and $n\geq 5$. There exists $g\sim f$ such that $g$ properly contains a closed or $0$-open pattern $\P$. 
\end{lemma}

\begin{proof}
Let $\P$ be any pattern and $X=X(\P)$. Firstly, for $4\leq \ell<n$, since $X$ contains no configuration of weight $\ell$, $\overline{X}$ contains ${n\choose \ell}>{n-1\choose \ell-1}$ configurations of weight $\ell$. Secondly, since $X$ contains at most one configuration of weight $3$, $\overline{X}$ contains least ${n\choose 3}-1$ configurations of weight $3$; this exceeds ${n-1\choose 2}$ since $n\geq 5$. Thirdly, since $X$ contains at most three configurations of weight $2$, $\overline{X}$ contains at least ${n\choose 2}-3$ configurations of weight $2$; this exceeds ${n-1\choose 1}$ since $n\geq 5$. Consequently, $g$ properly contains $\P$ iff $g$ contains $\P$ and $g(\overline{X})\neq\overline{X}$. So we only have to prove that there exists $g\sim f$ such that $g$ contains a closed or $0$-open pattern $\P$ with $g(\overline{X})\neq\overline{X}$, where $X$ always means $X(\P)$. We proceed case by case, denoting $L^-$ and $L^+$ the minimum and maximum length of a limit cycle of $f$:
\begin{itemize}
\item
$L^-=2$. If $L^+\geq 3$ then, considering a cycle of length $2$ and a path of length $2$ in the cycle of length $L^+$, we deduce that there exists $g\sim f$ containing $\P=\P^0_1$, and $g(\overline{X})\neq\overline{X}$ since the image ($0*$) of the special configuration ($e_1$) is in $\overline{X}$. If $L^+=2$ then $f$ is not a permutation since otherwise $f^2=\id$. Furthermore, there exists $g\sim f$ containing $\P=\P_1$, and $g(\overline{X})\neq\overline{X}$ since $g$ is not a permutation \footnote{Here we can remark that, even if $\P_1$ is contained in $\P^0_1$, both are needed to ensure that $g(\overline{X})\neq\overline{X}$.}. 
\item
$L^-=3$. If $L^+\geq 4$ then, considering a cycle of length $3$ and a path of length $3$ in the cycle of length $L^+$, we deduce that there exists $g\sim f$ containing $\P=\P^0_2$, and $g(\overline{X})\neq\overline{X}$ since the image of the special configuration is in $\overline{X}$. If $L^+=3$ then $f$ is not a permutation because the number of configurations is a power of 2,
and hence not divisible by 3. We then have two cases. If $f$ has at least two limit cycles, then there exists $g\sim f$ containing $\P=\P_4$, and $g(\overline{X})\neq\overline{X}$ since $g$ is not a permutation. If $f$ has a unique limit cycle then, since $f$ is not a permutation, there exists $g\sim f$ containing $\P=\P_3$, and $g(\overline{X})\neq\overline{X}$ since all the periodic configurations of $g$ are in $X$. 
\item
$L^-=4$. If $f$ has at least two limit cycles then, considering a cycle of length $4$ and a path of length $3$ in a second cycle, we deduce that there exists $g\sim f$ containing $\P=\P^0_3$, and $g(\overline{X})\neq\overline{X}$ since the image of the special configuration is in $\overline{X}$. If $f$ has a unique limit cycle then $f$ is not a permutation. Furthermore, there exists $g\sim f$ containing $\P=\P_2$, and $g(\overline{X})\neq\overline{X}$ since $g$ is not a permutation \footnote{Here again we can remark that, even if $\P_2$ is contained in $\P^0_3$, both are needed to ensure that $g(\overline{X})\neq\overline{X}$.}.
\item
$L^-=5$. If $f$ has at least two limit cycles then, considering a cycle of length $5$ and a path of length $3$ in a second cycle, we deduce that there exists $g\sim f$ containing $\P=\P^0_4$, and $g(\overline{X})\neq\overline{X}$ since the image of the special configuration is in $\overline{X}$. If $f$ has a unique limit cycle then $f$ is not a permutation. We deduce that there exists $g\sim f$ containing $\P_5$, and $g(\overline{X})\neq\overline{X}$ since all the periodic configurations of $g$ are in $X$. 
\item
$L^-=6$ and $L^+=6$. Then $f$ is not a permutation because the number of configurations is a power of 2,
and hence not divisible by 6. Furthermore, there exists $g\sim f$ containing $\P=\P_6$, and $g(\overline{X})\neq\overline{X}$ since $g$ is not a permutation.
\item
$L^-\geq 6$ and $L^+\geq 7$. Considering a path of length $6$ in a cycle of length $L^+$, we deduce that there exists $g\sim f$ containing $\P=\P^0_5$, and $g(\overline{X})\neq\overline{X}$ since the image of the special configuration is in~$\overline{X}$. 
\end{itemize}
\end{proof}

The proof Theorem~\ref{thm:A4} is now an easy consequence Lemma \ref{lem:unavoidable} and the following lemma, which is proved assuming Lemma~\ref{decreasing_lemma}.

\begin{lemma}\label{lem:pattern_converging}
Let $g\in F(n)$ without fixed point. Suppose that $g$ properly contains a closed or $0$-pattern $\P$. There exists $h\sim g$ such that $\A(h)$ has an attractor $A$ of size $\leq 4$ and an almost decreasing path from any configuration to $A$.
\end{lemma}

\begin{proof}[{\bf Proof of Lemma~\ref{lem:pattern_converging} (assuming Lemma~\ref{decreasing_lemma})}]
Let $X=X(\P)$. Since $g$ properly contains $\P$, by Lemma~\ref{decreasing_lemma}, there exists a robustly $X$-converging function $h\sim_X g$. If $h$ contains $\P$ then, by Lemma \ref{lem:pattern}, $A(\P)$ is an attractor of $\A(h)$ of size $\leq 4$ and $\A(h)$ has an almost decreasing path from any configuration in $X$ to $A(\P)$, so we are done. 

\medskip
So suppose that $h$ does not contain $\P$. Since $g$ contains $\P$ and $h\sim_X g$, we deduce that $\P$ is a $0$-open pattern, say $\P=\P^0_k$, and that $\S(h)$ contains all the arcs of $\P$ except its special arc $x\to 0*$, that is: $\S(h)$ contains the induced subgraph $\P[X]$ but $h_1(x)=1\neq g_1(x)$; since $x\in X$ and $h(x)\neq g(x)$, we have $h(x)\not\in X$. Let $\P'=\P^1_k$. Since $\P[X]$ and $\P'[X]$ are isomorphic, there exists a permutation $\pi$, acting as the identity on $\overline{X}$, such that $h'=\pi\circ h\circ \pi^{-1}$ contains $\P'[X]$.  Since $x$ is the unique vertex of out-degree zero in $\P[X]$, $\pi(x)$ is the unique vertex of out-degree zero in $\P'[X]$, that is, $\pi(x)$ is the special configuration of $\P'$. Since $x\in X$ and $h(x)\not\in X$,  we have $h'(\pi(x))=h(x)$, thus $h'_1(\pi(x))=h_1(x)=1$ and thus $\S(h')$ contains the special arc $\pi(x)\to 1*$ of $\P'$. Hence, $h'$ contains $\P'$ (see Figure \ref{fig:correction} for an illustration). Since $h$ is robustly $X$-converging and $\pi$ acts as the identity on $\overline{X}$, the function $h'$ is $X$-converging: $\A(h')$ has a decreasing path from any configuration to $X$. By Lemma \ref{lem:pattern}, $A(\P')$ is an attractor of $\A(h')$ of size $\leq 4$ and $\A(h')$ has an almost decreasing path from any configuration in $X$ to $A(\P')$, so we are done. 
\end{proof}

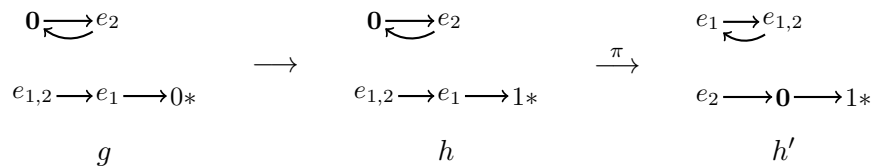
\begin{figure}[h]
\[
	\begin{array}{ccccc}
	\begin{array}{c}
    \begin{tikzpicture}
        \node[inner sep=1] (1) at (1,0){\small $e_1$};
        \node[inner sep=1] (12) at (0,0){\small $e_{1,2}$};
        \node[inner sep=1] (2) at (1,1){\small $e_2$};
        \node[inner sep=1] (0) at (0,1){\small $\ZERO$};
        \node[inner sep=1] (0x) at (2,0){\small $0*$};
        \path[thick,->]
        (0) edge (2)
        (2) edge[bend left=35] (0)
        (12) edge (1)
        (1) edge (0x)
        ;
    \end{tikzpicture}
    \end{array}
    &\begin{array}{c}\longrightarrow\end{array}&
    \begin{array}{c}
    \begin{tikzpicture}
        \node[inner sep=1] (1) at (1,0){\small $e_1$};
        \node[inner sep=1] (12) at (0,0){\small $e_{1,2}$};
        \node[inner sep=1] (2) at (1,1){\small $e_2$};
        \node[inner sep=1] (0) at (0,1){\small $\ZERO$};
        \node[inner sep=1] (0x) at (2,0){\small $1*$};
        \path[thick,->]
        (0) edge (2)
        (2) edge[bend left=35] (0)
        (12) edge (1)
        (1) edge (0x)
        ;
    \end{tikzpicture}
    \end{array}
    &\begin{array}{c}\stackrel{\pi}{\longrightarrow}\end{array}&
	\begin{array}{c}
	\begin{tikzpicture}
 	  
        \node[inner sep=1] (1) at (0,1){\small $e_1$};
        \node[inner sep=1] (12) at (1,1){\small $e_{1,2}$};
        \node[inner sep=1] (2) at (0,0){\small $e_2$};
        \node[inner sep=1] (0) at (1,0){\small $\ZERO$};
        \node[inner sep=1] (1x) at (2,0){\small $1*$};
        \path[thick,->]
        (1) edge (12)
        (12) edge[bend left=35] (1)
        (2) edge (0)
        (0) edge (1x)
        ;
    \end{tikzpicture}
    \end{array}
    \\[7mm]
    g && h && h'
    \end{array}
\]
{\caption{\label{fig:correction}Correction step in the proof of Lemma~\ref{lem:pattern_converging} for $\P=\P^0_1$.
}}
\end{figure}

It remains to treat the case $f^2=\id$, which is easy. 

\begin{lemma}\label{lem:f^2=id}
Let $f\in F(n)$ with no fixed point and $f^2= \id$. There exists $h\sim f$ such that $\A(h)$ has an attractor $A$ of size $\leq 4$ and an almost decreasing path from any configuration to $A$.
\end{lemma}

\begin{proof}
Since $f$ has no fixed point and $f^2 = \id$, $\S(f)$ consists of exactly $2^{n-2}$ cycles of length~$2$. For $n=2$, we consider the negation $h$, that is $h(x)=\overline{x}$ for all $x\in\B^2$; then $h\sim f$ and $\A(h)$ is strongly connected so we are done. So suppose that $n\geq 3$.  Let $X=\B^n\setminus \{\ZERO, e_1, e_2, \overline{e_1}, \overline{e_2}, \ONE\}$ and let $h \in F(n)$ defined by:
\[
\begin{array}{lll}
h(\ZERO) &=& e_1\\
h(e_1) &=&\ZERO
\end{array}
\quad
\begin{array}{lll}
h(\ONE) &=& e_2\\
h(e_2) &=& \ONE
\end{array}
\quad
\begin{array}{lll}
h(\overline{e_1}) &=& \overline{e_2}\\
h(\overline{e_2}) &=& \overline{e_1}
\end{array}
\quad
h(x) = \overline{x}\textrm{ for all }x\in X.
\]
One easily check that $\S(h)$ consists of $2^{n-2}$ cycles of length $2$, so $h \sim f$, and that $\{\ZERO, e_1\}$ is an attractor of $\A(h)$. If $x\not\in \{\ZERO, e_2\}$, then there exists $i\in [n]$ with $h_i(x)<x_i$ and thus $\A(h)$ has the decreasing arc $x\to x+e_i$. We deduce $\A(h)$ has a decreasing path from any configuration to $\{\ZERO,e_2\}$. Since $\A(h)$ has the almost decreasing path $e_2\to e_{2,3}\to e_3\to\ZERO$, $\A(h)$ has an almost decreasing path from any configuration to $\ZERO$. 
\end{proof}

%%%%%%%%%%%%%%%%%%%%%%%%%%%%%%%%%%%%%%%%%%%%%%%%%%%%%%%%%%%%%%%%%%%
\subsection{Decreasing Lemma}\label{sec:decreasing_lemma}
%%%%%%%%%%%%%%%%%%%%%%%%%%%%%%%%%%%%%%%%%%%%%%%%%%%%%%%%%%%%%%%%%%%

In this subsection, we prove Lemma \ref{decreasing_lemma}, used several times above. Let us first give a rough description of the proof. Let $f\in F(n)$ and let $X\subseteq\B^n$ be a down set with $f(\overline{X})\neq\overline{X}$. First we consider a total order $\preceq$ on $\overline{X}$, called a good order, which is a kind of topological sort of the subgraph of $\S(f)$ induced by $\overline{X}$. Roughly, each cycle corresponds to an interval, and $\preceq$ is a topological sort on the acyclic part; see Figure \ref{fig:decreasing_lemma}(a). The important point is that the cycles are not intertwined, and that the maximal element is in the acyclic part (the acyclic part is non-empty because $f(\overline{X})\neq \overline{X}$). We then consider a permutation $\pi$ of $\B^n$, acting as the identity on $X$, such that $x\preceq y$ implies $w(\pi(x))\leq w(\pi(y))$; such a permutation is called monotone. Setting $h=\pi\circ f\circ\pi^{-1}$, and taking $x\in \overline{X}$ with $h(x)\prec x$, we have $w(h(x))\leq w(x)$ by monotonicity, and thus $\A(h)$ has a decreasing arc starting from $x$. Also, if $f(x)\in X$ then there is a decreasing transition starting from $x$ since $X$ is a down set. Actually, if $x$ is not a fixed point and there is no decreasing arc starting from $x$, that is $x<f(x)$, then $x\prec f(x)$ and thus $x$ is in a cycle; see Figure \ref{fig:decreasing_lemma}(b). We can however fix this problem by swapping $h(x)$ and a configuration $y\in \overline{X}$ with the same weight; for that  we need enough elements in $\overline{X}$, hence the counting condition in the statement of Lemma \ref{decreasing_lemma}. This swap keeps the monotonicity and adds a decreasing transition starting from $x$; see Figure \ref{fig:decreasing_lemma}(b). Since cycles are not intertwined, with several such corrections, we eventually obtain a permutation $\pi$, acting as the identity on $X$, such that $h=\pi\circ f\circ\pi^{-1}$ is $X$-converging; and because $X$ is a down set and $\pi$ acts as the identity on $X$, $h$ is robustly $X$-converging and $h\sim_X f$.

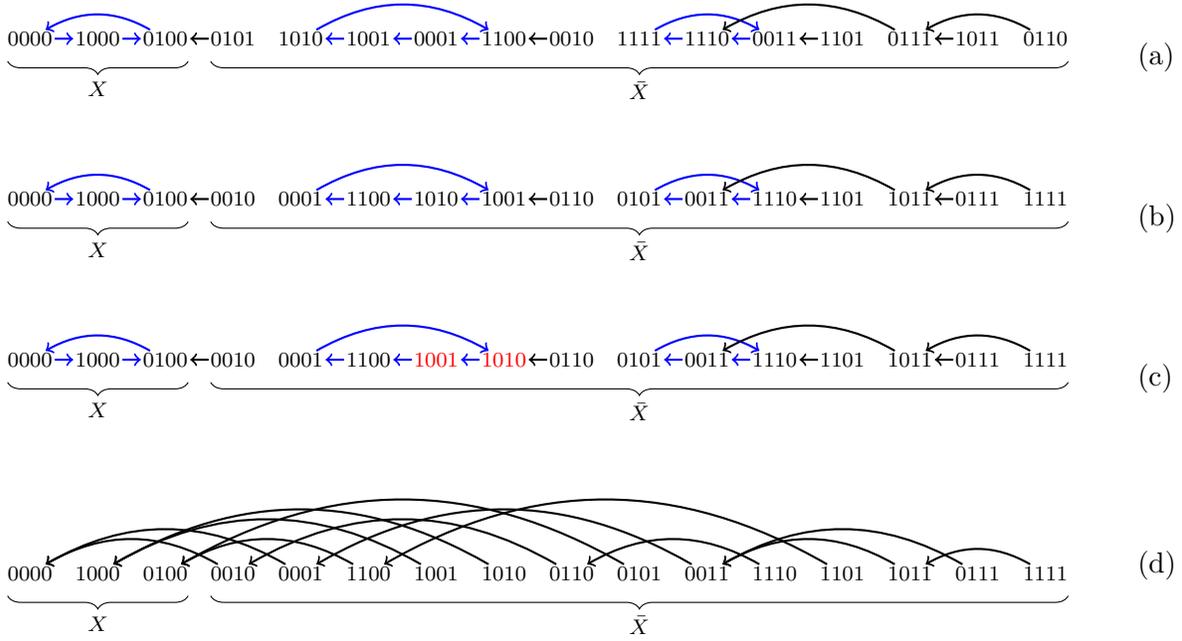
\begin{figure}
\[
\begin{array}{cl}
	\begin{array}{c}
    \begin{tikzpicture}
		\def\x{0.9}
	    \node[inner sep=0.4,outer sep=0.4] (0) at ({0*\x},0){\scriptsize $0000$};
        \node[inner sep=0.4,outer sep=0.4] (1) at ({1*\x},0){\scriptsize $1000$};
        \node[inner sep=0.4,outer sep=0.4] (2) at ({2*\x},0){\scriptsize $0100$};
        \node[inner sep=0.4,outer sep=0.4] (3) at ({3*\x},0){\scriptsize $0101$};%1
        \node[inner sep=0.4,outer sep=0.4] (4) at ({4*\x},0){\scriptsize $1010$};%2
        \node[inner sep=0.4,outer sep=0.4] (12) at ({5*\x},0){\scriptsize $1001$};%3
        \node[inner sep=0.4,outer sep=0.4] (13) at ({6*\x},0){\scriptsize $0001$};%4
        \node[inner sep=0.4,outer sep=0.4] (14) at ({7*\x},0){\scriptsize $1100$};%5
        \node[inner sep=0.4,outer sep=0.4] (23) at ({8*\x},0){\scriptsize $0010$};%6
        \node[inner sep=0.4,outer sep=0.4] (24) at ({9*\x},0){\scriptsize $1111$};%7
        \node[inner sep=0.4,outer sep=0.4] (34) at ({10*\x},0){\scriptsize $1110$};%8
        \node[inner sep=0.4,outer sep=0.4] (123) at ({11*\x},0){\scriptsize $0011$};%9
        \node[inner sep=0.4,outer sep=0.4] (124) at ({12*\x},0){\scriptsize $1101$};%10
		\node[inner sep=0.4,outer sep=0.4] (134) at ({13*\x},0){\scriptsize $0111$};%11
        \node[inner sep=0.4,outer sep=0.4] (234) at ({14*\x},0){\scriptsize $1011$};%12
        \node[inner sep=0.4,outer sep=0.4] (1234) at ({15*\x},0){\scriptsize $0110$};%13        

		\draw [decorate,decoration={brace,amplitude=5pt,mirror}]
  		({0*\x-0.3},-0.3) -- ({2*\x+0.3},-0.3) node[midway,below,yshift=-4]{\scriptsize $X$};
        \draw [decorate,decoration={brace,amplitude=5pt,mirror}]
  		({3*\x-0.3},-0.3) -- ({15*\x+0.3},-0.3) node[midway,below,yshift=-4]{\scriptsize $\overline{X}$};
		        
        \path[thick,->,draw,black]
        (2) edge[blue,bend right=30] (0)
        (0) edge[blue] (1)
        (1) edge[blue] (2)
        (3) edge (2)
        (4) edge[blue,bend left=30] (14)
        (14) edge[blue] (13)
        (13) edge[blue] (12)
        (12) edge[blue] (4)
        (23) edge (14) 
        (24) edge[blue,bend left=30] (123)
        (123) edge[blue] (34)
        (34) edge[blue] (24)
        (124) edge (123)
        (134) edge[bend right=30] (34)
        (234) edge (134)
        (1234) edge[bend right=30] (134)
		;
    \end{tikzpicture}
    \end{array}
&\quad\textrm{(a)}
\\~\\
	\begin{array}{c}
    \begin{tikzpicture}
		\def\x{0.9}
	    \node[inner sep=0.4,outer sep=0.4] (0) at ({0*\x},0){\scriptsize $0000$};
        \node[inner sep=0.4,outer sep=0.4] (1) at ({1*\x},0){\scriptsize $1000$};
        \node[inner sep=0.4,outer sep=0.4] (2) at ({2*\x},0){\scriptsize $0100$};
        \node[inner sep=0.4,outer sep=0.4] (3) at ({3*\x},0){\scriptsize $0010$};
        \node[inner sep=0.4,outer sep=0.4] (4) at ({4*\x},0){\scriptsize $0001$};
        \node[inner sep=0.4,outer sep=0.4] (12) at ({5*\x},0){\scriptsize $1100$};
        \node[inner sep=0.4,outer sep=0.4] (13) at ({6*\x},0){\scriptsize $1010$};
        \node[inner sep=0.4,outer sep=0.4] (14) at ({7*\x},0){\scriptsize $1001$};
        \node[inner sep=0.4,outer sep=0.4] (23) at ({8*\x},0){\scriptsize $0110$};
        \node[inner sep=0.4,outer sep=0.4] (24) at ({9*\x},0){\scriptsize $0101$};
        \node[inner sep=0.4,outer sep=0.4] (34) at ({10*\x},0){\scriptsize $0011$};
        \node[inner sep=0.4,outer sep=0.4] (123) at ({11*\x},0){\scriptsize $1110$};
        \node[inner sep=0.4,outer sep=0.4] (124) at ({12*\x},0){\scriptsize $1101$};
		\node[inner sep=0.4,outer sep=0.4] (134) at ({13*\x},0){\scriptsize $1011$};
        \node[inner sep=0.4,outer sep=0.4] (234) at ({14*\x},0){\scriptsize $0111$};
        \node[inner sep=0.4,outer sep=0.4] (1234) at ({15*\x},0){\scriptsize $1111$};  
        
		\draw [decorate,decoration={brace,amplitude=5pt,mirror}]
  		({0*\x-0.3},-0.3) -- ({2*\x+0.3},-0.3) node[midway,below,yshift=-4]{\scriptsize $X$};
        \draw [decorate,decoration={brace,amplitude=5pt,mirror}]
  		({3*\x-0.3},-0.3) -- ({15*\x+0.3},-0.3) node[midway,below,yshift=-4]{\scriptsize $\overline{X}$};
        
        \path[thick,->,draw,black]
        (2) edge[blue,bend right=30] (0)
        (0) edge[blue] (1)
        (1) edge[blue] (2)
        (3) edge (2)
        (4) edge[blue,bend left=30] (14)
        (14) edge[blue] (13)
        (13) edge[blue] (12)
        (12) edge[blue] (4)
        (23) edge (14) 
        (24) edge[blue,bend left=30] (123)
        (123) edge[blue] (34)
        (34) edge[blue] (24)
        (124) edge (123)
        (134) edge[bend right=30] (34)
        (234) edge (134)
        (1234) edge[bend right=30] (134)
		;
    \end{tikzpicture}
    \end{array}
&\quad\textrm{(b)}
\\~\\
	\begin{array}{c}
    \begin{tikzpicture}
		\def\x{0.9}
	    \node[inner sep=0.4,outer sep=0.4] (0) at ({0*\x},0){\scriptsize $0000$};
        \node[inner sep=0.4,outer sep=0.4] (1) at ({1*\x},0){\scriptsize $1000$};
        \node[inner sep=0.4,outer sep=0.4] (2) at ({2*\x},0){\scriptsize $0100$};
        \node[inner sep=0.4,outer sep=0.4] (3) at ({3*\x},0){\scriptsize $0010$};
        \node[inner sep=0.4,outer sep=0.4] (4) at ({4*\x},0){\scriptsize $0001$};
        \node[inner sep=0.4,outer sep=0.4] (12) at ({5*\x},0){\scriptsize $1100$};
        \node[inner sep=0.4,outer sep=0.4,red] (13) at ({6*\x},0){\scriptsize $1001$};
        \node[inner sep=0.4,outer sep=0.4,red] (14) at ({7*\x},0){\scriptsize $1010$};
        \node[inner sep=0.4,outer sep=0.4] (23) at ({8*\x},0){\scriptsize $0110$};
        \node[inner sep=0.4,outer sep=0.4] (24) at ({9*\x},0){\scriptsize $0101$};
        \node[inner sep=0.4,outer sep=0.4] (34) at ({10*\x},0){\scriptsize $0011$};
        \node[inner sep=0.4,outer sep=0.4] (123) at ({11*\x},0){\scriptsize $1110$};
        \node[inner sep=0.4,outer sep=0.4] (124) at ({12*\x},0){\scriptsize $1101$};
		\node[inner sep=0.4,outer sep=0.4] (134) at ({13*\x},0){\scriptsize $1011$};
        \node[inner sep=0.4,outer sep=0.4] (234) at ({14*\x},0){\scriptsize $0111$};
        \node[inner sep=0.4,outer sep=0.4] (1234) at ({15*\x},0){\scriptsize $1111$};  
        
		\draw [decorate,decoration={brace,amplitude=5pt,mirror}]
  		({0*\x-0.3},-0.3) -- ({2*\x+0.3},-0.3) node[midway,below,yshift=-4]{\scriptsize $X$};
        \draw [decorate,decoration={brace,amplitude=5pt,mirror}]
  		({3*\x-0.3},-0.3) -- ({15*\x+0.3},-0.3) node[midway,below,yshift=-4]{\scriptsize $\overline{X}$};
        
        \path[thick,->,draw,black]
        (2) edge[blue,bend right=30] (0)
        (0) edge[blue] (1)
        (1) edge[blue] (2)
        (3) edge (2)
        (4) edge[blue,bend left=30] (14)
        (14) edge[blue] (13)
        (13) edge[blue] (12)
        (12) edge[blue] (4)
        (23) edge (14) 
        (24) edge[blue,bend left=30] (123)
        (123) edge[blue] (34)
        (34) edge[blue] (24)
        (124) edge (123)
        (134) edge[bend right=30] (34)
        (234) edge (134)
        (1234) edge[bend right=30] (134)
		;
    \end{tikzpicture}
    \end{array}
&\quad\textrm{(c)}
\\~\\
	\begin{array}{c}
    \begin{tikzpicture}
		\def\x{0.9}
	    \node[inner sep=0.4,outer sep=0.4] (0) at ({0*\x},0){\scriptsize $0000$};
        \node[inner sep=0.4,outer sep=0.4] (1) at ({1*\x},0){\scriptsize $1000$};
        \node[inner sep=0.4,outer sep=0.4] (2) at ({2*\x},0){\scriptsize $0100$};
        \node[inner sep=0.4,outer sep=0.4] (3) at ({3*\x},0){\scriptsize $0010$};
        \node[inner sep=0.4,outer sep=0.4] (4) at ({4*\x},0){\scriptsize $0001$};
        \node[inner sep=0.4,outer sep=0.4] (12) at ({5*\x},0){\scriptsize $1100$};
        \node[inner sep=0.4,outer sep=0.4] (13) at ({6*\x},0){\scriptsize $1001$};
        \node[inner sep=0.4,outer sep=0.4] (14) at ({7*\x},0){\scriptsize $1010$};
        \node[inner sep=0.4,outer sep=0.4] (23) at ({8*\x},0){\scriptsize $0110$};
        \node[inner sep=0.4,outer sep=0.4] (24) at ({9*\x},0){\scriptsize $0101$};
        \node[inner sep=0.4,outer sep=0.4] (34) at ({10*\x},0){\scriptsize $0011$};
        \node[inner sep=0.4,outer sep=0.4] (123) at ({11*\x},0){\scriptsize $1110$};
        \node[inner sep=0.4,outer sep=0.4] (124) at ({12*\x},0){\scriptsize $1101$};
		\node[inner sep=0.4,outer sep=0.4] (134) at ({13*\x},0){\scriptsize $1011$};
        \node[inner sep=0.4,outer sep=0.4] (234) at ({14*\x},0){\scriptsize $0111$};
        \node[inner sep=0.4,outer sep=0.4] (1234) at ({15*\x},0){\scriptsize $1111$};  
        
		\draw [decorate,decoration={brace,amplitude=5pt,mirror}]
  		({0*\x-0.3},-0.3) -- ({2*\x+0.3},-0.3) node[midway,below,yshift=-4]{\scriptsize $X$};
        \draw [decorate,decoration={brace,amplitude=5pt,mirror}]
  		({3*\x-0.3},-0.3) -- ({15*\x+0.3},-0.3) node[midway,below,yshift=-4]{\scriptsize $\overline{X}$};
        
        \path[thick,->,draw,black]
        (3) edge[bend right=30] (0)
        (4) edge[bend right=30] (0)
        (12) edge[bend right=30] (2)
		(14) edge[bend right=30] (1)
		(13) edge[bend right=30] (1)
		(23) edge[bend right=30] (3)
		(24) edge[bend right=30] (2)
		(34) edge[bend right=30] (4)
		(123) edge[bend right=30] (23)
		(124) edge[bend right=30] (12)
		(134) edge[bend right=30] (34)
		(234) edge[bend right=30] (34)
		(1234) edge[bend right=30] (134)
		;
    \end{tikzpicture}
    \end{array}
&\quad\textrm{(d)}
\end{array}
\]
{\caption{\label{fig:decreasing_lemma}
(a) Down set $X\subseteq \B^4$ and function $f\in F(4)$ represented with a good order on $\overline{X}$, from left to right. (b) Function $h\sim f$ obtained with a monotone permutation: configurations in $\overline{X}$ have non-decreasing weight. $h$ is not $X$-converging: there is no decreasing transition starting from $0001$ since $0001<h(0001)=1001$. (c) Function $h'\sim h$ obtained by swapping $1001$ and $1010$. $h'$ is $X$-converging since there is no configurations $x\in \overline{X}$ with $x<h'(x)$. (d) Subgraph of $\A(h')$, showing that $h'$ is $X$-converging. 
}}
\end{figure}

\medskip
We now proceed with the details. Let $f\in F(n)$ and $Y \subseteq \B^n$. A \EM{good order} of $f$ on $Y$ is a total order $\preceq$ on $Y$ such that
\[
x\prec y\preceq f(x)~\Rightarrow~ f(y)\prec y.
\]

\begin{lemma}\label{lem:good}
Let $f\in F(n)$ and $Y \subseteq \B^n$. Then, $f$ has a good order on $Y$ and, for every $x\in Y\setminus f(Y)$, $f$ has a good order on $Y$ whose maximal element is $x$.  
\end{lemma}

\begin{proof}
Let $S_1,\dots,S_r$ be the connected components of the subgraph of $\S(f)$ induced by $Y$. For $1\leq i\leq r$, let $T_i$ be defined as follows. If $S_i$ is acyclic then $T_i=S_i$. Otherwise, $S_i$ has a unique cycle; we take any $x$ in this cycle and let $T_i$ be obtained from $S_i$ by deleting the arc from $x$ to $f(x)$. In both cases, $T_i$ is a tree with a unique vertex of out-degree zero, denoted $x^i$. Hence, $T_i$ has a topological order $\preceq_i$, that is, a total order on the vertex set of $T_i$ with $f(x)\prec_i x$ for all $x$ in $T_i$ distinct from $x^i$. Let $\preceq$ be the total order on $Y$ defined as follows: for all $x,y\in Y$, let $1\leq i,j\leq r$ such that $x$ is in $T_i$ and $y$ in $T_j$; then $x\preceq y$ if and only if either $i<j$ or $i=j$ and $x\preceq_i y$. Then $\preceq$ is a good order of $f$ on $Y$. Indeed, suppose that $x\prec y\preceq f(x)$. Since $x$ and $f(x)$ are in the same component, say $S_i$, we have $x\prec_i f(x)$ thus $x=x^i$. Let $S_j$ be the connected component containing $y$. Since $x\prec y\preceq f(x)$ we have $i\leq j\leq i$. Thus, $y$ and $f(y)$ are in $S_i$ and since $y\neq x=x^i$ we have $f(y)\prec_i y$ and thus $f(y)\prec y$. Furthermore, given $x\in Y\setminus f(Y)$ we can re-order the connected components so that $x$ is in $S_r$. Since $x$ has in-degree zero in $T_r$, we can choose $\preceq_r$ so that $x$ is its maximal element. Then $x$ is the maximal element of $\preceq$.  
\end{proof}

\medskip
We need another lemma, giving a sufficient condition to obtain the conclusion of Lemma~\ref{decreasing_lemma}. 

\begin{lemma}\label{lem:converging_to_robust}
Let $f\in F(n)$ and $X \subseteq \B^n$ be a non-empty down set. Suppose that there exists a permutation $\pi$ of $\B^n$, acting as the identity on $X$, such that $h=\pi\circ f\circ\pi^{-1}$ is $X$-converging. Then $h$ is robustly $X$-converging  and $h\sim_X f$.
\end{lemma}

\begin{proof}
The fact that $h\sim_X f$ simply follows from the fact that $\pi$ acts as the identity on $X$. So it remains to prove that $h$ is robustly $X$-converging. Let $\sigma$ be any permutation of $\B^n$ acting as the identity on $\overline{X}$, and let $h'=\sigma\circ h\circ\sigma^{-1}$. We have to prove that $h'$ is $X$-converging, which is equivalent to prove that $\A(h')$ has a decreasing arc starting from each configuration in $\overline{X}\setminus \FP(h')$. Suppose, for a contradiction, that there exists $x\in \overline{X}\setminus \FP(h')$ with $x< h'(x)$. Since $x\in \overline{X}$ and since $\overline{X}$ is an upper set, we have $h'(x)\in \overline{X}$. Since $\sigma$ acts as the identity on $\overline{X}$ we obtain $h(x)=h'(x)$. Thus, $x<h(x)$, and this  contradicts the fact that $h$ is $X$-converging. So $h'$ is $X$-converging, and thus $h$ is robustly $X$-converging. 
\end{proof}

We are now ready to prove Lemma \ref{decreasing_lemma}, which we restate. 

\setcounter{lemma}{8} % Définir le numéro du lemme comme 9
\begin{lemma}
Let $f\in F(n)$ and let $X\subseteq \B^n$ be a non-empty down set with $f(\overline{X})\neq \overline{X}$. Suppose that, for all $1<\ell<n$, $\overline{X}$ contains at least ${n-1\choose \ell-1}+1$ configurations of weight $\ell$. Then there exists a robustly $X$-converging function $h\sim_X f$. 
\end{lemma}
\setcounter{lemma}{15} % Réinitialiser le compteur des lemmes à 16

\begin{proof}
Let $Y=\overline{X}$ which is an upper set since $X$ is a down set. Since $f(Y)\neq Y$, there exists $a\in Y\setminus f(Y)$. By Lemma~\ref{lem:good}, there exists a good order $\preceq$ of $f$ on $Y$ whose maximal element is $a$. Given a permutation $\pi$ of $\B^n$, we denote by $\preceq_\pi$ the total order on $\pi(Y)$ defined as follows: for all $x,y\in Y$, 
\[
\pi(x)\preceq_{\pi} \pi(y)~\iff ~ x\preceq y.
\]

\medskip
Let $\Pi$ be the set of permutations $\pi$ of $\B^n$ such that $\pi$ acts as the identity on $X$ and $\preceq_\pi$ is {\em monotone}, that is
\[
x\preceq_\pi y~\Rightarrow~ w(x)\leq w(y).
\]
It is obvious that $\Pi\neq\emptyset$. Furthermore, for all $\pi\in \Pi$ we have $\pi(Y)=Y$ (since $\pi$ acts as the identity on $X$), so $\preceq_\pi$ is a monotone total order on $Y$. In addition, $\preceq_\pi$ is a good order of $h=\pi\circ f\circ\pi^{-1}$. Indeed, let $x,y\in Y$ and suppose that $\pi(x)\prec_\pi \pi(y)\preceq_\pi h(\pi(x))$, which is equivalent to $x\prec y\preceq f(x)$. Since $\preceq$ is a good order of $f$, we have $f(y)\prec y$, and thus $h(\pi(y))=\pi(f(y))\prec_\pi \pi(y)$. 

\medskip
Given $\pi\in\Pi$, we set $f^\pi=\pi\circ f\circ\pi^{-1}$ and 
\[
Z_\pi=\{x\in Y\mid x<f^\pi(x)\}.
\]

\medskip
Suppose that there exists $\pi\in\Pi$ such that $Z_\pi=\emptyset$, and let $h=f^\pi$. Since $Z_\pi=\emptyset$, for all $x\in Y$, we have $x\not< h(x)$ thus either $h(x)=x$ or $h_i(x)<x_i$ for some $i\in [n]$, that is, $\A(h)$ has a decreasing arc starting from $x$. Thus, $h$ is $X$-converging. By Lemma~\ref{lem:converging_to_robust}, $h$ is robustly $X$-converging and $h\sim_X f$ so the lemma holds. 

\medskip
Consequently, it is sufficient to prove that $Z_\pi=\emptyset$ for some $\pi\in\Pi$. Suppose, by contradiction, that $Z_\pi\neq\emptyset$ for all $\pi\in\Pi$. Given $\pi\in\Pi$, we set 
\[
w_\pi=\min_{x\in Z_\pi} w(x).
\] 
Let $\pi\in\Pi$ which maximizes $w_\pi$ and $h=f^\pi$. Let $x\in Z_\pi$ with $w(x)=w_\pi$. Since $x<h(x)$ and $Y$ is an upper set, we have $h(x)\in Y$. Since $a$ is the maximal element of $\preceq$, and since $\ONE$ is the maximal element of $\preceq_\pi$, we have $\pi(a)=\ONE$. Thus if $h(x)=\ONE$ we have 
\[
a=\pi^{-1}(\ONE)=\pi^{-1}(h(x))=\pi^{-1}(\pi(f(\pi^{-1}(x))))=f(\pi^{-1}(x)), 
\]
and since $x\in Y$, we have $\pi^{-1}(x)\in Y$ so $a\in f(Y)$ and this contradicts our choice of $a$. Consequently, $h(x)<\ONE$. Furthermore, since $|Y|<2^n$ (because $X$ is non-empty) and since $Y$ is an upper set, we have $\ZERO\not\in Y$. Thus, 
\[
\ZERO<x<h(x)<\ONE.
\]

\smallskip
Let us prove that there exists $y\in Y$ such that $w(y)=w(h(x))$ and $x\not\leq y$. Let $k=w(x)$ and $\ell=w(h(x))$. By the above inequality, we have $1<\ell<n$. Let $A$ be the set of $y\in \B^n$ with $w(y)=\ell$, and let $B$ be the set of $y\in A$ such that $x\leq y$. We have $|B|={n-k\choose \ell-k}$. By hypothesis, $|Y\cap A|>{n-1\choose \ell-1}\geq |B|$. Thus, there indeed exists $y\in (Y\cap A)\setminus B$ such that $w(y)=w(h(x))$ and $x\not\leq y$.

\smallskip
Let $\sigma=(h(x)\leftrightarrow y)$. Since $y,h(x)\in Y$,  $\sigma$ acts as the identity on $X$. Furthermore, it preserves the weight, and thus $\sigma\circ\pi\in \Pi$. Note that $\sigma=\sigma^{-1}$ and thus $f^{\sigma\circ\pi}=\sigma\circ h\circ\sigma$. 

\smallskip
We will prove that $w_{\sigma\circ\pi}>w_{\pi}$, which is the desired contradiction. First, we have $x\not\in Z_{\sigma\circ\pi}$ since $x\not\leq y$ and 
\[
f^{\sigma\circ\pi}(x)=\sigma(h(\sigma(x)))=\sigma(h(x))=y.
\] 
Let any $z\in Z_{\sigma\circ\pi}$; so $z\neq x$. Suppose, by contradiction, that $w(z)\leq w(x)$. Then $w(z)<w(y)=w(h(x))$, thus $\sigma(z)=z$ and $f^{\sigma\circ\pi}(z)=\sigma(h(z))$. Since $\sigma$ preserves the weight and since $z<f^{\sigma\circ\pi}(z)$ we have $w(z)<w(h(z))$. Since $w(z)\leq w(x)<w(h(x))$ and $z\neq x$ we obtain, by monotonicity,
\[
z\prec_\pi h(z),\quad z\prec_\pi h(x),\quad x\prec_\pi h(x).
\]
Since $\preceq_\pi$ is a good order for $h$, we cannot have $x\prec_\pi z\prec_\pi h(x)$ and this forces $z\prec_\pi x$. Similarly, we cannot have $z\prec_\pi x\prec_\pi h(z)$ and this forces $h(z)\preceq_\pi x$. Thus, since $\preceq_\pi$ is monotone, 
\[
w(z)<w(h(z))\leq w(x).
\]
Consequently, $f^{\sigma\circ\pi}(z)=\sigma(h(z))=h(z)$. Thus, $z<h(z)$, but then $z\in Z_\pi$ and since $w(z)<w(x)$ this contradicts the choice of $x$. Thus, $w(z)>w(x)$ for all $z\in Z_{\sigma\circ\pi}$. But then $w_{\sigma\circ\pi}>w_\pi$, which is the desired contradiction.
\end{proof}

%%%%%%%%%%%%%%%%%%%%%%%%%%%%%%%%%%%%%%%%%%%%%%%%%%%%%%%%%%%%%%%%%%%
\section{Many small attractors}\label{sec:many_att}
%%%%%%%%%%%%%%%%%%%%%%%%%%%%%%%%%%%%%%%%%%%%%%%%%%%%%%%%%%%%%%%%%%%

In this section we prove Theorem \ref{thm:many_att}, that taking $f\in F(n)$ uniformly at random, the probability that there exists $h\sim f$ such that $\A(h)$ has at least $0.046\cdot 2^n$ attractors, each of size $\leq 4$, tends to $1$ as $n\to\infty$. We start by showing that if the synchronous dynamics has many limit cycles of length $\leq 2$ and many vertex-disjoint paths of length $2$, then the asynchronous dynamics can have many small attractors.

\begin{lemma}\label{lem:C1C2P2}
Let $f \in F(n)$ and suppose that there exists a collection of $d\leq 2^{n-3}$ vertex-disjoint subgraphs of $\S(f)$, each isomorphic to $C_1$, $C_2$, or $2P_2$. There exists $h\sim f$ such that $\A(h)$ has at least $d$ attractors, each of size $\leq 4$.  
\end{lemma}

\begin{proof}
Let $A\subseteq \FP(f)$ and let $B$ be a set of periodic points of $f$ of period $2$ with $f(B)\cap B=\emptyset$, that is, no two configurations in $B$ belong to the same limit cycle of length $2$. Let $P^1,\dots,P^{2c}$ be vertex-disjoint paths of length $2$ contained in $\S(f)\setminus (A\cup B\cup f(B))$, and let $x^k,y^k,z^k$ be the vertices of $P^k$ in order. Let $d=|A|+|B|+c$ and suppose that $d\leq 2^{n-3}$; see Figure \ref{fig:2P2}(a) for an illustration. We have to prove that there exists $h\sim f$ such that $\A(h)$ has at least $d$ attractors. 

\medskip
Since $d\leq 2^{n-3}$ there is an injection $\phi$ from $A\cup B\cup [c]$ to the set of configurations $x\in\B^n$ with $x_1=x_2=x_3=0$. Let $h=\pi\circ f\circ\pi^{-1}$ where $\pi$ is any permutation of $\B^n$ such that:
\begin{itemize}
\item $\pi(a)=\phi(a)$ for all $a\in A$,
\item $\pi(b)=\phi(b)$ and $\pi(f(b))=\phi(b)+e_1$ for all $b\in B$,
\item $\pi(x^k)=\phi(k)+e_2$, $\pi(y^k)=\phi(k)$ and $\pi(z^k)=\phi(k)+e_{1,2}$ for all odd $k\in [2c]$,
\item $\pi(x^k)=\phi(k)+e_{1,3}$, $\pi(y^k)=\phi(k)+e_1$ and $\pi(z^k)=\phi(k)+e_3$ for all even $k\in[2c]$.
\end{itemize}
See Figure \ref{fig:2P2}(b) for an illustration. Then, for any $a\in A$, $b\in B$ and $k\in [c]$,  $\{\phi(a)\}$, $\{\ZERO,e_1\}+\phi(b)$  and $\{\ZERO,e_1,e_2,e_{1,3}\}+\phi(k)$ are attractors of $\A(h)$; see Figure \ref{fig:2P2}(c) for an illustration. Thus, $\A(h)$ has at least $d$ attractors.   
\end{proof}

\begin{figure}[t]
\[
\tag{a}
\begin{array}{c}
    \begin{tikzpicture}
 		\draw [draw=white] (-0.5,-0.2) rectangle (13.6,2.9);
		\node[inner sep=1,outer sep=1] (a1) at (0.00,1){\small$a^1$};
        \node[inner sep=1,outer sep=1] (a2) at (1.00,1){\small$a^2$};
        \node[inner sep=1,outer sep=1] (b1) at (2.00,1){\small$b^1$};
        \node[inner sep=1,outer sep=1] (B1) at (3.50,1){\small$b'^1$};
        \node[inner sep=1,outer sep=1] (b2) at (4.75,1){\small$b^2$};
        \node[inner sep=1,outer sep=1] (B2) at (6.25,1){\small$b'^2$};
        \node[inner sep=1,outer sep=1] (x1) at (8.00,2){\small$x^1$};
		\node[inner sep=1,outer sep=1] (y1) at (8.00,1){\small$y^1$};
		\node[inner sep=1,outer sep=1] (z1) at (8.00,0){\small$z^1$};
		\node[inner sep=1,outer sep=1] (x2) at (9.50,2){\small$x^2$};
		\node[inner sep=1,outer sep=1] (y2) at (9.50,1){\small$y^2$};
		\node[inner sep=1,outer sep=1] (z2) at (9.50,0){\small$z^2$};
		\node[inner sep=1,outer sep=1] (x3) at (11.0,2){\small$x^3$};
		\node[inner sep=1,outer sep=1] (y3) at (11.0,1){\small$y^3$};
		\node[inner sep=1,outer sep=1] (z3) at (11.0,0){\small$z^3$};
		\node[inner sep=1,outer sep=1] (x4) at (12.5,2){\small$x^4$};
		\node[inner sep=1,outer sep=1] (y4) at (12.5,1){\small$y^4$};
		\node[inner sep=1,outer sep=1] (z4) at (12.5,0){\small$z^4$};
		%labels
%		\node[inner sep=1,outer sep=1] (A) at (0.50,2.7){\small$A=\{a^1,a^2\}$};
%		\node[inner sep=1,outer sep=1] (B) at (4.125,2.7){\small$B=\{b^1,b^2\}$};
%		\node[inner sep=1,outer sep=1] (B) at (4.125,2.2){\small$f(B)=\{b'^1,b'^2\}$};
		\node[inner sep=1,outer sep=1] (B) at (8,2.7){\small$P_1$};
		\node[inner sep=1,outer sep=1] (B) at (9.5,2.7){\small$P_2$};
		\node[inner sep=1,outer sep=1] (B) at (11,2.7){\small$P_3$};
		\node[inner sep=1,outer sep=1] (B) at (12.5,2.7){\small$P_4$};
        \path[thick,->,draw,black]
        (b1) edge[bend left=40] (B1)
		(B1) edge[bend left=40] (b1)
		(b2) edge[bend left=40] (B2)
		(B2) edge[bend left=40] (b2)
		(x1) edge (y1)
		(y1) edge (z1)
		(x2) edge (y2)
		(y2) edge (z2)
		(x3) edge (y3)
		(y3) edge (z3)
		(x4) edge (y4)
		(y4) edge (z4)
        ;
        \draw[->,thick] (a1.-112) .. controls ({0-0.5},{1-0.7}) and ({0+0.5},{1-0.7}) .. (a1.-68);
        \draw[->,thick] (a2.-112) .. controls ({1-0.5},{1-0.7}) and ({1+0.5},{1-0.7}) .. (a2.-68);
    \end{tikzpicture}
\end{array}
\]
\[
\tag{b}
\begin{array}{c}
    \begin{tikzpicture}
 		\draw [draw=white] (-0.5,-0.2) rectangle (13.6,2.2);
		\node[inner sep=1,outer sep=1] (a1) at (0.00,1){\small$x^1$};
        \node[inner sep=1,outer sep=1] (a2) at (1.00,1){\small$x^2$};
        \node[inner sep=1,outer sep=1] (b1) at (2.00,1){\small$x^3$};
        \node[inner sep=1,outer sep=1] (B1) at (3.50,1){\small$x^3+e_1$};
        \node[inner sep=1,outer sep=1] (b2) at (4.75,1){\small$x^4$};
        \node[inner sep=1,outer sep=1] (B2) at (6.25,1){\small$x^4+e_1$};
        \node[inner sep=1,outer sep=1] (x1) at (8.00,2){\small$x^5+e_2$};
		\node[inner sep=1,outer sep=1] (y1) at (8.00,1){\small$x^5$};
		\node[inner sep=1,outer sep=1] (z1) at (8.00,0){\small$x^5+e_{1,2}$};
		\node[inner sep=1,outer sep=1] (x2) at (9.50,2){\small$x^5+e_{1,3}$};
		\node[inner sep=1,outer sep=1] (y2) at (9.50,1){\small$x^5+e_1$};
		\node[inner sep=1,outer sep=1] (z2) at (9.50,0){\small$x^5+e_3$};
		\node[inner sep=1,outer sep=1] (x3) at (11.0,2){\small$x^6+e_2$};
		\node[inner sep=1,outer sep=1] (y3) at (11.0,1){\small$x^6$};
		\node[inner sep=1,outer sep=1] (z3) at (11.0,0){\small$x^6+e_{1,2}$};
		\node[inner sep=1,outer sep=1] (x4) at (12.5,2){\small$x^6+e_{1,3}$};
		\node[inner sep=1,outer sep=1] (y4) at (12.5,1){\small$x^6+e_1$};
		\node[inner sep=1,outer sep=1] (z4) at (12.5,0){\small$x^6+e_3$};
        \path[thick,->,draw,black]
        (b1) edge[bend left=40] (B1)
		(B1) edge[bend left=40] (b1)
		(b2) edge[bend left=40] (B2)
		(B2) edge[bend left=40] (b2)
		(x1) edge (y1)
		(y1) edge (z1)
		(x2) edge (y2)
		(y2) edge (z2)
		(x3) edge (y3)
		(y3) edge (z3)
		(x4) edge (y4)
		(y4) edge (z4)
        ;
        \draw[->,thick] (a1.-112) .. controls ({0-0.5},{1-0.7}) and ({0+0.5},{1-0.7}) .. (a1.-68);
        \draw[->,thick] (a2.-112) .. controls ({1-0.5},{1-0.7}) and ({1+0.5},{1-0.7}) .. (a2.-68);
    \end{tikzpicture}
\end{array}
\]
\[
\tag{c}
\begin{array}{c}
    \begin{tikzpicture}
    \def\x{1.3}
    \def\z{0.6}
	\def\zz{0.9}
 		\draw [draw=white] (-0.5,0.8) rectangle (13.6,2.7);
		\node[inner sep=1,outer sep=1] (a1) at (0.00,{1+\x/2}){\small$x^1$};
        \node[inner sep=1,outer sep=1] (a2) at (1.00,{1+\x/2}){\small$x^2$};
        \node[inner sep=1,outer sep=1] (b1) at (2.00,{1+\x/2}){\small$x^3$};
        \node[inner sep=1,outer sep=1] (B1) at (3.50,{1+\x/2}){\small$x^3+e_1$};
        \node[inner sep=1,outer sep=1] (b2) at (4.75,{1+\x/2}){\small$x^4$};
        \node[inner sep=1,outer sep=1] (B2) at (6.25,{1+\x/2}){\small$x^4+e_1$};
        \path[thick,<->,draw,blue]
        (b1) edge (B1)
		(b2) edge (B2)
        ;
		\node[inner sep=1] (2) at ({8},{1+\x}){\small $x^5+e_2$};
		\node[inner sep=1] (0) at ({8},1){\small $x^5$};
		\node[inner sep=1] (13) at ({8+\x+\z},{1+\zz}){\small $x^5+e_{1,3}$};
		\node[inner sep=1] (1) at ({8+\x},1){\small $x^5+e_1$};
		\path[thick,<->,blue]
		(2) edge (0)
		(0) edge (1)
		(1) edge (13)
		;
		\node[inner sep=1] (2) at ({11},{1+\x}){\small $x^6+e_2$};
		\node[inner sep=1] (0) at ({11},1){\small $x^6$};
		\node[inner sep=1] (13) at ({11+\x+\z},{1+\zz}){\small $x^6+e_{1,3}$};
		\node[inner sep=1] (1) at ({11+\x},1){\small $x^6+e_1$};
		\path[thick,<->,blue]
		(2) edge (0)
		(0) edge (1)
		(1) edge (13)
		;
    \end{tikzpicture}
\end{array}
\]
{\caption{\label{fig:2P2} Illustration for Lemma \ref{lem:C1C2P2}. (a) A set $A=\{a^1,a^2\}$ of fixed points, a set $B=\{b^1,b^2\}$ of periodic configurations of period two, with $f(B)=\{b'^1,b'^2\}$ disjoint from $B$, and a collection of $4$ disjoint paths of length $2$, with notations as described in the lemma; hence $c=2$ and $d=6$. (b) Relabelling of the configurations according to the injection $\phi$ that sends $a^1,a^2,b^1,b^2,1,2$ on $x^1,x^2,x^3,x^4,x^5,x^6$, respectively, where $x^i_1=x^i_2=x^i_3=0$ for $1\leq i\leq 6$. (c) Resulting $d$ attractors in the asynchronous graph.}}
\end{figure}

We now prove that if $f^2$ has many images, then $\S(f)$ has a large collection of vertex-disjoint subgraphs isomorphic to $C_1$, $C_2$ or $2P_2$, and thus, by the preceding lemma, the asynchronous dynamics can have many small attractors.  

\begin{lemma}\label{lem:many_att}
Let $f \in F(n)$ and let $d$ be the number of images of $f^2$. There exists $h\sim f$ such that $\A(h)$ has at least $\lfloor d/10\rfloor$ attractors, each of size $\leq 4$. 
\end{lemma}

\begin{proof}
Since $\lfloor d/10\rfloor\leq 2^n/10\leq 2^{n-3}$, by Lemma~\ref{lem:C1C2P2}, it is sufficient to prove that there exists a collection of at least $\lfloor d/10\rfloor$ vertex-disjoint subgraphs of $\S(f)$, each isomorphic to $C_1$, $C_2$ or $2P_2$. Let $A$ be the fixed points of $f$, let $B$ be the set of periodic points of $f$ of period $2$ (which has even size), and let $\S'=\S(f)\setminus (A\cup B)$. It is sufficient to prove that $\S'$ has at least $\lfloor d/10\rfloor -|A|-|B|/2$ vertex-disjoint subgraphs isomorphic to $2P_2$ or, equivalently, $\lfloor d/5\rfloor -2|A|-|B|$ vertex-disjoint paths of length $2$. 

\medskip
Let $X$ be the image set of $f^2$ and for each $x\in X$, let $\tilde x$ be a pre-image of $x$ by $f^2$. Let $X^0=X\setminus (A\cup B)$. Since $\FP(f^2)=A\cup B$, for each $x\in X^0$ we have $\tilde x\neq x$ and thus $\tilde x\to f(\tilde x)\to x$ is a path of $\S'$ of length $2$, denoted $P_x$. Given $y\in X^0$, the paths $P_x$ and $P_y$ are vertex-disjoint if and only if $y\not\in\{\tilde x,f(\tilde x),x,f(x),f^2(x)\}$. Thus, we find a collection of vertex-disjoint paths of length $2$ in $\S'$ by selecting $x\in X^0$, removing $\{\tilde x,f(\tilde x),x,f(x),f^2(x)\}$ from the set $X^0$, and repeating the process until the set is empty. More precisely, let us write $|X^0|=5k+r$ where $0\leq r<5$ is an integer, and, for $1\leq i\leq k$, consider a sequence $x^1,\dots,x^k$ of configurations in $X^0$, and a sequence $X^1,\dots,X^k$ of subsets of $X^0$ such that, for $i=1,\dots,k$:
\begin{itemize}
\item
$x^i\in X^{i-1}$,
\item
$X^i=X^{i-1}\setminus\{\tilde x^i,f(\tilde x^i),x^i,f(x^i),f^2(x^i)\}$. 
\end{itemize}
These sequences exist since $|X^{i-1}|\geq |X^0|-5(i-1)\geq |X^0|-5(k-1)=5-r\geq 1$. By the argument above, the paths $P_{x^1},\dots,P_{x^k}$ are vertex-disjoint, and we are done since 
\[
k= \left\lfloor \frac{|X^0|}{5}\right\rfloor=\left\lfloor \frac{d-|A|-|B|}{5}\right\rfloor\geq \left\lfloor \frac{d}{5}\right\rfloor -2|A|-|B|.
\] 
\end{proof}

Taking $\phi:[N]\to [N]$ uniformly at random, it is proved in \cite{DS97} that the limit distribution of the number of images of $\phi^k$ is a Gaussian distribution with mean value $(1-\tau_k)N$, where $\tau_0=0$ and $\tau_{k+1}=e^{-1+\tau_k}$ (the mean was previously determined in \cite{FO89}), and with variance $c_k\cdot N$ for some explicit constant $c_k>0$ that only depends on $k$. Hence, taking $f\in F(n)$ uniformly at random, the limit distribution of the number of images of $f^2$ is a Gaussian distribution with mean value $(1-e^{-1+e^{-1}})2^n\simeq 0.468\cdot 2^n$ and variance $c_2\cdot 2^n$. Hence, for any $d<(1-e^{-1+e^{-1}})$, the probability that $f^2$ has at least $d\cdot 2^n$ images tends to $1$ as $n\to\infty$. Combining this with Lemma~\ref{lem:many_att}, we obtain Theorem~\ref{thm:many_att}.

\begin{remark}\label{rem:4}
The upper bound $4$ on attractors sizes in Theorems \ref{thm:small_att} and \ref{thm:many_att} is tight. Indeed, suppose that $f$ is a permutation without limit cycles of length $\leq 2$. Then it is clear that all the attractors of $\A(f)$ are of size $\geq 3$. If $\A(f)$ has an attractors $A$ of size $3$, then $A=\{x,x+e_i,x+e_j\}$ for some $x\in\B^n$ and distinct $i,j\in [n]$. But then, since $A$ is an attractor, we have $f(x+e_i)=f(x+e_j)=x$, a contradiction. Thus all the attractors of $\A(f)$ are of size $\geq 4$. 
\end{remark}

%%%%%%%%%%%%%%%%%%%%%%%%%%%%%%%%%%%%%%%%%%%%%%%%%%%%%%%%%%%%%%%%%%%
\section{One big attractor}\label{sec:strong}
%%%%%%%%%%%%%%%%%%%%%%%%%%%%%%%%%%%%%%%%%%%%%%%%%%%%%%%%%%%%%%%%%%%

For which $f\in F(n)$ does there exist $h\sim f$ such that $\A(h)$ has just one attractor, which is large? If $f$ has a fixed point, then it is clear that $\A(h)$ has an attractor of size one for every $h\sim f$. So we have to forbid fixed points. In this section, we show that this basic necessary condition is sufficient when $f$ is a permutation: if $f$ is a permutation without fixed point (derangement), then there exists $h\sim f$ such that $\A(h)$ is strongly connected, that is, $\A(h)$ has an attractor spanning the $2^n$ configurations; this is Theorem~\ref{thm:strong}. 

\medskip
The proof is based on the existence of a particular coloring of the synchronous graph. We need some definitions.  Given digraphs $G,H$, an \EM{$H$-coloring} of $G$ is a map $\phi:V(G)\to V(H)$ such that $\phi(u)\phi(v)\in E(H)$ for all $uv\in E(G)$; such a map is also often called a homomorphism from $G$ to $H$. An $H$-coloring $\phi$ of $G$ is \EM{balanced} if $|\phi^{-1}(u)|=|\phi^{-1}(v)|$ for any $u,v\in V(H)$, that is, color classes have the same size. Let $H_4$ be the following digraph:
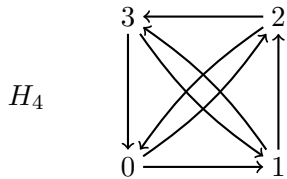
\[
H_4
\qquad
\begin{array}{c}
\begin{tikzpicture}
\node[outer sep=1,inner sep=2] (00) at (0,0){$0$};
\node[outer sep=1,inner sep=2] (01) at (2,0){$1$};
\node[outer sep=1,inner sep=2] (10) at (0,2){$3$};
\node[outer sep=1,inner sep=2] (11) at (2,2){$2$};
\path[->,thick]
(00) edge (01)
(00) edge[bend right=10] (11)
(01) edge (11)
(01) edge[bend right=10] (10)
(11) edge (10)
(11) edge[bend right=10] (00)
(10) edge (00)
(10) edge[bend right=10] (01)
;
\end{tikzpicture}
\end{array}
\]

\medskip
We prove below that if $f$ is a derangement, that is if $\S(f)$ consists of a disjoint union of cycles of length $\geq 2$, then there exists a balanced $H_4$-coloring of $\S(f)$. For the purpose of induction, we prove the statement for a larger class of disjoint unions of cycles than just those on $2^n$ vertices.

\begin{lemma}\label{lem:coloring}
If $G$ is a digraph on $4n$ vertices that consists of a disjoint union of cycles, each of length $\geq 2$, then $G$ has a balanced $H_4$-coloring. 
\end{lemma}

\begin{proof}
In this proof, for every integer $a$ we set $[a]=a\mod 4$. By {\em coloring} we always mean $H_4$-coloring. We denote by $C=\{0,1,2,3\}$ the set of possible colors. Given $A\subseteq C$, we say that a function $\phi:V\to C$ is {\em $A$-exceeding} if $|\phi^{-1}(a)|=|\phi^{-1}(b)|+1$ for all $a\in A$ and $b\in C\setminus A$; and $\phi$ is {\em $A$-defecting} if it is ($C\setminus A$)-exceeding. 

\medskip
Let $C_\ell$ be a cycle of length $\ell\geq 2$ with vertices $v_0,\dots,v_{\ell-1}$ in order. We say that $C_\ell$ is of {\em type} $t$ if $t=[\ell]$. Let us prove that for any $c\in C$,  
\begin{enumerate}
\item \label{item:1}
if $C_\ell$ is of type $0$ then it has a balanced coloring;
\item \label{item:2}
if $C_\ell$ is of type $1$ then it has an $\{c\}$-exceeding coloring;
\item \label{item:3}
if $C_\ell$ is of type $3$ then it has an $\{c\}$-defecting coloring;
\item \label{item:4}
if $C_\ell$ is of type $2$ then it has an $\{c,[c+2]\}$-exceeding coloring. 
\end{enumerate}
Indeed, for the type $0$, color $v_i$ with $[i]$. 
For the type $1$, color $v_i$ with $[c+i]$ for $0\leq i< \ell-5$ and color $v_{\ell-5},v_{\ell-4},v_{\ell-3},v_{\ell-2},v_{\ell-1}$ with $c,[c+2],c,[c+1],[c+3]$ respectively.
For the type $3$, color $v_i$ with $[c+i]$ for $0\leq i< \ell-3$ and color $v_{\ell-3},v_{\ell-2},v_{\ell-1}$ with $[c+1],[c+2],[c+3]$ respectively.
 For the type $2$, color $v_i$ with $[c+i]$ for $0\leq i< \ell-2$ and color $v_{\ell-2},v_{\ell-1}$ with $c,[c+2]$ respectively. See Figure~\ref{fig:coloring} for an illustration.

\medskip
We deduce that the following digraphs have a balanced coloring:
\begin{enumerate}
\item \label{subgraph1}
a cycle of type $0$, as shown above;
\item \label{subgraph2}
a disjoint union of $4$ cycles of type $1$: for each $c\in C$, color one of the $4$ cycles with a $\{c\}$-exceeding coloring;
\item \label{subgraph3}
a disjoint union of $4$ cycles of type $3$: for each $c\in C$, color one of the $4$ cycles with a $\{c\}$-defecting coloring;
\item \label{subgraph4}
a disjoint union of $2$ cycles of type $2$: color one cycle with a $\{0,2\}$-exceeding coloring and the other one with a $\{1,3\}$-exceeding coloring;
\item \label{subgraph5}
a disjoint union of $1$ cycles of type $1$ and $1$ cycle of type $3$: color the first one with a $\{0\}$-exceeding coloring and the other one with a $\{0\}$-defecting coloring;
\item \label{subgraph6}
a disjoint union of $2$ cycles of type $1$ and $1$ cycle of type $2$: color the two first ones with a $\{0\}$-exceeding coloring and a $\{2\}$-exceeding coloring, and color the last with a $\{1,3\}$-exceeding coloring; 
\item \label{subgraph7}
a disjoint union of $2$ cycles of type $3$ and $1$ cycle of type $2$: color the two first ones with a $\{0\}$-defecting coloring and a $\{2\}$-defecting coloring, and color the last with a $\{0,2\}$-exceeding coloring.
\end{enumerate}

Let $G$ be a non-empty digraph on $4n$ vertices that consists of a disjoint union of cycles, each of length $\geq 2$.
Let us prove that $G$ contains, as a subgraph, one of the digraphs listed above. Suppose, for a contradiction, that this is false. Let $a_t$ be the number of cycles of type $t$ in $G$. Since $G$ does not contain the digraphs described in points \ref{subgraph1},\ref{subgraph2},\ref{subgraph3},\ref{subgraph4}, we have $a_0 = 0$, $a_1,a_3 \leq 3$, and $a_2 \leq 1$. Furthermore, since $G$ does not contain the digraphs described in point~\ref{subgraph5} either $a_1 = 0$ or $a_3=0$. If $a_1 = 0$ (\textit{resp.} $a_3 = 0$) then all the odd cycles of $G$ are of type $3$ (\textit{resp.} $1$) and therefore $a_3$ (\textit{resp.} $a_1$) is even (since $G$ has an even number of vertices). Therefore, $(a_1,a_3) \in \{ (0,0) , (2,0) , (0,2) \}$. Note that since $G$ contains $4n$ vertices, $[a_1+2a_2+3a_3] = 0$. Suppose that $(a_1,a_3) = (0,0)$ then $0 = [a_1+2a_2+3a_3] = [2a_2]$ and therefore $a_2=0$, so $G$ is empty, a contradiction. Now suppose that $(a_1,a_3) = (2,0)$ (\textit{resp.} $(a_1,a_3) = (0,2)$).  Since $G$ does not contain the digraphs described in points~\ref{subgraph6} (\textit{resp.} \ref{subgraph7}), we have $a_2=0$. Therefore,  $[a_1+2a_2+3a_3] = [a_1] = [2] = 2$ ( \textit{resp.} $[a_1+2a_2+3a_3] = [3a_3] = [6] = 2$) which is a contradiction. This proves that $G$ contains one of the listed digraphs. 

\medskip
Let us now prove, by induction on the number $m$ of cycles in $G$, that $G$ has a balanced coloring. 
If $m=1$ then $G$ is a cycle of type 0, which admits a balanced coloring. Suppose that $m>1$. As proved above, $G$ has a subgraph $G_1$, with $4n'$ vertices, which admits a balanced coloring. If $G=G_1$ we are done. Otherwise, let $G_2$ be the subgraph of $G$ obtained by deleting the cycles of $G_1$. Since $G_2$ has $4(n-n')$ vertices and at most $m-1$ cycles, by induction hypothesis, $G_2$ has a balanced coloring, and with the balanced coloring of $G_1$ we obtain a balanced coloring of~$G$.
\end{proof}

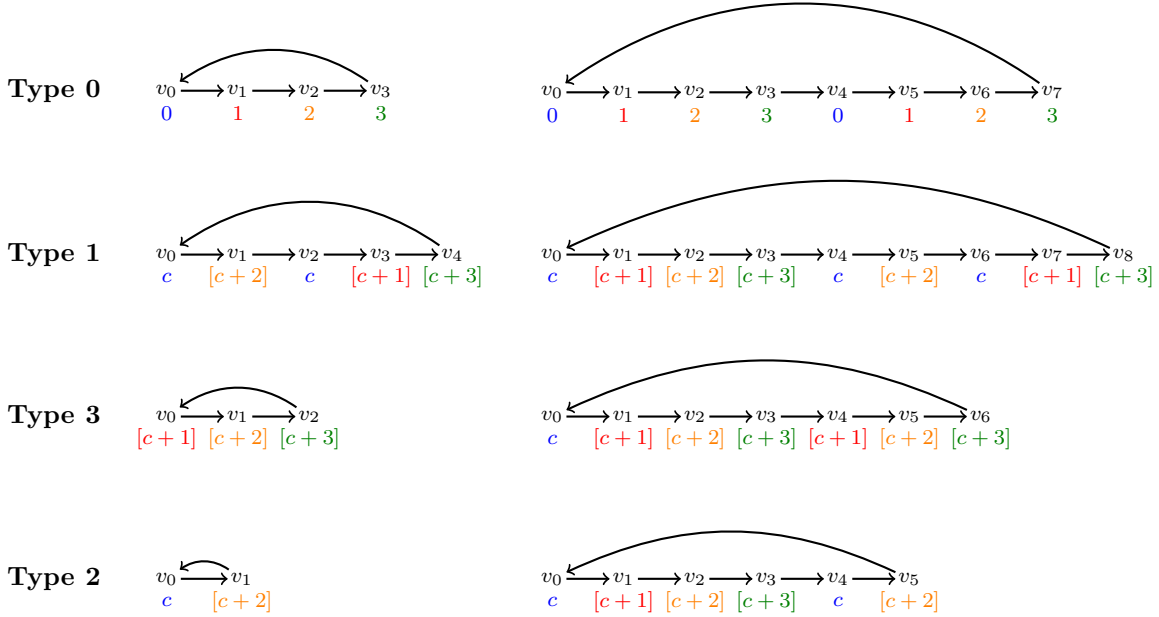
\begin{figure}[h]
\[
\def\x{0.95}
\begin{array}{ll}
	\begin{array}{c}
	%C4
    \begin{tikzpicture}
 	   \draw[white] (-0.5,-0.5) rectangle (3.5,1.5);
	   
		\node[inner sep=1] (t0) at (-1.5,0){\textbf{\small Type 0}};
 	   	
        \node[inner sep=1] (v0) at ({0*\x},0){\scriptsize $v_0$};
        \node[inner sep=1] (v1) at ({1*\x},0){\scriptsize $v_1$};
        \node[inner sep=1] (v2) at ({2*\x},0){\scriptsize $v_2$};
        \node[inner sep=1] (v3) at ({3*\x},0){\scriptsize $v_3$};
        
        \node[inner sep=1,color=blue] (c1) at 			({0*\x},-0.3){\scriptsize $0$};
        \node[inner sep=1,color=red] (c2) at 			({1*\x},-0.3){\scriptsize $1$};
        \node[inner sep=1,color=orange] (c3) at 		({2*\x},-0.3){\scriptsize $2$};
        \node[inner sep=1,color=green!50!black] (c4) at ({3*\x},-0.3){\scriptsize $3$};
        \path[thick,->]
        (v0) edge (v1)
        (v1) edge (v2)
        (v2) edge (v3)
        (v3) edge[bend right=35] (v0)
        ;
    \end{tikzpicture}
    \end{array}
    &
    %C8
    \begin{array}{c}
    \begin{tikzpicture}
 	   \draw[white] (0,-0.5) rectangle (3.5,1.5);
 	   	
        \node[inner sep=1] (v0) at ({0*\x},0){\scriptsize $v_0$};
        \node[inner sep=1] (v1) at ({1*\x},0){\scriptsize $v_1$};
        \node[inner sep=1] (v2) at ({2*\x},0){\scriptsize $v_2$};
        \node[inner sep=1] (v3) at ({3*\x},0){\scriptsize $v_3$};
        \node[inner sep=1] (v4) at ({4*\x},0){\scriptsize $v_4$};
        \node[inner sep=1] (v5) at ({5*\x},0){\scriptsize $v_5$};
        \node[inner sep=1] (v6) at ({6*\x},0){\scriptsize $v_6$};
        \node[inner sep=1] (v7) at ({7*\x},0){\scriptsize $v_7$};
 
        \node[inner sep=1,color=blue] (c1) at 			({0*\x},-0.3){\scriptsize $0$};
        \node[inner sep=1,color=red] (c2) at 			({1*\x},-0.3){\scriptsize $1$};
        \node[inner sep=1,color=orange] (c3) at 		({2*\x},-0.3){\scriptsize $2$};
        \node[inner sep=1,color=green!50!black] (c4) at ({3*\x},-0.3){\scriptsize $3$};
        \node[inner sep=1,color=blue] (c5) at 			({4*\x},-0.3){\scriptsize $0$};
        \node[inner sep=1,color=red] (c6) at 			({5*\x},-0.3){\scriptsize $1$};
        \node[inner sep=1,color=orange] (c7) at 		({6*\x},-0.3){\scriptsize $2$};
        \node[inner sep=1,color=green!50!black] (c8) at ({7*\x},-0.3){\scriptsize $3$};
        \path[thick,->]
        (v0) edge (v1)
        (v1) edge (v2)
        (v2) edge (v3)
        (v3) edge (v4)
        (v4) edge (v5)
        (v5) edge (v6)
        (v6) edge (v7)
        (v7) edge[bend right=35] (v0)
        ;
    \end{tikzpicture}
    \end{array}
    \\ 
	\begin{array}{c}
	%C5
    \begin{tikzpicture}
 	   \draw[white] (-0.5,-0.5) rectangle (3.5,1.5);
	   
	   	\node[inner sep=1] (t1) at (-1.5,0){\textbf{\small Type 1}};
 	   	
        \node[inner sep=1] (v0) at ({0*\x},0){\scriptsize $v_0$};
        \node[inner sep=1] (v1) at ({1*\x},0){\scriptsize $v_1$};
        \node[inner sep=1] (v2) at ({2*\x},0){\scriptsize $v_2$};
        \node[inner sep=1] (v3) at ({3*\x},0){\scriptsize $v_3$};
        \node[inner sep=1] (v4) at ({4*\x},0){\scriptsize $v_4$};

        \node[inner sep=1,color=blue] (c0) at 				({0*\x},-0.3){\scriptsize $c$};
        \node[inner sep=1,color=orange] (c2) at 			({1*\x},-0.3){\scriptsize $[c+2]$};
        \node[inner sep=1,color=blue] (c3) at 				({2*\x},-0.3){\scriptsize $c$};
        \node[inner sep=1,color=red] (c4) at 				({3*\x},-0.3){\scriptsize $[c+1]$};
          \node[inner sep=1,color=green!50!black] (c4) at 	({4*\x},-0.3){\scriptsize $[c+3]$};

        \path[thick,->]
        (v0) edge (v1)
        (v1) edge (v2)
        (v2) edge (v3)
        (v3) edge (v4)
        (v4) edge[bend right=35] (v0)
        ;
    \end{tikzpicture}
    \end{array}
    &
    %C9
    \begin{array}{c}
    \begin{tikzpicture}
 	   \draw[white] (0,-0.5) rectangle (3.5,1.5);
 	   	
        \node[inner sep=1] (v0) at ({0*\x},0){\scriptsize $v_0$};
        \node[inner sep=1] (v1) at ({1*\x},0){\scriptsize $v_1$};
        \node[inner sep=1] (v2) at ({2*\x},0){\scriptsize $v_2$};
        \node[inner sep=1] (v3) at ({3*\x},0){\scriptsize $v_3$};
        \node[inner sep=1] (v4) at ({4*\x},0){\scriptsize $v_4$};
        \node[inner sep=1] (v5) at ({5*\x},0){\scriptsize $v_5$};
        \node[inner sep=1] (v6) at ({6*\x},0){\scriptsize $v_6$};
        \node[inner sep=1] (v7) at ({7*\x},0){\scriptsize $v_7$};
        \node[inner sep=1] (v8) at ({8*\x},0){\scriptsize $v_8$};

        \node[inner sep=1,color=blue] (c1) at 				({0*\x},-0.3){\scriptsize $c$};
        \node[inner sep=1,color=red] (c2) at 				({1*\x},-0.3){\scriptsize $[c+1]$};
        \node[inner sep=1,color=orange] (c3) at 			({2*\x},-0.3){\scriptsize $[c+2]$};
        \node[inner sep=1,color=green!50!black] (c4) at 	({3*\x},-0.3){\scriptsize $[c+3]$};
        \node[inner sep=1,color=blue] (c0) at 				({4*\x},-0.3){\scriptsize $c$};
        \node[inner sep=1,color=orange] (c2) at 			({5*\x},-0.3){\scriptsize $[c+2]$};
        \node[inner sep=1,color=blue] (c3) at 				({6*\x},-0.3){\scriptsize $c$};
        \node[inner sep=1,color=red] (c4) at 				({7*\x},-0.3){\scriptsize $[c+1]$};
          \node[inner sep=1,color=green!50!black] (c4) at 	({8*\x},-0.3){\scriptsize $[c+3]$};
        \path[thick,->]
        (v0) edge (v1)
        (v1) edge (v2)
        (v2) edge (v3)
        (v3) edge (v4)
        (v4) edge (v5)
        (v5) edge (v6)
        (v6) edge (v7)
        (v7) edge (v8)
        (v8) edge[bend right=25] (v0)
        ;
    \end{tikzpicture}
    \end{array}
    \\
	\begin{array}{c}
	%C3
    \begin{tikzpicture}
 	   \draw[white] (-0.5,-0.5) rectangle (3.5,1.5);
	   
		\node[inner sep=1] (t3) at (-1.5,0){\textbf{\small Type 3}};
 	   	
        \node[inner sep=1] (v0) at ({0*\x},0){\scriptsize $v_0$};
        \node[inner sep=1] (v1) at ({1*\x},0){\scriptsize $v_1$};
        \node[inner sep=1] (v2) at ({2*\x},0){\scriptsize $v_2$};

        \node[inner sep=1,color=red] (c0) at 			({0*\x},-0.3){\scriptsize $[c+1]$};
        \node[inner sep=1,color=orange] (c1) at 		({1*\x},-0.3){\scriptsize $[c+2]$};
        \node[inner sep=1,color=green!50!black] (c2) at ({2*\x},-0.3){\scriptsize $[c+3]$};

        \path[thick,->]
        (v0) edge (v1)
        (v1) edge (v2)
        (v2) edge[bend right=35] (v0)
        ;
    \end{tikzpicture}
    \end{array}
    &
    %C7
    \begin{array}{c}
    \begin{tikzpicture}
 	   \draw[white] (0,-0.5) rectangle (3.5,1.5);
 	   	
        \node[inner sep=1] (v0) at ({0*\x},0){\scriptsize $v_0$};
        \node[inner sep=1] (v1) at ({1*\x},0){\scriptsize $v_1$};
        \node[inner sep=1] (v2) at ({2*\x},0){\scriptsize $v_2$};
        \node[inner sep=1] (v3) at ({3*\x},0){\scriptsize $v_3$};
        \node[inner sep=1] (v4) at ({4*\x},0){\scriptsize $v_4$};
        \node[inner sep=1] (v5) at ({5*\x},0){\scriptsize $v_5$};
        \node[inner sep=1] (v6) at ({6*\x},0){\scriptsize $v_6$};

        \node[inner sep=1,color=blue] (c0) at 			({0*\x},-0.3){\scriptsize $c$};
        \node[inner sep=1,color=red] (c1) at 			({1*\x},-0.3){\scriptsize $[c+1]$};
        \node[inner sep=1,color=orange] (c2) at 		({2*\x},-0.3){\scriptsize $[c+2]$};
        \node[inner sep=1,color=green!50!black] (c3) at ({3*\x},-0.3){\scriptsize $[c+3]$};
        \node[inner sep=1,color=red] (c4) at 			({4*\x},-0.3){\scriptsize $[c+1]$};
        \node[inner sep=1,color=orange] (c5) at 		({5*\x},-0.3){\scriptsize $[c+2]$};
        \node[inner sep=1,color=green!50!black] (c6) at ({6*\x},-0.3){\scriptsize $[c+3]$};
        \path[thick,->]
        (v0) edge (v1)
        (v1) edge (v2)
        (v2) edge (v3)
        (v3) edge (v4)
        (v4) edge (v5)
        (v5) edge (v6)
        (v6) edge[bend right=25] (v0)
        
        ;
    \end{tikzpicture}
    \end{array}
    \\ 
	\begin{array}{c}
	%C2
    \begin{tikzpicture}
 	   \draw[white] (-0.5,-0.5) rectangle (3.5,1.5);
	   
		\node[inner sep=1] (t2) at (-1.5,0){\textbf{\small Type 2}};
 	   	
        \node[inner sep=1] (v0) at (0,0){\scriptsize $v_0$};
        \node[inner sep=1] (v1) at (1,0){\scriptsize $v_1$};

        \node[inner sep=1,color=blue] (c0) at (0,-0.3){\scriptsize $c$};
        \node[inner sep=1,color=orange] (c2) at (1,-0.3){\scriptsize $[c+2]$};

        \path[thick,->]
        (v0) edge (v1)
        (v1) edge[bend right=35] (v0)
        ;
    \end{tikzpicture}
    \end{array}
    &
    %C6
    \begin{array}{c}
    \begin{tikzpicture}
 	   \draw[white] (0,-0.5) rectangle (3.5,1.5);
 	   	
        \node[inner sep=1] (v0) at ({0*\x},0){\scriptsize $v_0$};
        \node[inner sep=1] (v1) at ({1*\x},0){\scriptsize $v_1$};
        \node[inner sep=1] (v2) at ({2*\x},0){\scriptsize $v_2$};
        \node[inner sep=1] (v3) at ({3*\x},0){\scriptsize $v_3$};
        \node[inner sep=1] (v4) at ({4*\x},0){\scriptsize $v_4$};
        \node[inner sep=1] (v5) at ({5*\x},0){\scriptsize $v_5$};

        \node[inner sep=1,color=blue] (c1) at 			({0*\x},-0.3){\scriptsize $c$};
        \node[inner sep=1,color=red] (c2) at 			({1*\x},-0.3){\scriptsize $[c+1]$};
        \node[inner sep=1,color=orange] (c3) at 		({2*\x},-0.3){\scriptsize $[c+2]$};
        \node[inner sep=1,color=green!50!black] (c4) at ({3*\x},-0.3){\scriptsize $[c+3]$};
        \node[inner sep=1,color=blue] (c0) at 			({4*\x},-0.3){\scriptsize $c$};
        \node[inner sep=1,color=orange] (c2) at 		({5*\x},-0.3){\scriptsize $[c+2]$};
        \path[thick,->]
        (v0) edge (v1)
        (v1) edge (v2)
        (v2) edge (v3)
        (v3) edge (v4)
        (v4) edge (v5)
        (v5) edge[bend right=25] (v0)
        ;
    \end{tikzpicture}
    \end{array}
\end{array}
\]
{\caption{\label{fig:coloring} Examples of coloring for the four types of cycles.}}
\end{figure}

We can now prove the following strengthening of Theorem \ref{thm:strong}.

\begin{theorem}
Let $f\in F(n)$ be a derangement. There exists $h\sim f$ such that, for all $x,y\in\B^n$, $\A(h)$ has a path from $x$ to $y$ of length at most $d(x,y)+4$. 
\end{theorem}

\begin{proof}
If $n=2$, then $\S(f)$ is isomorphic to $2C_2$ or $C_4$; in the first case, take $h\in F(2)$ defined by $h(x)=\overline{x}$ and in the second case, take $h\in F(2)$ defined by $h(x)=(x_2,\overline{x_1})$. So suppose that $n\geq 3$. Then, $\S(f)$ is a digraph on $4\cdot 2^{n-2}$ vertices that consists of a disjoint union of cycles of length $\geq 2$. Hence, by Lemma~\ref{lem:coloring}, $\S(f)$ has a balanced $H_4$-coloring $\phi$. For all $a,b\in\B$, let $X_{ab}$ be the set of $x\in\B^n$ with $x_1=a$ and $x_2=b$. Let $h=\pi\circ f\circ\pi^{-1}$ where $\pi$ is any permutation of $\B^n$ such that, for all $x\in\B^n$, 
\begin{itemize}
\item
if $\phi(x)=0$ then $\pi(x)\in X_{00}$, 
\item
if $\phi(x)=1$ then $\pi(x)\in X_{10}$, 
\item
if $\phi(x)=2$ then $\pi(x)\in X_{11}$, 
\item
if $\phi(x)=3$ then $\pi(x)\in X_{01}$, 
\item
if $\phi(x)=0$ then $\pi(f(x))_i\neq\pi(x)_i$ for all $3\leq i\leq n$.
\end{itemize}
The first four points are possible since the $4$ colors classes and $\{X_{00},X_{01},X_{10},X_{11}\}$ are both balanced partitions of the configurations; and the last point is possible because $f$ is a derangement and the given constraints are independent of the previous ones.   

\medskip
Note that, for any $x\in X_{00}$, $\A(h)$ contains the cycle of length $4$ whose vertices are $x,x+e_1,x+e_{1,2},x+e_2$ in order. Indeed, since $x\in X_{00}$, setting $x'=\pi^{-1}(x)$, we have $\phi(x')=0$, so $\phi(f(x'))\in\{1,2\}$ (since $\phi$ is an $H_4$-coloring) thus $h(x)=\pi(f(x'))\in\{X_{10},X_{11}\}$. So $h_1(x)=1$ and thus $\A(h)$ has an arc from $x$ to $x+e_1$. We prove the presence of the three other arcs of the cycle similarly.

\medskip
Let $F_4$ be the cycle of length $4$ whose vertices are $00$, $10$, $01$, $11$ in order, and let $d_{F_4}(a,b)$ be the length of a shortest path from $a$ to $b$ in $F_4$. Note that, $d_{F_4}(a,00)+d_{F_4}(00,b)\leq d(a,b)+4$. Given $x,y\in\B^n$, we denote by $\delta(x,y)$ the number of $3\leq i\leq n$ with $x_i\neq y_i$ and set 
\[
D(x,y):=d_{F_4}(x_1x_2,00)+d_{F_4}(00,y_1y_2)+\delta(x,y)\leq d(x_1x_2,y_1y_2)+4+\delta(x,y)=d(x,y)+4.
\]

\medskip
We now prove that, for any $x,y\in\B^n$, $\A(h)$ has a path from $x$ to $y$ of length at most $D(x,y)$. We proceed by induction on $\delta(x,y)$. The argument above shows that if $\delta(x,y)=0$ then $\A(h)$ has a path from $x$ to $y$ contained in a cycle of length $4$ and this path is of length $d_{F_4}(x_1x_2,y_1y_2)\leq D(x,y)$. So suppose that $\delta(x,y)\geq 1$ and let $3\leq i\leq n$ with $x_i\neq y_i$. Let $z\in X_{00}$ with $z_j=x_j$ for all $3\leq j\leq n$. Since $\delta(x,z)=0$, by the same argument $\A(h)$ has a path from $x$ to $z$ of length $d_{F_4}(x_1x_2,00)$. Since $z\in X_{00}$, setting $z'=\pi^{-1}(z)$, we have  $\phi(z')=0$ and by the fifth point $h_i(z)=\pi(f(z'))_i\neq \pi(z')_i=z_i$. Thus, $\A(h)$ has an arc from $z$ to $z+e_i\in X_{00}$. Since $\delta(z+e_i,y)=\delta(z,y)-1=\delta(x,y)-1$, by induction, $\A(h)$ has a path from $z+e_i$ to $y$ of length at most $d_{F_4}(00,y_1y_2)+\delta(x,y)-1$. Combining, we obtain a path from $x$ to $y$ of length at most $D(x,y)$. 
\end{proof}

%%%%%%%%%%%%%%%%%%%%%%%%%%%%%%%%%%%%%%%%%%%%%%%%%%%%%%%%%%%%%%%%%%%
\section{Conclusion}\label{sec:conclusion}
%%%%%%%%%%%%%%%%%%%%%%%%%%%%%%%%%%%%%%%%%%%%%%%%%%%%%%%%%%%%%%%%%%%

In this paper, we study the two most classical dynamics derived from a Boolean network $f$, namely the synchronous dynamics $\S(f)$ and the asynchronous dynamics $\A(f)$. Since many dynamical properties of interest are invariant under isomorphism, we study these dynamics up to isomorphism. More precisely, we address the following question: what information about one dynamic can be inferred from the other when they are only known up to isomorphism ? We essentially show that the relation between the two dynamics is highly asymmetric: up to isomorphism, $\A(f)$ is much more informative about $\S(f)$ than the converse.

\medskip
Let us say that the synchronous reconstruction is possible when $\S(f)$ can be fully reconstructed from $\A(f)$; formally $\S(h)\sim \S(f)$ for every $h$ such that $\A(h)\sim\S(f)$. We prove that asynchronous reconstruction is possible with high probability under the uniform distribution (Theorem \ref{thm:A_to_S}). To this end, we introduce a sufficient condition for synchronous reconstruction: each edge of $Q_n$ is solid in $U\A(f)$, the undirected version of $\A(f)$. We do not know the complexity of testing this sufficient condition, and one may ask whether necessary conditions can be identified. Since this sufficient condition is somewhat difficult to analyse, we introduce a stronger and more concrete one: each edge of $Q_n$ has at least three staples in $U\A(f)$; this stronger condition can be easily checked in $O(n^22^{n-1})$, which is polynomial with respect to the size $N=2^n$ of the input asynchronous dynamics. We then prove that this stronger condition holds with high probability under the uniform distribution. In that case, each edge of $Q_n$ appears in $U\A(f)$ with probability $p=3/4$ and our proof extends to $p=0.72$. An open question is whether synchronous reconstruction remains possible with high probability for smaller values of $p$
or under other kinds of distributions.

\medskip
Conversely, the asynchronous reconstruction is never possible: for $n\geq 3$ and $f\neq\cst,\id$, there always exists $h$ such that $\S(h)\sim \S(f)$ but $\A(h)\not\sim \A(f)$ (Theorem \ref{thm:S_to_A}). Let us denote by $\tilde \A(f)$ the maximal set of pairwise non-isomorphic asynchronous dynamics $\A(h)$ such that $\S(h)\sim \S(f)$; hence the previous theorem says that $|\tilde\A(f)|\geq 2$. Our arguments to establish this are very local, as they only analyse a small portion of $\S(f)$. A more global analysis could significantly improve the lower bound of 2, possibly yielding a bound that increases with $n$. It is plausible that such a lower bound grows rapidly with $n$, which could be supported by showing that, with high probability, the size of $\tilde \A(f)$ is exponential in $n$.  

\medskip
We then consider the number of attractors in the asynchronous dynamics of $\tilde \A(f)$, that is, the set $\K(f)$ of integers $k$ such that $\tilde \A(f)$ has a member with $k$ attractors. It is straightforward that $\min \K(f)\geq \max(1,\fp(f))$ and we prove that this inequality is actually always an equality (Theorem~\ref{thm:small_att}). In the other direction, we prove that, under the uniform distribution, $\max \K(f)=\Omega(2^n)$ with high probability (Theorem~\ref{thm:many_att}). This provides a huge gap with $\min\K(f)$. It may be that, by combining the two proof techniques, one can show that there are many intermediate possibilities, that is, that $|\K(f)|$ is large, at least with high probability. This would emphasise that we cannot say much about the number of attractors of $\A(f)$ if we only know $\S(f)$ up to isomorphism. This could also provide a direction for proving that $|\tilde \A(f)|$ is large.

\medskip
We finally analyse the size of attractors in the asynchronous dynamics in $\tilde \A(f)$, that is, we study the set $\L(f)$ of integers $\ell$ such that $\tilde \A(f)$ has a member with an attractor of size $\ell$. We prove that $\min \L(f)\leq 4$ (Theorem~\ref{thm:small_att}), and that this lower bound is reached if $f$ is a permutation without cycle of length $\leq 2$ (Remark~\ref{rem:4}). We furthermore prove that $\max \L(f)=2^n$ if $f$ is a permutation without fixed point (Theorem \ref{thm:strong}). This results in a significant gap with $\min \K(f)$. It would be interesting to give a general lower bound on $\max \L(f)$ or to show that this maximum is exponential with high probability. Finally, for the same reasons as those discussed above, it would be interesting to analyse the size of $\L(f)$. 

\paragraph{Acknowledgments} This work has been partially funded by the HORIZON-MSCA-2022-
SE-01 project 101131549 ‘‘ACANCOS’’ project and the ANR-24-CE48-7504 ‘‘ALARICE’’ project.

\appendix 

\section{Theorem~\ref{thm:A4} for $n\leq 4$}\label{sec:n4}

For $n=1$ the theorem is obvious, and for $n=2$ this is an easy exercise. For $n=3$, this is not difficult, but annoying, and as, moreover, one can check this case easily by computer, we do not present the proof. The case $n=4$ is much more difficult to check by computer but, hopefully, there is a proof almost identical to the case $n\geq 5$. Indeed, one easily check that Lemma \ref{lem:unavoidable} holds for $n=4$ with $4$ exceptions.

\begin{lemma}\label{lem:unavoidable_n4}
Let $f\in F(4)$ without fixed point. Suppose that $f^2\neq\id$ and that $\S(f)$ is not isomorphic to $\P_5+2C_5$, $2C_5+C_6$, $C_4+2C_6$, or $4C_4$. There exists $g\sim f$ such that $g$ properly contains a closed or $0$-open pattern $\P$. 
\end{lemma}

Now we can prove, as in Section~\ref{sec:one_att}, that Theorem \ref{thm:A4} holds for $n=4$, with $4$ exceptions. Indeed, let $f\in F(4)$ without fixed point and suppose that $\S(f)$ is not isomorphic to $\P_5+2C_5$, $2C_5+C_6$, $C_4+2C_6$, or $4C_4$. If $f^2=\id$ then we are done by Lemma~\ref{lem:f^2=id}. Otherwise, by Lemma~\ref{lem:unavoidable_n4}, there exists $g\sim f$ such that $g$ properly contains a closed or $0$-open pattern $\P$, and we are done by Lemma \ref{lem:pattern_converging}. 

\medskip
It remains to treat the $4$ exceptions. This is done in Figure~\ref{fig:n4}. 

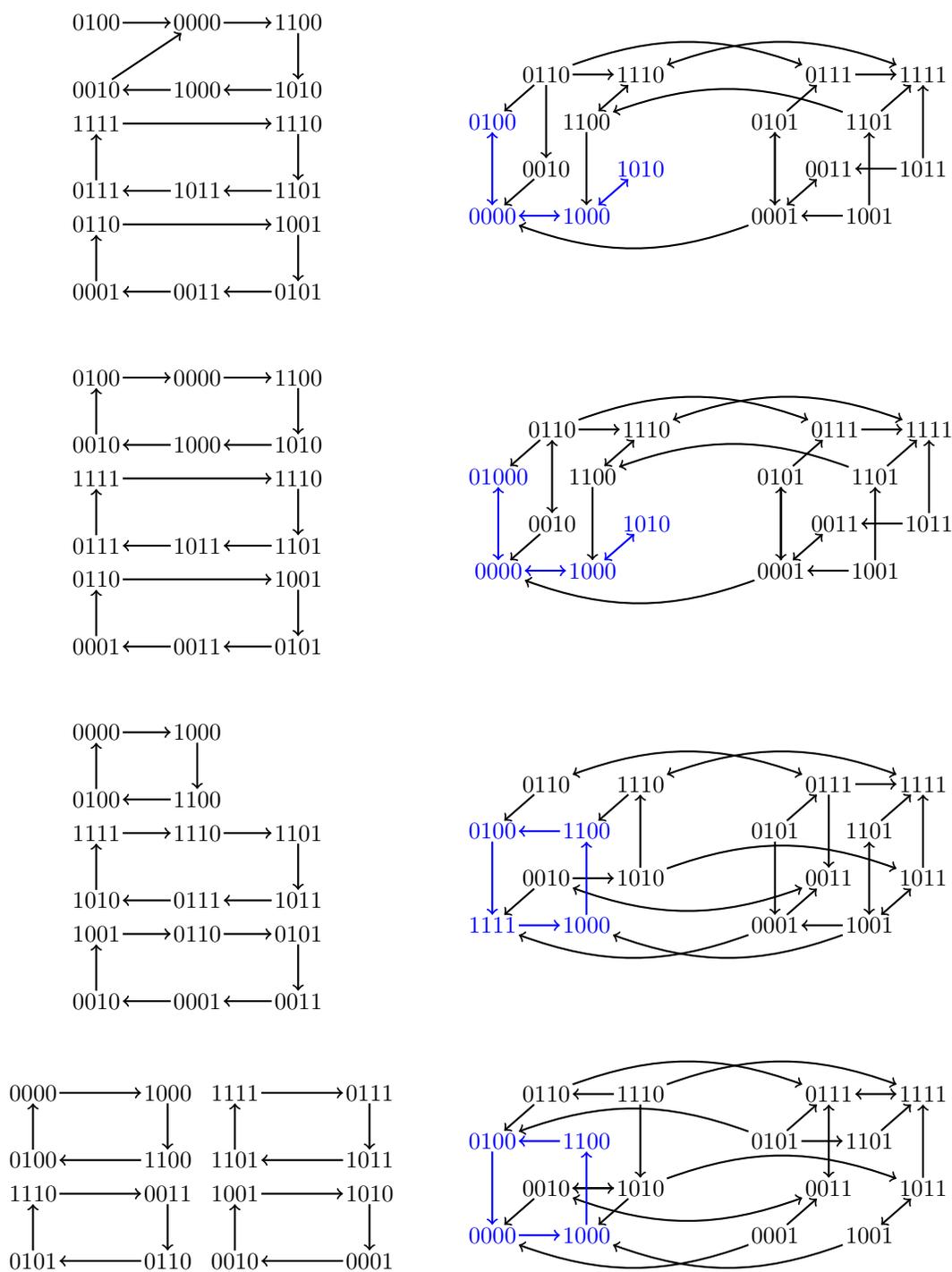
\begin{figure}[h]
\[
\def\x{1.4}
\def\z{0.8}
\def\zz{0.7}
\begin{array}{ccl}
	%\P_5+ 2C_5
    \begin{array}{c}
    \begin{tikzpicture}
    	\draw[white] (-0.5,-0.5) rectangle (3.5,1.5);
        \node[inner sep=1] (2) at (0,1){\small $0100$};
        \node[inner sep=1] (0) at (1.5,1){\small $0000$};
        \node[inner sep=1] (12) at (3,1){\small $1100$};
        \node[inner sep=1] (13) at (3,0){\small $1010$};
        \node[inner sep=1] (1) at (1.5,0){\small $1000$};
        \node[inner sep=1] (3) at (0,0){\small $0010$};

        \node[inner sep=1] (1234) at (0,-0.5){\small $1111$};
        \node[inner sep=1] (123) at (3,-0.5){\small $1110$};
        \node[inner sep=1] (124) at (3,-1.5){\small $1101$};
        \node[inner sep=1] (134) at (1.5,-1.5){\small $1011$};
        \node[inner sep=1] (234) at (0,-1.5){\small $0111$};

        \node[inner sep=1] (23) at (0,-2){\small $0110$};
        \node[inner sep=1] (14) at (3,-2){\small $1001$};
        \node[inner sep=1] (24) at (3,-3){\small $0101$};
        \node[inner sep=1] (34) at (1.5,-3){\small $0011$};
        \node[inner sep=1] (4) at (0,-3){\small $0001$};
        
        \path[thick,->]
   	  	(2) edge (0)
        (0) edge (12)
        (12) edge (13)
        (13) edge (1)
        (1) edge (3)
        (3) edge (0)
   	  	(1234) edge (123)
        (123) edge (124)
        (124) edge (134)
        (134) edge (234)
        (234) edge (1234)
   	  	(23) edge (14)
        (14) edge (24)
        (24) edge (34)
        (34) edge (4)
        (4) edge (23)
        ;
       \end{tikzpicture}
    \end{array}
    &&
    %A
    \begin{array}{c}
    \begin{tikzpicture}
    	\draw[white] (-0.3,-0.3) rectangle ({\x+\z+0.3},{\x+0.3});
        \node[inner sep=1,color=blue] (2) at (0,\x){\small $0100$};
        \node[inner sep=1,color=blue] (0) at (0,0){\small $0000$};
        \node[inner sep=1] (12) at (\x,\x){\small $1100$};
        \node[inner sep=1,color=blue] (13) at ({\x+\z},{0+\zz}){\small $1010$};
        \node[inner sep=1,color=blue] (1) at (\x,0){\small $1000$};
        \node[inner sep=1] (3) at (\z,\zz){\small $0010$};
        \node[inner sep=1] (123) at ({\x+\z},{\x+\zz}){\small $1110$};
        \node[inner sep=1] (23) at (\z,\x+\zz){\small $0110$};
        
        \node[inner sep=1] (24) at (3*\x,\x){\small $0101$};
        \node[inner sep=1] (4) at (3*\x,0){\small $0001$};
        \node[inner sep=1] (124) at (3*\x+\x,\x){\small $1101$};
        \node[inner sep=1] (134) at ({3*\x+\x+\z},{0+\zz}){\small $1011$};
        \node[inner sep=1] (14) at (3*\x+\x,0){\small $1001$};
        \node[inner sep=1] (34) at (3*\x+\z,\zz){\small $0011$};
        \node[inner sep=1] (1234) at ({3*\x+\x+\z},{\x+\zz}){\small $1111$};
        \node[inner sep=1] (234) at (3*\x+\z,\x+\zz){\small $0111$};
        
        \path[thick,->]
        (0) edge[<->,color=blue] (1)
        (0) edge[<->,color=blue] (2)
        (1) edge[<->,color=blue] (13)
        (3) edge (0)
		(4) edge[bend left=20] (0)
		(4) edge (24)
		(4) edge[<->] (34)
        (12) edge (1)
        (12) edge[<->] (123)
        (14) edge (4)
        (14) edge (124)
        (23) edge (123)
        (23) edge (3)
        (23) edge (2)
        (23) edge[bend left=20] (234)
        (24) edge (4)
        (24) edge (234)
        (123) edge[<->,bend left=20] (1234)
        (124) edge (1234)
        (124) edge[bend right=20] (12)
        (134) edge (1234)
        (134) edge (34)
        (234) edge (1234)
        ;
       \end{tikzpicture}
    \end{array}
	\\
	\\
	%2C_5+C_6
    \begin{array}{c}
    \begin{tikzpicture}
    	\draw[white] (-0.5,-0.5) rectangle (3.5,1.5);
        \node[inner sep=1] (2) at (0,1){\small $0100$};
        \node[inner sep=1] (0) at (1.5,1){\small $0000$};
        \node[inner sep=1] (12) at (3,1){\small $1100$};
        \node[inner sep=1] (13) at (3,0){\small $1010$};
        \node[inner sep=1] (1) at (1.5,0){\small $1000$};
        \node[inner sep=1] (3) at (0,0){\small $0010$};

        \node[inner sep=1] (1234) at (0,-0.5){\small $1111$};
        \node[inner sep=1] (123) at (3,-0.5){\small $1110$};
        \node[inner sep=1] (124) at (3,-1.5){\small $1101$};
        \node[inner sep=1] (134) at (1.5,-1.5){\small $1011$};
        \node[inner sep=1] (234) at (0,-1.5){\small $0111$};

        \node[inner sep=1] (23) at (0,-2){\small $0110$};
        \node[inner sep=1] (14) at (3,-2){\small $1001$};
        \node[inner sep=1] (24) at (3,-3){\small $0101$};
        \node[inner sep=1] (34) at (1.5,-3){\small $0011$};
        \node[inner sep=1] (4) at (0,-3){\small $0001$};
        
        \path[thick,->]
   	  	(2) edge (0)
        (0) edge (12)
        (12) edge (13)
        (13) edge (1)
        (1) edge (3)
        (3) edge (2)
   	  	(1234) edge (123)
        (123) edge (124)
        (124) edge (134)
        (134) edge (234)
        (234) edge (1234)
   	  	(23) edge (14)
        (14) edge (24)
        (24) edge (34)
        (34) edge (4)
        (4) edge (23)
        ;
       \end{tikzpicture}
    \end{array}
    &&
    %A
    \begin{array}{c}
    \begin{tikzpicture}
    	\draw[white] (-0.3,-0.3) rectangle ({\x+\z+0.3},{\x+0.3});
        \node[inner sep=1,color=blue] (2) at (0,\x){\small $01000$};
        \node[inner sep=1,color=blue] (0) at (0,0){\small $0000$};
        \node[inner sep=1] (12) at (\x,\x){\small $1100$};
        \node[inner sep=1,color=blue] (13) at ({\x+\z},{0+\zz}){\small $1010$};
        \node[inner sep=1,color=blue] (1) at (\x,0){\small $1000$};
        \node[inner sep=1] (3) at (\z,\zz){\small $0010$};
        \node[inner sep=1] (123) at ({\x+\z},{\x+\zz}){\small $1110$};
        \node[inner sep=1] (23) at (\z,\x+\zz){\small $0110$};
        
        \node[inner sep=1] (24) at (3*\x,\x){\small $0101$};
        \node[inner sep=1] (4) at (3*\x,0){\small $0001$};
        \node[inner sep=1] (124) at (3*\x+\x,\x){\small $1101$};
        \node[inner sep=1] (134) at ({3*\x+\x+\z},{0+\zz}){\small $1011$};
        \node[inner sep=1] (14) at (3*\x+\x,0){\small $1001$};
        \node[inner sep=1] (34) at (3*\x+\z,\zz){\small $0011$};
        \node[inner sep=1] (1234) at ({3*\x+\x+\z},{\x+\zz}){\small $1111$};
        \node[inner sep=1] (234) at (3*\x+\z,\x+\zz){\small $0111$};
        
        \path[thick,->]
        (0) edge[<->,color=blue] (1)
        (0) edge[<->,color=blue] (2)
        (1) edge[<->,color=blue] (13)
        (3) edge (0)
        (3) edge[<->] (23)
		(4) edge[bend left=20] (0)
		(4) edge (24)
		(4) edge[<->] (34)
        (12) edge (1)
        (12) edge[<->] (123)
        (14) edge (4)
        (14) edge (124)
        (23) edge (123)
        (23) edge (2)
        (23) edge[bend left=20] (234)
        (24) edge (4)
        (24) edge (234)
        (123) edge[<->,bend left=20] (1234)
        (124) edge (1234)
        (124) edge[bend right=20] (12)
        (134) edge (1234)
        (134) edge (34)
        (234) edge (1234)
        ;
       \end{tikzpicture}
    \end{array}
	\\
	\\
	%C_4 + 2C_6
    \begin{array}{c}
    \begin{tikzpicture}
    	\draw[white] (-0.5,-0.5) rectangle (3.5,1.5);
        \node[inner sep=1] (2) at (0,0){\small $0100$};
        \node[inner sep=1] (0) at (0,1){\small $0000$};
        \node[inner sep=1] (1) at (1.5,1){\small $1000$};
        \node[inner sep=1] (12) at (1.5,0){\small $1100$};

        \node[inner sep=1] (1234) at (0,-0.5){\small $1111$};
        \node[inner sep=1] (123) at (1.5,-0.5){\small $1110$};
        \node[inner sep=1] (124) at (3,-0.5){\small $1101$};
        \node[inner sep=1] (134) at (3,-1.5){\small $1011$};
        \node[inner sep=1] (234) at (1.5,-1.5){\small $0111$};
        \node[inner sep=1] (13) at (0,-1.5){\small $1010$};

        \node[inner sep=1] (14) at (0,-2){\small $1001$};
        \node[inner sep=1] (23) at (1.5,-2){\small $0110$};
        \node[inner sep=1] (24) at (3,-2){\small $0101$};
        \node[inner sep=1] (34) at (3,-3){\small $0011$};
        \node[inner sep=1] (4) at (1.5,-3){\small $0001$};
        \node[inner sep=1] (3) at (0,-3){\small $0010$};
        
        \path[thick,->]
   	  	(2) edge (0)
        (0) edge (1)
        (1) edge (12)
        (12) edge (2)
   	  	(1234) edge (123)
        (123) edge (124)
        (124) edge (134)
        (134) edge (234)
        (234) edge (13)
        (13) edge (1234)
        
   	  	(14) edge (23)
        (23) edge (24)
        (24) edge (34)
        (34) edge (4)
        (4) edge (3)
        (3) edge (14)
        ;
       \end{tikzpicture}
    \end{array}
    &&
    %A
    \begin{array}{c}
    \begin{tikzpicture}
    	\draw[white] (-0.3,-0.3) rectangle ({\x+\z+0.3},{\x+0.3});
        \node[inner sep=1,color=blue] (2) at (0,\x){\small $0100$};
        \node[inner sep=1,color=blue] (0) at (0,0){\small $1111$};
        \node[inner sep=1,color=blue] (12) at (\x,\x){\small $1100$};
        \node[inner sep=1,color=blue] (1) at (\x,0){\small $1000$};
        \node[inner sep=1] (13) at ({\x+\z},{0+\zz}){\small $1010$};
        \node[inner sep=1] (3) at (\z,\zz){\small $0010$};
        \node[inner sep=1] (123) at ({\x+\z},{\x+\zz}){\small $1110$};
        \node[inner sep=1] (23) at (\z,\x+\zz){\small $0110$};
        
        \node[inner sep=1] (24) at (3*\x,\x){\small $0101$};
        \node[inner sep=1] (4) at (3*\x,0){\small $0001$};
        \node[inner sep=1] (124) at (3*\x+\x,\x){\small $1101$};
        \node[inner sep=1] (134) at ({3*\x+\x+\z},{0+\zz}){\small $1011$};
        \node[inner sep=1] (14) at (3*\x+\x,0){\small $1001$};
        \node[inner sep=1] (34) at (3*\x+\z,\zz){\small $0011$};
        \node[inner sep=1] (1234) at ({3*\x+\x+\z},{\x+\zz}){\small $1111$};
        \node[inner sep=1] (234) at (3*\x+\z,\x+\zz){\small $0111$};
        
        \path[thick,->]
   	  	(123) edge[<->,bend left=20] (1234)
   	  	(123) edge (12)
   	  	(124) edge (1234)
   	  	(134) edge (1234)
   	  	(234) edge (1234)
   	  	(234) edge[<->,bend right=20] (23)
   	  	(234) edge (34)
   	  	(13) edge (123)
   	  	(13) edge[bend left=20] (134)
   	  	(14) edge[bend left=20] (1)
   	  	(14) edge (4)
   	  	(14) edge[<->] (124)
   	  	(14) edge[<->] (134)
   	  	(23) edge (2)
   	  	(24) edge (4)
   	  	(24) edge (234)
   	  	(4) edge[bend left=20] (0)
   		(4) edge (34)
        (3) edge (0)
        (3) edge (13)
        (3) edge[<->,bend right=20] (34)
		(0) edge[color=blue] (1)
        (1) edge[color=blue] (12)
   	  	(12) edge[color=blue] (2)
   	  	(2) edge[color=blue] (0)
        ;
       \end{tikzpicture}
    \end{array}
	\\
	\\
	%4C_4
    \begin{array}{c}
    \begin{tikzpicture}
    	\draw[white] (-0.5,-0.5) rectangle (3.5,1.5);
        \node[inner sep=1] (2) at (0,0){\small $0100$};
        \node[inner sep=1] (0) at (0,1){\small $0000$};
        \node[inner sep=1] (1) at (2,1){\small $1000$};
        \node[inner sep=1] (12) at (2,0){\small $1100$};

        \node[inner sep=1] (124) at (3,0){\small $1101$};
        \node[inner sep=1] (1234) at (3,1){\small $1111$};
        \node[inner sep=1] (234) at (5,1){\small $0111$};
        \node[inner sep=1] (134) at (5,0){\small $1011$};

        \node[inner sep=1] (24) at (0,-1.5){\small $0101$};
        \node[inner sep=1] (123) at (0,-0.5){\small $1110$};
        \node[inner sep=1] (34) at (2,-0.5){\small $0011$};
        \node[inner sep=1] (23) at (2,-1.5){\small $0110$};

        \node[inner sep=1] (3) at (3,-1.5){\small $0010$};
        \node[inner sep=1] (14) at (3,-0.5){\small $1001$};
        \node[inner sep=1] (13) at (5,-0.5){\small $1010$};
        \node[inner sep=1] (4) at (5,-1.5){\small $0001$};

        \path[thick,->]
   	  	(2) edge (0)
        (0) edge (1)
        (1) edge (12)
        (12) edge (2)

   	  	(124) edge (1234)
        (1234) edge (234)
        (234) edge (134)
        (134) edge (124)
        
   	  	(123) edge (34)
        (34) edge (23)
        (23) edge (24)
        (24) edge (123)

   	  	(3) edge (14)
        (14) edge (13)
        (13) edge (4)
        (4) edge (3)        
        ;
       \end{tikzpicture}
    \end{array}
    &&
    %A
    \begin{array}{c}
    \begin{tikzpicture}
    	\draw[white] (-0.3,-0.3) rectangle ({\x+\z+0.3},{\x+0.3});
        \node[inner sep=1,color=blue] (2) at (0,\x){\small $0100$};
        \node[inner sep=1,color=blue] (0) at (0,0){\small $0000$};
        \node[inner sep=1,color=blue] (12) at (\x,\x){\small $1100$};
        \node[inner sep=1] (13) at ({\x+\z},{0+\zz}){\small $1010$};
        \node[inner sep=1,color=blue] (1) at (\x,0){\small $1000$};
        \node[inner sep=1] (3) at (\z,\zz){\small $0010$};
        \node[inner sep=1] (123) at ({\x+\z},{\x+\zz}){\small $1110$};
        \node[inner sep=1] (23) at (\z,\x+\zz){\small $0110$};
        
        \node[inner sep=1] (24) at (3*\x,\x){\small $0101$};
        \node[inner sep=1] (4) at (3*\x,0){\small $0001$};
        \node[inner sep=1] (124) at (3*\x+\x,\x){\small $1101$};
        \node[inner sep=1] (134) at ({3*\x+\x+\z},{0+\zz}){\small $1011$};
        \node[inner sep=1] (14) at (3*\x+\x,0){\small $1001$};
        \node[inner sep=1] (34) at (3*\x+\z,\zz){\small $0011$};
        \node[inner sep=1] (1234) at ({3*\x+\x+\z},{\x+\zz}){\small $1111$};
        \node[inner sep=1] (234) at (3*\x+\z,\x+\zz){\small $0111$};
        
        \path[thick,->]
   	  	(234) edge[<->] (1234)
   	  	(134) edge (1234) 
   	  	(134) edge[<->] (14)   	  	
   	  	(124) edge (1234)
   	  	(123) edge (23)
   	  	(123) edge (13)
   	  	(123) edge[bend left=20] (1234)
   	  	(34) edge[<->] (234)
   	  	(34) edge[<->,bend left=20] (3)
   	  	(23) edge[bend left=20] (234)
   	  	(23) edge (2)
   	  	(24) edge (124)
   	  	(24) edge (234)
   	  	(24) edge[bend right=20] (2)
   	  	(14) edge[bend left=20] (1)
   	  	(13) edge[bend left=20] (134)
   	  	(13) edge (3)
   	  	(13) edge (1)
   	  	(4) edge[bend left=20] (0)
   	  	(4) edge (34)
   	  	(3) edge (13)
   	  	(3) edge (0)
		(0) edge[color=blue] (1)
        (1) edge[color=blue] (12)
   	  	(12) edge[color=blue] (2)
   	  	(2) edge[color=blue] (0)
        ;
       \end{tikzpicture}
    \end{array}
\end{array}
\]
{\caption{\label{fig:n4}Functions in $F(4)$ isomorphic to $\P_5+2C_5$, $2C_5+C_6$, $C_4+2C_6$ and $4C_4$ (left), with their asynchronous graphs (right); each contains a unique attractor $A$ of size $4$ (in blue) and an almost decreasing path from every configuration to $A$.}}
\end{figure}

%%%%%%%%%%%%%%%%%%%%%%%%%
\bibliographystyle{plain}
\bibliography{BIB}

@article{azpeitia2024bridging,
  title={Bridging abstract dialectical argumentation and Boolean gene regulation},
  author={Azpeitia, Eugenio and Guti{\'e}rrez, Stan Mu{\~n}oz and Rosenblueth, David A and Zapata, Octavio},
  journal={arXiv preprint arXiv:2407.06106},
  year={2024}
}

@article{ehrenfeucht2007reaction,
  title={Reaction systems},
  author={Ehrenfeucht, Andrzej and Rozenberg, Grzegorz},
  journal={Fundamenta informaticae},
  volume={75},
  number={1-4},
  pages={263--280},
  year={2007},
  publisher={IOS Press}
}

@article{HMM13,
  title={Attractors in Boolean networks: a tutorial},
  author={Hopfensitz, Martin and M{\"u}ssel, Christoph and Maucher, Markus and Kestler, Hans A},
  journal={Computational Statistics},
  volume={28},
  pages={19--36},
  year={2013},
  publisher={Springer}
}

@article{GDX08,
  title={Synchronous versus asynchronous modeling of gene regulatory networks},
  author={Garg, Abhishek and Di Cara, Alessandro and Xenarios, Ioannis and Mendoza, Luis and De Micheli, Giovanni},
  journal={Bioinformatics},
  volume={24},
  number={17},
  pages={1917--1925},
  year={2008},
  publisher={Oxford University Press}
}

@article{ZYL13,
  title={An efficient algorithm for computing attractors of synchronous and asynchronous Boolean networks},
  author={Zheng, Desheng and Yang, Guowu and Li, Xiaoyu and Wang, Zhicai and Liu, Feng and He, Lei},
  journal={PloS one},
  volume={8},
  number={4},
  pages={e60593},
  year={2013},
  publisher={Public Library of Science San Francisco, USA}
}

@article{RT23,
  title={Attractor separation and signed cycles in asynchronous Boolean networks},
  author={Richard, Adrien and Tonello, Elisa},
  journal={Theoretical Computer Science},
  pages={113706},
  year={2023},
  publisher={Elsevier}
}

@article{DS97,
  title={Images and preimages in random mappings},
  author={Drmota, Michael and Soria, Michele},
  journal={SIAM Journal on Discrete Mathematics},
  volume={10},
  number={2},
  pages={246--269},
  year={1997},
  publisher={SIAM}
}

@article{FMR21,
  title={Isometries of the hypercube: A tool for Boolean regulatory networks analysis},
  author={Fabre-Monplaisir, Jean and Moss{\'e}, Brigitte and Remy, Elisabeth},
  journal={Physica D: Nonlinear Phenomena},
  volume={424},
  pages={132831},
  year={2021},
  publisher={Elsevier}
}

@article{H00,
  title={The Automorphism Group of a Hypercube.},
  author={Harary, Frank},
  journal={J. Univers. Comput. Sci.},
  volume={6},
  number={1},
  pages={136--138},
  year={2000},
  publisher={Citeseer}
}

@Article{NS17,
author="Noual, M. and Sen{\'e}, S.",
title="Synchronism versus asynchronism in monotonic Boolean automata networks",
journal="Natural Computing",
year="2017",
month="Jan",
day="05",
issn="1572-9796",
doi="10.1007/s11047-016-9608-8",
url="https://doi.org/10.1007/s11047-016-9608-8"
}

@inproceedings{FO89,
  title={Random mapping statistics},
  author={Flajolet, Philippe and Odlyzko, Andrew M},
  booktitle={Workshop on the Theory and Application of of Cryptographic Techniques},
  pages={329-354},
  year={1989},
  organization={Springer}
}

@article{Ru2017,
title = "Negative local feedbacks in Boolean networks ",
journal = "Discrete Applied Mathematics ",
volume = "221",
number = "",
pages = "1-17",
year = "2017",
note = "",
issn = "0166-218X",
doi = "http://dx.doi.org/10.1016/j.dam.2017.01.001",
url = "http://www.sciencedirect.com/science/article/pii/S0166218X17300069",
author = "Ruet, P."
}

@article{GN12,
  title={Disjunctive networks and update schedules},
  author={Goles, Eric and Noual, Mathilde},
  journal={Advances in Applied Mathematics},
  volume={48},
  number={5},
  pages={646-662},
  year={2012},
  publisher={Elsevier}
}

@article{GT83,
author = {Goles, E. and Tchuente, M.},
title ={Iterative behaviour of generalized majority functions},
journal ={Mathematical Social Sciences},
volume ={4},
pages = {197-204},
year = {1982}
}

@article{PS83,
author = {Poljak, S. and Sura, M.},
title ={On periodical behaviour in societies with symmetric influences},
journal ={Combinatorica},
volume ={3},
pages = {119-121},
year = {1982}
}

@article{H82,
author = {Hopfield, J.},
title ={Neural networks and physical systems with emergent collective computational abilities},
journal ={Proc. Nat. Acad. Sc. U.S.A.},
volume ={79},
pages = {2554-2558},
year = {1982}
}

@article{G85,
author = {Goles, E.},
title ={Dynamics of positive automata networks},
journal ={Theoretical Computer Science},
volume ={41},
pages = {19-32},
year = {1985}
}

@article{MP43,
author = {Mac Culloch, W.~S. and Pitts, W.~S.},
title ={A logical calculus of the ideas immanent in nervous activity},
journal ={Bull. Math Bio. Phys.},
volume ={5},
pages = {113-115},
year = {1943}
}

@ARTICLE{R09,
  author = {Richard, A.},
  title = {Positive circuits and maximal number of fixed points in discrete
	dynamical systems},
  journal = {Discrete Applied Mathematics},
  year = {2009},
  volume = {157},
  pages = {3281-3288},
  number = {15},
  doi = {10.1016/j.dam.2009.06.017},
  issn = {0166-218X}
}

@ARTICLE{T73,
  author = {Thomas, R.},
  title = {{B}oolean formalization of genetic control circuits},
  journal = {Journal of Theoretical Biology},
  year = {1973},
  volume = {42},
  pages = {563-585},
  number = {3},
  doi = {10.1016/0022-5193(73)90247-6},
  issn = {0022-5193}
}

@BOOK{TA90,
  title = {Biological Feedback},
  publisher = {CRC Press},
  year = {1990},
  author = {Thomas, R. and d'Ari, R.}
}

@BOOK{R95,
  title = {Les syst\`emes dynamiques discrets},
  publisher = {Springer},
  year = {1995},
  author = {Robert, F.},
  volume = {19},
  series = {Math\'ematiques et Applications}
}

@ARTICLE{K69,
  author = {Kauffman, S. A.},
  title = {Metabolic stability and epigenesis in randomly connected nets},
  journal = {Journal of Theoretical Biology},
  year = {1969},
  volume = {22},
  pages = {437-467},
  doi = {10.1016/0022-5193(69)90015-0}
}

@ARTICLE{GS10,
  author = {Goles, E. and Salinas, L.},
  title = {Sequential operator for filtering cycles in {B}oolean networks},
  journal = {Advances in Applied Mathematics},
  year = {2010},
  volume = {45},
  pages = {346-358},
  number = {3}
}

@ARTICLE{GS08,
  author = {Goles, E. and Salinas, L.},
  title = {Comparison between parallel and serial dynamics of {B}oolean networks},
  journal = {Theoretical Computer Science},
  year = {2008},
  volume = {396},
  pages = {247-253},
  number = {1-3}
}

@book{GM90,
 author = {Goles, E. and Mart{\'\i}nez, S.},
 title = {Neural and Automata Networks: Dynamical Behavior and Applications},
 year = {1990},
 publisher = {Kluwer Academic Publishers}
}

@ARTICLE{J02,
    author = {Hidde De Jong},
    title = {Modeling and simulation of genetic regulatory systems: A literature review},
    journal = {Journal of Computational Biology},
    year = {2002},
    volume = {9},
    pages = {67-103}
}

@ARTICLE{TK01,
  author = {Thomas, R. and Kaufman, M.},
  title = {Multistationarity, the basis of cell differentiation and memory.
	{II}. {L}ogical analysis of regulatory networks in terms of feedback
	circuits},
  journal = {Chaos: An Interdisciplinary Journal of Nonlinear Science},
  year = {2001},
  volume = {11},
  pages = {180-195},
  number = {1},
  doi = {10.1063/1.1349893},
  publisher = {AIP}
}
%%%%%%%%%%%%%%%%%%%%%%%%%

\end{document}